\documentclass[onecolumn]{IEEEtran}

\usepackage{booktabs}   
\usepackage{subcaption} 
\usepackage{cite}
\usepackage{hyperref}
\usepackage{graphicx}
\usepackage{wrapfig}
\usepackage{paralist}
\usepackage{stackrel}
\usepackage{mathpartir}
\usepackage{algpseudocode}
\usepackage{algorithm}
\usepackage{graphicx} 
\usepackage{xspace, ifthen}
\usepackage{amsmath, amsthm, amssymb}
\usepackage{xcolor}
\usepackage[justification=centering]{caption}
\usepackage{framed}
\usepackage{chngcntr}
\usepackage{verbatim}
\usepackage{pifont}
\usepackage{enumerate}
\usepackage{listings}
\makeatletter
\newsavebox{\@brx}
\newcommand{\llangle}[1][]{\savebox{\@brx}{\(\m@th{#1\langle}\)}%
  \mathopen{\copy\@brx\kern-0.5\wd\@brx\usebox{\@brx}}}
\newcommand{\rrangle}[1][]{\savebox{\@brx}{\(\m@th{#1\rangle}\)}%
  \mathclose{\copy\@brx\kern-0.5\wd\@brx\usebox{\@brx}}}
\makeatother

\newcommand{\Paragraph}[1]{\vspace{5pt}\noindent{\bf #1:}}
\newcommand{\m}[1]{\mathsf{#1}}
\newcommand{\mi}[1]{\mathit{#1}}
\newcommand{\bnfdef}{::=}
\newcommand{\bnfalt}{\,|\,}
\newcommand{\rulename}[1]{\textsc{#1}}

\newcommand{\etrue}{\m{true}}
\newcommand{\efalse}{\m{false}}

\newcommand{\eif}{\mathsf{if}}
\newcommand{\ethen}{\mathsf{then}}
\newcommand{\eelse}{\mathsf{else}}
\newcommand{\ewhile}{\mathsf{while}}
\newcommand{\edo}{\mathsf{do}}

\newcommand{\epair}[2]{\langle #1\bnfalt #2\rangle}
\newcommand{\eupdate}{\mathsf{upd}}
\newcommand{\eread}{\mathsf{rd}}
\newcommand{\eskip}{\mathsf{skip}}
\newcommand{\eoutput}{\mathsf{output}}
\newcommand{\eabort}{\mathsf{abort}}

\newcommand{\bop}{\mathrel{\mathsf{bop}}}
\newcommand{\proj}[2]{\lfloor{#1}\rfloor_{#2}}
\newcommand{\trace}{\mathbb{T}}

\newcommand{\wg}{\mathbb{WHILE}^{\mathsf{G}}}
\newcommand{\wge}{\mathbb{WHILE}^{\mathsf{G}}_{\mathsf{Evd}}}

\newcommand{\tpint}{\mathsf{int}}
\newcommand{\tbool}{\mathsf{bool}}

\newcommand{\lab}{\ell}
\newcommand{\pc}{\mathit{pc}}

\newcommand{\labof}{\textit{intvl}\,}
\newcommand{\glab}{g}
\newcommand{\gpc}{g_\pc}

\newcommand{\labless}{\preccurlyeq}
\newcommand{\nlabless}{\not\preccurlyeq}
\newcommand{\labjoin}{\mathrel{\curlyvee}}
\newcommand{\labmeet}{\mathrel{\curlywedge}}

\newcommand{\cjoin}{\mathrel{\curlyvee{\!\!}_c}}

\newcommand{\clabless}{\mathrel{\preccurlyeq{\!}_c}}
\newcommand{\csubtp}{\mathrel{\leq{\!}_c}}

\newcommand{\gsubtp}{\mathrel{\sqsubseteq}}

\newcommand{\eundef}{\mathtt{undef}}

\newcommand{\wtsetof}{\textit{WtSet}}
\newcommand{\refineof}{\textit{refine}}

\newcommand{\rflof}[1]{\textit{rfL}_{#1}}
\newcommand{\reflvof}[1]{\textit{refineLB}_{#1}}

\newcommand{\updval}{\mathsf{updL}}
\newcommand{\newlab}{\textit{restrictLB}}

\newcommand{\stepsto}{\longrightarrow}
\newcommand{\evalsto}{\Downarrow}
\newcommand{\sepidx}[1]{\mathrel{/_{#1}}}

\newcommand{\ee}{\mathcal{E}}
\newcommand{\de}{\mathcal{D}}

\newtheorem{thm}{Theorem}
\newtheorem{lem}{Lemma}

\newenvironment{proofsketch}{\vspace{-1pt}\textsc{Proof (sketch). }\hspace*{0.25em}}{ \hspace*{\fill} \qed}

\definecolor{mypurple}{rgb}{.41,.28,.63}

\makeatletter 
\def\arcr{\@arraycr}
\makeatother

\newcommand{\iffull}[1]{#1}
\newcommand{\ifconf}[1]{}

\begin{document}

\title{First-order Gradual Information Flow Types with Gradual Guarantees}

\author{
\IEEEauthorblockN{Abhishek Bichhawat, McKenna McCall and Limin Jia} \\
\IEEEauthorblockA{\textit{Carnegie Mellon University}, 
Pittsburgh, USA \\
\{abichhaw, mckennak, liminjia\}@andrew.cmu.edu}
}

\maketitle

\thispagestyle{plain}
\pagestyle{plain}

\begin{abstract}
  Information flow type systems enforce the security property of
  noninterference by detecting unauthorized data flows at
  compile-time. However, they require precise type annotations,
  making them difficult to use in practice as much of the legacy
  infrastructure is written in untyped or dynamically-typed languages.
  Gradual typing seamlessly integrates static and
  dynamic typing, providing the best of both approaches, and has been applied
  to information flow control, where information flow monitors are
  derived from gradual security types. Prior work on gradual
  information flow typing uncovered tensions between noninterference and
  the dynamic gradual guarantee---the property that less precise
  security type annotations in a program should not cause more runtime
  errors. 

  This paper re-examines the connection between gradual
  information flow types and information flow monitors to 
  identify the root cause of the tension between the gradual
  guarantees and noninterference. We develop runtime semantics for a
  simple imperative language with gradual information flow types that
  provides both noninterference and gradual guarantees. We
  leverage a proof technique developed for FlowML and reduce
  noninterference proofs to preservation proofs.


\end{abstract}
\section{Introduction}
\label{sec:intro}

Information flow type systems combine types and security labels to
ensure that well-typed programs do not leak secrets to attackers at
compile-time~\cite{volpano1996}. However, purely statically-typed
languages face significant adoption challenges. Most programmers
are unfamiliar with and may be unwilling to use complex information flow type
systems. Moreover, much of the legacy infrastructure is written in untyped and
dynamically-typed languages without precise security type annotations.  

Gradual typing is one promising technique to address these
challenges~\cite{siek2006}; it aims to seamlessly integrate
statically-typed programs with dynamically-typed programs. At a
high-level, gradual type systems introduce a dynamic type, often
written as $?$, to accommodate untyped portions of the program. The
type system allows any program to be typed under $?$. The type
  system enforces type safety on statically typed parts and the
runtime semantics of gradual type systems monitor the interactions
between parts typed as $?$ and statically typed parts to ensure type
preservation.

Gradual typing has been applied to information flow
types~\cite{disney2011,fennell2013,fennell2016,toro2018,gifc-lics2020}, 
where certain
expressions have a dynamic security label $?$ (and are typed as, e.g.,
$\m{int}^?$), which is determined at runtime.  
Information flow monitors are then derived from the runtime semantics
of gradual information flow types.
Earlier adoptions of gradual typing to information flow types
  have developed ad hoc approaches to treat ``gradual types''.  For
  instance, Disney and Flanagan did not include a dynamic security
  label~\cite{disney2011}. Instead, the programmer would insert type casts, which
  are checked at runtime, and are used to ``gradually'' make the
  programs more secure. Later, ML-GS~\cite{fennell2013} and
  LJGS~\cite{fennell2016} included the dynamic security label $?$ and a
  runtime monitor that performed checks on the dynamically typed parts
  of the program. Programmers still need to write annotations and casts
  in ML-GS and label constraints in LJGS. The dynamic label is
  instantiated as a single security label at run time in both ML-GS
  and LJGS.

At the same time, interests in the formal foundations of gradual
  types grew significantly. Formal properties related to gradual
typing such as the \emph{gradual guarantee}~\cite{siek2015} were
introduced, which
says that loosening policies should not cause more type errors or
runtime failures. Roughly, the
gradual guarantee ensures that programs which type-check and run to
completion with precise type annotations will also type-check and
run to completion with less precise types, i.e., $?$.
This property ensures that programmers are not punished for
  not specifying type annotations if the program is safe. Such a
  guarantee is important for information flow type systems, as
  security annotations have been a road block for adoption.
  Without the gradual guarantee, the programmers' burden
  of providing (unnecessary) security annotations is increased.

Garcia et al. developed the abstracting gradual typing (AGT)
  framework which provides a formal interpretation of gradual type
  systems~\cite{garcia2016}. In AGT, a principled interpretation of
  the dynamic type is that it represents the set of all possible types
  that are refined by the monitor to preserve type safety at
  runtime. By that interpretation, the semantics of the dynamic
  information flow label $?$ is the set of all possible labels. Early
  work on gradual information flow typing all instantiate the dynamic
  label as a single label at runtime~\cite{fennell2013,fennell2016}.
  Recent work by Toro et al., $\m{GSL_{Ref}}$, aims to apply the AGT
  framework to information flow types~\cite{toro2018}; however,
it has to give up the dynamic gradual guarantee in favor of
\emph{noninterference}, the key information flow security
property~\cite{goguen1982}, when dealing with mutable references.

In this paper, we re-examine the connection between gradual
information flow types and information flow monitors
(c.f.~\cite{austin2009,austin2010,russo2010}).  We aim to identify the root
cause of the tension between the dynamic gradual guarantee and security
in systems that refine the set of possible labels for dynamically labeled
   programs at runtime. To this end, we focus on a simple imperative language
  with first-order stores, which has been widely used to design
information flow control
systems~\cite{hunt2006,volpano1996,volpano1997,terauchi2005,russo2010,moore2011}.
While simple, this language includes all the features to
  illustrate the problem of refining dynamic labels at runtime.  We
develop runtime semantics for this language with gradual information
flow types that enjoy both noninterference and the dynamic gradual
guarantee.  We draw ideas from abstracting gradual typing, which
advocates deriving runtime semantics for gradual types via the
preservation proof~\cite{garcia2016}.

We observe that as dynamic labels are updated, the
  semantics that only gradually refine the possible security labels
  during program execution \emph{resemble} a naive
  flow-sensitive monitor and therefore inherit the problems with implicit
  leaks of flow-sensitive monitors~\cite{austin2009}. To enforce
  noninterference and remove the implicit leaks caused by insecure
  writes in branches, the runtime semantics needs to take into
  consideration the variable and channel writes in the untaken
  branch~\cite{russo2010}.  Pure guessing which ignores
  information about the untaken branch like $\m{GSL_{Ref}}$'s runtime
  yields rigid semantics that break the gradual guarantee (more in
  Section~\ref{sec:motv}).  The no-sensitive-upgrade
  (NSU) check~\cite{austin2009} also doesn't solve the problem. Instead, a
  ``hybrid'' approach~\cite{moore2011, russo2010} that leverages
  static analysis to obtain the write effects of the untaken branch
  and upgrade relevant references for both branches can be used to
  remove the implicit leaks and provide the gradual guarantees.

We leverage a proof technique developed for FlowML, which reduces
noninterference proofs to preservation proofs~\cite{pottier2002}.
The main idea is to extend the language with \emph{pairs}
of expressions and commands, representing two executions with
different secrets in one program. Noninterference follows from
preservation. This proof technique clearly illustrates the
problem with purely dynamic flow-sensitive monitors and
naturally suggests the hybrid approach~\cite{moore2011, russo2010}. 


To summarize, we study the connection between gradual security types
and information flow monitors and identify the conservative handling
of implicit flows in $\m{GSL_{Ref}}$ as the reason
that it gives up dynamic gradual guarantee in favor of
noninterference. Additionally, we show that the dynamic gradual guarantee
can be recovered by using a hybrid approach that leverages the static
phase to generate a list of variables that are written to in both the
branches. Due to space constraints, we omit detailed definitions and proofs, 
which can be found in the full version of the paper~\cite{full-version}.


\section{Overview of  Information Flow Control}
\label{sec:overview}
In information flow control systems, variables are annotated with
a label from a security lattice, which
have a partial-ordering ($\labless$) and a well-defined join and meet
operation. $\lab_1 \labless \lab_2$ 
means information can flow from $\lab_1$ to $\lab_2$. 
Consider a two-point security lattice with labels
$\{L,H\}$ with $L \labless H$ where $L$ represents public and
$H$ represents secret. A variable $x$
having type $\tpint^H$ contains a sensitive integer value. 

Information flows can be broadly classified as \emph{explicit} or
\emph{implicit}~\cite{denning1977,goguen1982}. Explicit flows arise from
variable assignments. For instance, the statement $x = y +
z$ causes an explicit flow of values from $y$ and $z$ to $x$. Implicit
flows arise from control structures in the program. For example, in
the program $l = \efalse;\; \eif (h) \{l = \etrue;\}$, there is an
implicit flow of information from $h$ to the final value of $l$ (it is 
$\etrue$ iff $h$ is $\etrue$). Implicit flows are handled by
maintaining a $\pc$ (program-context) label, 
which is an upper bound on the labels of all the predicates that have
influenced the control flow thus far. In the example, 
the $\pc$ inside the branch is the label of $h$.  

Information flow control systems aim to prevent leaks through these
flows by either enforcing information flow typing rules and
ruling out insecure programs at compile-time or dynamically monitoring
programs and aborting the execution of insecure programs. In both
systems, assignment to a variable is disallowed if either the $\pc$
label or the join of the label of the operands is not less than or 
equal to the
label of the variable being
assigned~\cite{volpano1996,austin2009}. Thus, in the above examples,
if either the label of $y$ or $z$ is greater than the label of $x$ 
or the label of $h$ is greater than the label of $l$, the assignment
does not type-check or  the execution aborts at runtime. This
guarantees a variant of noninterference, known as
termination-insensitive noninterference~\cite{volpano1996}, which 
we prove for our gradual type system. We assume that an adversary cannot
observe or gain any information if a program's execution diverges or
aborts and can only observe ``public'' outputs by the program. 

We consider a flow-insensitive, fixed-label system in this paper and
prove termination-insensitive noninterference for it.

\section{Gradual Security Typing}
\label{sec:overview-gradual}
Static information flow type systems do not scale up to scenarios
where the security levels of some of the variables are not known at
compile-time, while pure monitoring approaches cannot reject obviously
insecure programs at compile-time. Gradual typing extends the reach of
type-system based analysis by adding an imprecise (or dynamic) label, 
$?$, for variables whose labels are not known at
compile-time. The runtime semantics then ensures that no information
is leaked due to the relaxation of the type-system's handling
of $?$ labels.

\subsection{Imprecise Security Label: Interpretations and Operations}
\label{sec:overview-g}
The label $?$ is not an actual element of the security lattice
and its meaning is not universally agreed upon. Differences will
manifest in the runtime monitoring semantics and proof of
noninterference. For illustration, consider a variable $x$ of type
$\m{int}^?$. Semantically, in the literature $?$ has meant one of the following.  (1)
$x$'s label is dynamic (flow-sensitive) and can change at runtime
(e.g., from $?$ representing $L$ to representing $H$ when $x$ is
assigned a secret). (2) the set of possible labels for $x$ is refined
at runtime and the set in a future state will be a subset of its the
current state. 
Since we build on a flow-insensitive type system, we opt for
the second meaning of $?$.  Our runtime monitor will keep track of the
set of possible labels for $x$. 
Note that a typical flow-sensitive IFC monitor
(e.g.~\cite{russo2010,austin2009}) takes an approach aligned with (1). 

\begin{figure}
\begin{minipage}{.5\linewidth}
  \centering
\begin{lstlisting}[caption={\small Statically typed},
  label=introEg1,xleftmargin=.15\columnwidth]
  $y := \efalse^H$
  if ($x$) then $y := \etrue^H$ @\label{intr1}@ 
\end{lstlisting}
\end{minipage}\hfill
\begin{minipage}{.5\linewidth}
  \centering
\begin{lstlisting}[caption={\small Dynamically typed},
  label=introEg2,xleftmargin=.15\columnwidth]
  $y := \efalse^?$
  if ($x$) then $y := \etrue^?$ @\label{intr2}@ 
\end{lstlisting}
\end{minipage}
\medskip
\noindent
\begin{minipage}{.5\linewidth}
\begin{lstlisting}[caption={\small NSU},label=introEg3,xleftmargin=.15\columnwidth]
  $y := \efalse^?$
  $z := \efalse^L$
  if ($x$) then $y := \etrue^?$ @\label{intr3}@ 
  if ($y$) then $z := \etrue^L$ @\label{intr4}@
  output($L$, $z$)
\end{lstlisting}
\end{minipage}\hfill
\begin{minipage}{.5\linewidth}
\begin{lstlisting}[caption={\small Implicit flows},label=motvEg,xleftmargin=.15\columnwidth]
  $y := \etrue^?$
  $z := \etrue^L$
  if ($x$) then $y := \efalse^?$ @\label{motv1}@ 
  if ($y$) then $z := \efalse^L$ @\label{motv2}@
  output($L$, $z$)@\label{motv3}@
\end{lstlisting}
\end{minipage}
\vspace*{-2mm}
\end{figure}

In the initial program state, $?$ could be interpreted as:
(1) $x$ could contain a secret or be observable to an
adversary. Therefore, we should treat $x$ conservatively as if it is
both secret and public. (2) $?$ indicates indifference; the data $x$
contains in the initial state is of no security value; otherwise, $x$ should have been
given the label $H$. We choose (2) again, as it is a cleaner
interpretation. Note that this is only for the initial state. At
runtime, the monitor maintains enough state to concretely know whether
$x$ contains a secret or not.

\subsection{Gradually Refined Security Policy via Examples}
\renewcommand{\thelstlisting}{\arabic{lstlisting}}
Next, we describe
  how security labels can be gradually refined at runtime.
Consider the program in Listing~\ref{introEg1} and its
variant with dynamic types in Listing~\ref{introEg2}. Suppose the
lattice contains four elements $\{\bot, L, H, \top\}$ such that $\bot
\labless L \labless H \labless \top$. Assume that the initial type of
$x$ is $\tbool^H$
in both examples while the type of $y$ is
$\tbool^H$ in Listing~\ref{introEg1} and $\tbool^?$ in
Listing~\ref{introEg2}\footnote{We will write $x^\ell$ 
  to indicate that variable $x$ has the type $\tau^\ell$.}.  The program in Listing~\ref{introEg1} does
not leak any information. With gradual typing, its variant in
Listing~\ref{introEg2} is also accepted by the type system.
The runtime refines the set of possible labels for $y$ as the
  program runs.
  With $x:\texttt{bool}^H$,  $y$ cannot be $\bot$ or $L$
  as that would result in an implicit flow. Thus, the
  possible labels of $y$ are refined to the set $\{H, \top\}$ when
  $x=\etrue$.
  If, suppose, $x: \texttt{bool}^L$, then $y$ is labeled 
  $\{L, H, \top\}$ after executing line~\ref{intr2}.

Suppose the program in Listing~\ref{introEg2} is extended with another
branch as shown in Listing~\ref{introEg3} with 
$z:\texttt{bool}^L$. 
When $x:\texttt{bool}^L$, the possible labels for $y$ on
line~\ref{intr4} are $\{L, H, \top \}$. This allows the assignment on
line~\ref{intr4} to succeed as the assignment is to a variable that
has a label contained in the set of possible $\pc$ labels.  When $x$
has the type $\texttt{bool}^H$, then the possible labels for $y$ on
line~\ref{intr4} are $\{H, \top \}$. The assignment on
line~\ref{intr4} is aborted as $L$ is not equal or higher than any of the 
possible $\pc$ labels ($\{H, \top\}$).
In this case, the monitor enforces NSU and prevents the implicit leak.

\subsection{Gradual Guarantees} 

Desirable formal properties for gradual type systems
are the \emph{gradual guarantees}, proposed by~\cite{siek2015}.
The gradual guarantees relate programs that differ
only in the precision of the type annotations. They state that
changes that make the annotations of a gradually typed program less precise should not
change the static or dynamic behavior of the program. In other words,
if a program with more precise type annotations is well-typed in the static type
system, and terminates in the runtime semantics,
then the same program with less precise terms is also
well-typed, and terminates, respectively.

For illustration, consider the previous example from
Listing~\ref{introEg1}. Assume that the program is well-typed under a
gradual type system with $x : \texttt{bool}^H$ and $y :
\texttt{bool}^H$ as secret variables. When $x$ is \texttt{true}, the branch on
line~\ref{intr1} is taken and $y$ is assigned the value \texttt{true}
and has the label $H$. When $x$ is \texttt{false}, $y$ remains
\texttt{false}. This program is accepted by the security type system 
and dynamic information flow monitor, and runs to completion in all
possible executions. As per the \emph{static gradual guarantee}, the
program should also be well-typed if, for instance, $y$ had an
imprecise $?$ security level as shown in Listing~\ref{introEg2}. By 
the \emph{dynamic gradual guarantee}, the program in
Listing~\ref{introEg2} should run to completion at runtime, even with
the imprecise label for $y$ in all executions of the program. 

The gradual guarantees are important in the context of
information flow systems to show that the gradual security type system
is strictly more permissive than the static security type system and
the dynamic IFC monitor, while providing the same
guarantees. They mean that programmers need not worry about insufficient
  annotations causing their safe programs to be rejected by the type
  system, or worse, at runtime, which may lead to undesirable
  behavior. Concerning programmers with unnecessary annotations
  defeats the purpose of gradual typing, which is meant to
  alleviate the burden of annotation.

\subsection{Implicit Flows vs. Dynamic Gradual Guarantee}
\label{sec:motv}

The example in Listing~\ref{motvEg} illustrates how gradual
  typing semantics handle implicit flows.
  Assume that $x : \tbool^H$ is a secret variable while the
  security types of $y : \tbool^?$ and $z : \tbool^?$ are unknown at
  compile-time, and the security lattice is $\bot \labless L \labless
  H \labless \top$. Consider two runs of the program with different
  initial values of $x$. When $x$ is $\etrue$, the branch on
  line~\ref{motv1} is taken and $y$ is assigned the value $\efalse$.
  With gradual typing, the labels of $y$ are refined to $\{H, \top\}$.
  As $y$ is $\efalse$, the branch on line~\ref{motv2} is not taken and
  $z$ remains $\etrue$ with the set of labels $\{\bot, L, H,
  \top\}$. As $z$'s value is visible at $L$, the output on
  line~\ref{motv3} succeeds.  When $x$ is $\efalse$, the branch on
  line~\ref{motv1} is not taken and $y$ remains $\etrue$ with the set
  of labels $\{\bot, L, H, \top\}$.  As the $\pc$ on line~\ref{motv2}
  contains the set of labels $\{\bot, L, H, \top\}$, the assignment on
  line~\ref{motv2} succeeds, and $z$ becomes $\efalse$ while the set
  of labels remains $\{\bot, L, H, \top\}$. Again, as $z$'s value is
  visible at $L$, the output on line~\ref{motv3} succeeds. Thus, in
  the two runs of the program, different values of $z$ are output for
  different values of $x$, thereby leaking  $x$ to
  the adversary at level $L$. Here, the NSU mechanism does not apply, as
  the assignment to $y$ on line~\ref{motv1} is merely refining, not
  ``upgrading'', the label of $y$. If $y$'s label
  had been $L$, this program would have been rejected.

This program can be rejected by deploying a special monitoring rule
that preemptively aborts an assignment statement if there is a
possibility that the $\pc$ is not lower than or equal to the
variable's label, as deployed by $\m{GSL_{Ref}}$~\cite{toro2018}. In
  the above example, the assignment on line~\ref{motv1} will be
  aborted, because the $\pc$ is $H$, and $y$'s label could be $\bot$
  or $L$, which might leak information. This ensures noninterference,
  but unfortunately, the extra check does not retain the dynamic gradual
  guarantee. That is, enlarging the set of possible labels for
  the dynamic security type of $y$ 
  will cause the monitor to abort, which contradicts the dynamic
  gradual guarantee.
 
  %
\begin{lstlisting}[caption={\small Secure program violating gradual guarantee},
  label=motvEg2,xleftmargin=.3\columnwidth,
  float=t]
  $y := \etrue^?$
  if ($x$) then $y := \efalse^?$ @\label{motv21}@ 
  output($H$, $y$)@\label{motv23}@
\end{lstlisting}
%

 Consider the program in Listing~\ref{motvEg2} with the same 
security lattice as before such that the variable $x$ is
labeled $H$ and $y$'s label is not specified at compile-time. 
As the $\pc$ on line~\ref{motv21} is $H$ and the possible set of
labels of $y$ on line~\ref{motv21} is $\{\bot, L, H, \top\}$,
the monitor aborts the execution of the program when $x$ is
$\etrue$ to satisfy noninterference.
However, if $y$ was labeled $H$ or $\top$ at compile-time
instead of being $?$, the execution would have 
proceeded and output the value of $y$ to $H$. In other words, 
the larger set of possible labels for $y$ on line~\ref{motv21}
(because of the unknown label) causes the monitor to abort while
the precisely typed version of the program with $y:\tbool^H$
is accepted by the monitor, which violates the dynamic gradual
guarantee.

To tackle this problem, we leverage the static phase of the gradual
type system to determine the set of variables being written to in
different branches and loops, and \emph{refine} their possible security
labels to implement a monitoring strategy that preserves the
dynamic gradual guarantee.  At the branch on
line~\ref{motv21} in Listing~\ref{motvEg2}, we know that $y$ may be
written to inside the branch, therefore, we narrow the possibility
of the labels for $y$ to $\{H, \top\}$ as the first step of executing
the if statement, regardless of whether $x$ is $\etrue$ or $\efalse$.
This is very similar to how hybrid monitors stop implicit
leaks~\cite{russo2010,moore2011,hedin2015,bedford2017}.
We will discuss this further in Section~\ref{sec:monitor-semantics}.

\section{A Language with Gradual Security Types} 
\label{sec:language}

\iffull{
  \begin{figure}[!htbp]
    \centering
    \(
    \begin{array}{lcll}
      \textit{Labels} & \lab& \bnfdef & L \bnfalt \ldots \bnfalt H
      \\
      \textit{Raw values} & u &\bnfdef & n \bnfalt
                                         \etrue \bnfalt \efalse 
      \\ \textit{Values} & v & \bnfdef & u^\glab
      \\
      \textit{Types} & \tau & \bnfdef & \tbool \bnfalt \tpint
      \\
      \textit{Gradual labels} & g& \bnfdef & ? \bnfalt \lab
      \\
      \textit{Gradual types} & U & \bnfdef & \tau^\glab
      \\
      \textit{Typing Context} & \Gamma & \bnfdef & \cdot \bnfalt \Gamma, x: U
      \\\\
      \textit{Expressions} & e & \bnfdef &
                                           x \bnfalt v \bnfalt e_1 \bop e_2 \bnfalt e :: \tau^\glab
      \\
      \textit{Commands} & c & \bnfdef &
                                        \eskip \bnfalt c_1;c_2
                                        \bnfalt x\, := \,e  \bnfalt \eoutput (\lab, e) 
      \\ & & \bnfalt & \eif\ e\, \ethen\, c_1\, \eelse \, c_2
               \bnfalt \ewhile\ e\; \edo\; c
    \end{array}
    \)
    \caption{Syntax for the language $\wg$}
    \label{fig:app-syntax}
  \end{figure}

\begin{figure}
  \small
    \begin{mathpar}
      \inferrule*{\lab_1\labless\lab_2 }{ 
        \lab_1\clabless\lab_2
      }
      \and
      \inferrule*{ }{ 
        ?\clabless\glab
      }
      \and
      \inferrule*{ }{ 
        \glab\clabless ?
      }
      \and
      \inferrule*{\glab_1 \clabless \glab_2}{ 
        \tau^{\glab_1}\csubtp\tau^{\glab_2}
      }
    \end{mathpar}
    \begin{mathpar}
      \inferrule*{ }{ 
        \lab_1\cjoin\lab_2 = 
        \lab_1\labjoin\lab_2 
      }
      \and
      \inferrule*{
        \glab \neq \top
      }{ 
        ?\cjoin\glab = \,?
      }
      \and
      \inferrule*{
        \glab \neq \top
      }{ 
        \glab\cjoin \, ? = \,?
      }
      \and
      \inferrule*{
      }{ 
        ?\cjoin\top = \top
      }
      \and
      \inferrule*{
      }{ 
        \top\cjoin \, ? = \top
      }
    \end{mathpar}
    \caption{Operations on gradual labels and types}
    \label{fig:clabelop}
          \end{figure}
          
  \begin{figure}[h]
  
    \flushleft
    \noindent\framebox{$\Gamma \vdash e : U$}
    \begin{mathpar}
      \inferrule*[right=Bool]{
      }{ 
        \Gamma  \vdash b^\glab : \tbool^\glab
      }
      \and
      \inferrule*[right=Int]{
      }{ 
        \Gamma  \vdash n^\glab : \tpint^\glab
      }
      \and
      \inferrule*[right=Var]{
      }{ 
        \Gamma  \vdash x : \Gamma(x) 
      }
      \and
      \inferrule*[right=Cast]{
        \Gamma  \vdash e : U' \\
        U' \csubtp U
      }{ 
        \Gamma  \vdash e::U : U
      }
      \and
      \inferrule*[right=Bop]{
        \forall i\in\{1,2\}, ~
        \Gamma\vdash e_i: \tau^{\glab_i}
        \\ \glab = \glab_1\cjoin\glab_2
      }{ 
        \Gamma  \vdash e_1 \bop e_2 :  \tau^\glab
      }
    \end{mathpar}

    \noindent\framebox{$ \Gamma ; \;\gpc \vdash c $}
    \begin{mathpar}
      \inferrule*[right=Skip]{
      }{ 
        \Gamma ; \;\gpc 
        \vdash  \eskip
      }
      \and
      \inferrule*[right=Seq]{ 
        \Gamma;\; \gpc \vdash c_1
        \\ \Gamma;\; \gpc \vdash c_2
      }{ 
        \Gamma; \;\gpc  \vdash c_1; c_2
      }
      \and
      \inferrule*[right=Assign]{
        \Gamma \vdash x : {\tau^\glab}
        \\   \Gamma \vdash e : \tau^{\glab'}
        \\  \gpc\clabless \glab
        \\  \glab'\clabless \glab
      }{ 
        \Gamma ; \;\gpc 
        \vdash  x := e
      }
      \and
      \inferrule*[right=Out]{ 
        \Gamma \vdash e: \tau^\glab
        \\\\ \glab\clabless\lab
        \\ \gpc \clabless\lab
      }{ 
        \Gamma; \;\gpc  \vdash \eoutput(\lab, e)
      }
      \and
      \inferrule*[right=While]{
        \Gamma \vdash e :  \tbool^\glab
        \\\\    \Gamma ; 
        \;\gpc\cjoin \glab \vdash c
      }{ 
        \Gamma ; \;\gpc \vdash 
        \ewhile \; e\ \edo\  c
      }
      \and
      \inferrule*[right=If]{
        \Gamma \vdash e :  \tbool^\glab 
        \\\\    \Gamma ; \;\gpc\cjoin \glab \vdash c_1 
        \\ \Gamma ; \;\gpc\cjoin \glab \vdash c_2
      }{ 
        \Gamma ; \;\gpc \vdash 
        \eif \; e\ \ethen\  c_1\ \eelse\  c_2 
      }
    \end{mathpar}
    \caption{Typing rules for $\wg$}
    \label{fig:app-typing}
  \end{figure}
  
The syntax and typing rules for the language with gradual
security types ($\wg$) are standard as shown in Fig.~\ref{fig:app-syntax}
and Fig.~\ref{fig:app-typing}.

The partial-ordering ($\labless$) and join operation ($\labjoin$) on
security labels ($\lab$) extends to consistent ordering ($\clabless$) 
and consistent-join ($\cjoin$) to account for $?$, as shown in Fig.~\ref{fig:clabelop}.
The consistent subtyping relation is written as $\tau^{\glab_1} \csubtp \tau^{\glab_2}$.

Next,} 
\ifconf{
The syntax and typing rules for the language with gradual
security types ($\wg$) are standard and provided in
the full version of the paper~\cite{full-version}.
In this section,} we present the syntax and typing rules for our
language with gradual information flow types and evidence ($\wge$),
define the translation from $\wg$ to $\wge$, and
explain the operational semantics for our monitor.

\subsection{$\wge$}
\label{sec:while-evd}

\begin{figure}[t]
  \small
    \centering
    \(
    \begin{array}{lcll}
       \textit{Labels} & \lab& \bnfdef & L \bnfalt \ldots \bnfalt H
      \\
      \textit{Label-intervals} & \iota & \bnfdef & [\lab_\textit{low},
                                                            \lab_\textit{high}]
      \\
      \textit{Raw values} & u &\bnfdef & n \bnfalt
                                         \etrue \bnfalt \efalse 
      \\
      \textit{Values} & v & \bnfdef & (\iota\;u)^\glab
      \\
      \textit{Types} & \tau & \bnfdef & \tbool \bnfalt \tpint
      \\
      \textit{Gradual labels} & g& \bnfdef & ? \bnfalt \lab
      \\
      \textit{Gradual types} & U & \bnfdef & \tau^\glab
      \\
      \textit{Typing Context} & \Gamma & \bnfdef & \cdot \bnfalt \Gamma, x: U
      \\ 
      \textit{Cast evidence} & E & \bnfdef & (\iota_1, \iota_2)
      \\
      \textit{Expressions} & e & \bnfdef &
                                           x \bnfalt v \bnfalt e_1 \bop e_2\bnfalt  E^g\; e 
      \\
      \textit{Variable Set} & X & \bnfdef & \{x_1, \ldots, x_n\}
      \\
      \textit{Commands} & c & \bnfdef &
                                        \eskip \bnfalt c_1;c_2   \bnfalt x\, := \,e  \bnfalt \eoutput (\lab, e) 
      \\ & &\bnfalt & \eif^X e\, \ethen\, c_1\, \eelse \, c_2
               \bnfalt \ewhile^X e\, \edo\, c
    \end{array}
    \)
    \caption{Syntax for the language $\wge$}
    \label{fig:evidence-syntax}
    \vspace*{-2mm}
  \end{figure}

The syntax of $\wge$ is shown in Fig.~\ref{fig:evidence-syntax}.
Gradual types, $U$, consist of a type
($\tbool$, or $\tpint$) and a gradual security label, $g$. This label
is either a static security label, denoted $\lab$; or an imprecise
dynamic label, denoted $?$.  As is standard,  $\lab$ is drawn from
$\mi{Labs}$, a set of labels, which is a part of a security lattice
$\mathcal{L} = (\mi{Labs}, \labless)$. Here $\labless$ is a partial
order between labels in $\mi{Labs}$.
We commonly use the label $H$ to indicate
secret, $L$ to indicate public data, and $L \labless H$.

\ifconf{
The partial-ordering ($\labless$) and join operation ($\labjoin$) on
security labels ($\lab$) extends to consistent ordering ($\clabless$) 
and consistent-join ($\cjoin$) to account for $?$, as shown in Fig.~\ref{fig:clabelop}.
The consistent subtyping relation is written as $\tau^{\glab_1} \csubtp \tau^{\glab_2}$.

\begin{figure}
  \small
    \begin{mathpar}
      \inferrule*{\lab_1\labless\lab_2 }{ 
        \lab_1\clabless\lab_2
      }
      \and
      \inferrule*{ }{ 
        ?\clabless\glab
      }
      \and
      \inferrule*{ }{ 
        \glab\clabless ?
      }
      \and
      \inferrule*{\glab_1 \clabless \glab_2}{ 
        \tau^{\glab_1}\csubtp\tau^{\glab_2}
      }
    \end{mathpar}
    \begin{mathpar}
      \inferrule*{ }{ 
        \lab_1\cjoin\lab_2 = 
        \lab_1\labjoin\lab_2 
      }
      \and
      \inferrule*{
        \glab \neq \top
      }{ 
        ?\cjoin\glab = \,?
      }
      \and
      \inferrule*{
        \glab \neq \top
      }{ 
        \glab\cjoin \, ? = \,?
      }
      \and
      \inferrule*{
      }{ 
        ?\cjoin\top = \top
      }
      \and
      \inferrule*{
      }{ 
        \top\cjoin \, ? = \top
      }
    \end{mathpar}
        \vspace*{-2mm}
    \caption{Operations on gradual labels and types}
    \label{fig:clabelop}
            \vspace*{-2mm}
  \end{figure}
}
\begin{figure*}
  \small
    \begin{mathpar}
      \inferrule*{ }{\gamma(?) = [\bot, \top]}
      \and
      \inferrule*{ }{\gamma(\lab) = [\lab, \lab]}
      \and
      \inferrule*{\lab_l \labless \lab_r}{\m{valid}([\lab_l, \lab_r])}
      \and
      \inferrule*{
        \lab_2\labless \lab_1
        \\ \lab'_1\labless \lab'_2
      }{
        [\lab_1, \lab'_1] \gsubtp [\lab_2, \lab'_2]
      }
      \and
      \inferrule*{
        \lab'_1\labless \lab_2
      }{
        [\lab_1, \lab'_1] \labless [\lab_2, \lab'_2]
      }
      \and
      \inferrule*
      {
        \lab_{1l} \labless \lab_{1r}\labmeet\lab_{2r}
        \\ \lab_{2l}\labjoin\lab_{1l} \labless \lab_{2r}
      }{
        \refineof([\lab_{1l}, \lab_{1r}], [\lab_{2l}, \lab_{2r}]) =([\lab_{1l},\lab_{1r}\labmeet\lab_{2r}], [\lab_{2l}\labjoin\lab_{1l},\lab_{2r}])
      }
      \and
      \inferrule*
      {
        (\lab_{1l} \nlabless \lab_{1r}\labmeet\lab_{2r}) \vee (\lab_{2l}\labjoin\lab_{1l} \nlabless \lab_{2r})
      }{
        \refineof([\lab_{1l}, \lab_{1r}], [\lab_{2l}, \lab_{2r}])= \eundef
      }
      \and
      \inferrule*{
        \iota_1 = [\lab_1, \lab'_1]
        \\ \iota_2 =  [\lab_2, \lab'_2]
      }{
        \iota_1\labjoin\iota_2 = [\lab_1\labjoin\lab_2, \lab'_1\labjoin\lab'_2]
      }
      \and
      \inferrule*{
        \iota_1 \sqsubseteq \gamma(g_1)
        \\ \iota_2 \sqsubseteq \gamma(g_2)
        \\ g_1\clabless g_2
      }{
        (\iota_1, \iota_2) \vdash \tau^{g_1} \csubtp \tau^{g_2} 
      }
    \end{mathpar}
    \begin{mathpar}
      \inferrule*
      {
        \iota_1 = [\lab_{1l}, \lab_{1r}]
        \\\iota_2 = [\lab_{2l}, \lab_{2r}]
      }{
        \iota_1\bowtie \iota_2 = [\lab_{1l}\labjoin\lab_{2l}, \lab_{1r}\labmeet\lab_{2r}]
      }
      \and
      \inferrule*{\iota'' = (\iota'\bowtie\iota) 
        \\ \iota'' \sqsubseteq \gamma(\glab)}{ 
        \iota'\bowtie (\iota\;u)^\glab = (\iota''\; u)^\glab
      }
      %
      %
      %
      %
      %
      \and
      \inferrule*
      {
        \refineof(\iota_1\bowtie\iota, \iota_2) = (\iota'_1, \iota'_2)
      }{
        \iota \bowtie  (\iota_1, \iota_2) = \iota'_2
      }
      \and
      \inferrule*
      {
        \refineof(\iota_1\bowtie\iota, \iota_2) = \eundef
      }{
        \iota \bowtie  (\iota_1, \iota_2) = \eundef
      }
    \end{mathpar}
    \vspace*{-2mm}
    \caption{Label and label-interval operations}
    \label{fig:label-op}
    \vspace*{-2mm}
  \end{figure*}
  
Recall that examples in Section~\ref{sec:overview-g} use a set of
possible security labels for preventing information leaks. This is
\emph{evidence} attesting to the validity of gradual labels. We use an interval of
labels, representing the lowest and the highest possible label, as the 
refinement only narrows the
interval, similar to $\m{GSL_{Ref}}$~\cite{toro2018}.

There are two types of evidence: a label-interval, $\iota$, 
that justifies the dynamic label $?$; and a pair of intervals
or the cast evidence, $E {=} (\iota_1, \iota_2)$,
that justifies the consistent subtyping 
relation between two gradual types used in casts. Intuitively, $\iota$ represents the range of possible static
labels that would allow the program to type-check. For static labels
$\lab$, the evidence is $[\lab,\lab]$. An interval $[\lab_l,
\lab_r]$ is valid iff $\lab_l \labless \lab_r$. The rest of the paper only 
considers valid intervals. Operations leading
to an invalid interval are aborted. 

A value in $\wge$ is a raw value with an interval of possible security
labels for the gradual label. Raw values are integer  
constants $n$, or boolean values. Expressions include values,
variables, casts, and binary operations on expressions.
An explicit type cast is written $E^g \; e$, where $E$ is the evidence
justifying the type cast and $g$ is the label of the resulting type.

Commands include $\eskip$, sequencing, assignments, branches, loops,
and outputs. This language does not have higher-order stores, so the
left-hand side of the assignment is always a global variable. 
The output command outputs a value at a fixed security label $\lab$. We include this
command mainly to have a clear statement of the system's observable
behavior, so we do not allow output at the imprecise label $?$.
To prevent implicit leaks, we include a \emph{write-set} of variables,
$X$, which takes into account variable writes in both
conditional branches and the loop body. 

\Paragraph{Label-interval operations} 
We first define functions and operations on
the label-intervals that are used by the typing rules and
operational semantics (shown in Fig.~\ref{fig:label-op}). 
The function
$\gamma(g)$ returns the maximum possible label-interval for the
gradual label $g$, assuming $\bot$ and $\top$ are the least and the
greatest element in the lattice, respectively.
Label-intervals form a lattice with the partial ordering 
defined as $\iota_1 \gsubtp \iota_2$.
Here, $\iota_1$ is said to be more precise than $\iota_2$.
The label-intervals are refined throughout the execution of the
program; i.e., they get more precise. Consider
the example in Listing~\ref{introEg3}. Assume that the security
lattice contains two elements $L$ and $H$ such that $L \labless
H$. Initially, $y$ has a label $?$ with the evidence $[L, H]$
indicating that any of the two labels are possible. If
$x:\tbool^H$, then the only possible 
label for $y$ that allows assignment on line~\ref{intr3} is $H$. Thus,
the evidence for the label on $y$ is refined to $[H,H]$, which makes
the program-context's evidence on line~\ref{intr4} $[H, H]$,
disallowing assignments to $L$-labeled variables.    

We define $\iota_1\labless \iota_2$ to mean for every security label
in $\iota_1$, all labels in $\iota_2$ are at higher or equal
positions in the security lattice; and for every security label in
$\iota_2$, all labels in $\iota_1$ are at lower or equal positions
in the security lattice. 
Even though this relation is not used in our typing rules or
operational semantics, it is an invariant that must hold on the results of the binary
label-interval operations used by the noninterference proofs.
The function $\refineof(\iota_1, \iota_2)$ returns the largest
sub-intervals of $\iota_1$ and $\iota_2$ ($\iota_1'$ and $\iota_2'$,
respectively) such that  $\iota_1' \labless 
\iota_2'$. If the relation does not hold between $\iota_1'$ and $\iota_2'$, 
the function returns $\eundef$.

The join of label-intervals is defined as $\iota_1 \labjoin \iota_2$. Note that
this is not to be confused with the join operation in the lattice that the intervals
form. The join of the label-intervals computes the interval
corresponding to all possible joins of security labels in those intervals. 

$\iota_1 \bowtie \iota_2$ computes the intersection of the
intervals  $\iota_1$ and $\iota_2$.
$\iota'\bowtie (\iota\;u)^\glab$ merges the labels for a value. 
$\iota \bowtie  (\iota_1, \iota_2)$ refines $\iota_2$ based on the
intersection of $\iota$ and $\iota_1$.

\Paragraph{Evidence-based consistent subtyping}
Next, we define consistent subtyping relations for both gradual labels
and types as supported by label-intervals.
The consistent subtyping relation between two gradual types is written
as $(\iota_1, \iota_2) \vdash \tau^{g_1} \csubtp \tau^{g_2} $ (defined in
Fig.~\ref{fig:label-op}). In this relation, $\iota_1$, resp. $\iota_2$
represents the set of possible labels for $g_1$, resp. $g_2$,
and $g_1\clabless g_2$. The
evidence $(\iota_1,\iota_2)$ is to justify the consistent
security label partial ordering relation between the labels of the
gradual types. 
Note that we do not have $\iota_1\labless \iota_2$ in the premise. The
reason is that $(\iota_1, \iota_2) \vdash \tau^{g_1} \csubtp \tau^{g_2}$ is used
to type runtime terms; even though $\iota_1\labless \iota_2$ holds
initially, as label-intervals are refined from $\iota_i$ to
$\iota'_i$, we cannot guarantee that $\iota'_1\labless \iota'_2$
hold. This will break preservation proofs. 
It is not the gradual type system's job to ensure all execution paths
are secure.  Instead, the runtime semantics will refine the intervals and
abort the computation if necessary when a term is evaluated. 

%
%
%

\begin{figure*}
    \flushleft
    \noindent\framebox{$\Gamma \vdash e : U$}
    \begin{mathpar}
      \inferrule*[right=G-Const]{
        \iota \gsubtp \gamma(\glab)
      }{ 
        \Gamma  \vdash (\iota\; u)^\glab : \Gamma(u)^{\glab}
      }
      \and
      \inferrule*[right=G-Var]{
      }{ 
        \Gamma  \vdash x : \Gamma(x) 
      }
      \and
      \inferrule*[right=G-Bop]{
        \forall i\in\{1,2\}, ~
        \Gamma\vdash e_i: \tau^{\glab_i}
        \\\\ \glab = \glab_1\cjoin\glab_2
      }{ 
        \Gamma  \vdash e_1 \bop e_2 :  \tau^\glab
      }
      \and
      \inferrule*[right=G-Cast]{
        \Gamma\vdash e: \tau^{\glab_1}
        \\\\ E \vdash \tau^{\glab_1} \csubtp \tau^g
      }{ 
        \Gamma  \vdash E^g\; e: \tau^g
      }
    \end{mathpar}

    \noindent\framebox{$ \Gamma ; \iota_\pc\;\gpc \vdash c $}
    \begin{mathpar}
      \inferrule*[right=G-Skip]{
      }{ 
        \Gamma ; \iota_\pc\;\gpc 
        \vdash  \eskip
      }
      \and
      \inferrule*[right=G-While]{
        \Gamma \vdash e :  \tbool^\glab
        \\ X = \wtsetof(c)
        \\ \iota_c = \gamma(\glab)
        \\    \Gamma ; 
        \iota_\pc\labjoin\iota_c\;\gpc\cjoin \glab \vdash c
      }{ 
        \Gamma ; \iota_\pc\;\gpc \vdash 
        \ewhile^X \; e\ \edo\  c
      }
      \and
      \inferrule*[right=G-Seq]{ 
        \Gamma;\iota_\pc\; \gpc \vdash c_1
        \\\\   \Gamma;\iota_\pc\; \gpc \vdash c_2
      }{ 
        \Gamma; \iota_\pc\;\gpc  \vdash c_1; c_2
      }
      \and
      \inferrule*[right=G-Assign]{
        \Gamma \vdash x : \tau^{\glab}
        \\   \Gamma \vdash e : \tau^{\glab}
        \\\\  \iota_\pc \gsubtp \gamma(\gpc)
        \\  \gpc \clabless \glab
      }{ 
        \Gamma ; \iota_\pc\;\gpc 
        \vdash  x:= e
      }
      \and
      \inferrule*[right=G-Out]{ 
        \Gamma \vdash e: \tau^{\lab} 
        \\ \iota_\pc \gsubtp \gamma(\gpc)
        \\ \gpc \clabless\lab
      }{ 
        \Gamma; \iota_\pc\;\gpc  \vdash \eoutput (\lab, e)
      }
      \and
      \inferrule*[right=G-If]{
        \Gamma \vdash e :  \tbool^\glab
        \\ \iota_c = \gamma(\glab)
        \\ X = \wtsetof(c_1)\cup \wtsetof(c_2)
        \\    \forall i = \{1,2\},\ \Gamma ; 
        \iota_\pc\labjoin\iota_c\;\gpc\cjoin \glab \vdash c_i 
      }{ 
        \Gamma ; \iota_\pc\;\gpc \vdash 
        \eif^X \; e\ \ethen\  c_1\ \eelse\  c_2 
      }
    \end{mathpar}
    \vspace*{-3mm}
    \caption{Typing rules for $\wge$.
      $\Gamma(u)$ maps a constant to its type (e.g. $n$ to $\tpint$, and $\etrue$ to $\tbool$)}.
    \label{fig:evidence-gradual-typing}
    \vspace*{-2mm}
  \end{figure*}
  
\Paragraph{Typing rules} 
Expressions and commands with evidence are typed using rules shown in 
Fig.~\ref{fig:evidence-gradual-typing}. 

\textsc{G-Cast} casts an expression of type $U_1$ to
$U_2$, if the cast evidence $E$ shows that $U_1$ is a consistent
subtype of $U_2$. 

We augment the command typing with an interval for the gradual $\pc$
label; $\iota_\pc$ is the range of possible static labels for the
$\gpc$.  The rules are similar to the ones in the original type-system
except for the use of evidence for consistent ordering between the
gradual labels. An exception is the use of $\wtsetof(c)$ that
returns the set of variables being updated or assigned to in the
command $c$. $\wtsetof$ is
 straightforwardly inductively defined over the structure of $c$ and
is shown below:
\[\small
  \begin{array}{l@{\quad}l}
    \mi{WtSet}(\m{skip}) = \emptyset & \mi{WtSet}(\m{output}(\lab, e)) = \emptyset \\ 
    \mi{WtSet}(x := e) = \{x\} &       \mi{WtSet}(\ewhile\; e\; \edo\ c) = \mi{WtSet}(c)    \\
    \multicolumn{2}{l}{\mi{WtSet}(c_1; c_2) = \mi{WtSet}(c_1)\cup  \mi{WtSet}(c_2)} \\
    \multicolumn{2}{l}{\mi{WtSet}(\eif\; e\; \ethen\ c_1\ \eelse\; c_2) = \mi{WtSet}(c_1)\cup  \mi{WtSet}(c_2)}
  \end{array}
\]
Further, \rulename{G-Assign} and
\rulename{G-Out} do not consider expression subtyping 
and instead rely on the casts
inserted by translation.

\begin{figure*}
    \flushleft
    \noindent\framebox{$\Gamma \vdash e \leadsto e' : U$}
    \begin{mathpar}
      \inferrule*[right=T-Bool]{
        \iota = \gamma(\glab)
      }{ 
        \Gamma  \vdash b^\glab \leadsto (\iota\;b)^\glab : \tbool^\glab
      }
      \and
      \inferrule*[right=T-Int]{
        \iota = \gamma(\glab)
      }{ 
        \Gamma  \vdash n^\glab \leadsto (\iota\;n)^\glab : \tpint^\glab
      }
      \and
      \inferrule*[right=T-Var]{
        \Gamma(x) = \tau^\glab
      }{ 
        \Gamma  \vdash x \leadsto x : \tau^\glab
      }
      \and
      \inferrule*[right=T-Bop]{
        \forall i\in\{1,2\}, ~
        \Gamma\vdash e_i \leadsto e_i' : \tau^{\glab_i}
        \\\\ \glab = \glab_1\cjoin\glab_2
      }{ 
        \Gamma  \vdash e_1 \bop e_2 \leadsto e_1' \bop e_2' : \tau^\glab
      }
      \and
      \inferrule*[right=T-Cast]{
        \Gamma  \vdash e \leadsto e' : \tau^{\glab_1} 
        \\ \glab_1 \clabless \glab \\\\ 
        (\iota_1, \iota_2) = \refineof(\gamma(\glab_1), \gamma(\glab)) 
      }{ 
        \Gamma  \vdash (e :: \tau^\glab) \leadsto (\iota_1, \iota_2)^g\; e' : \tau^\glab 
      }
    \end{mathpar}

    \noindent\framebox{$ \Gamma ; \;\gpc \vdash c \leadsto c'$}
    \begin{mathpar}
      \inferrule*[right=T-Assign]{
        \Gamma(x) = \tau^\glab
        \\   \Gamma \vdash e \leadsto e' : \tau^{\glab'}
         \\\\ \glab' \clabless \glab \\ 
        (\iota_1, \iota_2) = \refineof(\gamma(\glab'), \gamma(\glab))
      }{ 
        \Gamma ; \;\gpc 
        \vdash  x := e \leadsto x := (\iota_1, \iota_2)^g\;e' 
      }
      \and
      \inferrule*[right=T-Out]{ 
        \Gamma \vdash e \leadsto e' : \tau^\glab
        \\  \glab \clabless \lab
        \\\\ (\iota_1, \iota_2) = \refineof(\gamma(\glab), \gamma(\lab)) 
      }{ 
        \Gamma; \;\gpc  \vdash \eoutput(\lab, e) \leadsto \eoutput(\lab,  
        (\iota_1, \iota_2)^g\;e')
      }
      \and
      \inferrule*[right=T-While]{
        \Gamma \vdash e \leadsto e' : \tbool^\glab
        \\    \Gamma ; \;\gpc\cjoin \glab \vdash c \leadsto c'
        \\\\ X = \wtsetof(c')
      }{ 
        \Gamma ; \;\gpc \vdash 
        \ewhile \; e\ \edo\  c \leadsto \ewhile^X \; e'\ \edo\  c'
      }
      \and
      \inferrule*[right=T-If]{
        \Gamma \vdash e \leadsto e' : \tbool^\glab
        \\    \Gamma ; \;\gpc\cjoin \glab \vdash c_1 \leadsto c_1'
        \\\\    \Gamma ; \;\gpc\cjoin \glab \vdash c_2 \leadsto c_2'
        \\ X = \wtsetof(c_1') \cup \wtsetof(c_2')
      }{ 
        \Gamma ; \;\gpc \vdash 
        \eif \; e\ \ethen\  c_1\ \eelse\  c_2  \leadsto \eif^X \; e'\ \ethen\  c_1'\ \eelse\  c_2'
      }
    \end{mathpar}
    \vspace*{-2mm}
    \caption{Translation from $\wg$ to $\wge$}
    \label{fig:translation}
    \vspace*{-2mm}
\end{figure*}
\subsection{From $\wg$ to $\wge$}
The programs are written in $\wg$, the language without
evidence, which is then translated to
$\wge$. The
explicit casts for
\rulename{Assign} and \rulename{Out} are automatically inserted
 to account for the subtyping
of expressions. We show the interesting rules for
translating $\wg$ expressions and commands to
$\wge$ with evidence insertion in Fig.~\ref{fig:translation}.
\rulename{T-Assign} inserts a cast for the expression $e$ to have
  the same label as that of $x$. For example, $x^L := y^?$ 
  is rewritten to $x:=(\refineof([\bot,\top], [L,L])^L) y$. Similarly,
  \rulename{T-Out} casts the expression $e$ to the channel level $\ell$.

The other interesting translation rules
are \rulename{T-If} and \rulename{T-While}, which insert a write-set
$X$ that includes the set of all variables that might be written to in
both the branches and the loop body. 
We prove that any well-typed term in $\wg$ is translated to another
well-typed term in $\wge$.
\iffull{
The lemmas and their proofs are described in Appendix~\ref{app:trans-lemmas}.
}

\subsection{Operational Semantics}
\label{sec:monitor-semantics}

\noindent{\bf Runtime constructs:}
We define additional runtime constructs for our
semantics, shown below. The store, $\delta$, maps variables to values
with their gradual labels and intervals. The gradual labels of the
variables are suffixed on the values for the purpose of evaluation.
$\kappa$ is a stack of $\pc$ labels, each of which is a gradual label,
$g_\pc$, with the corresponding interval, $\iota_\pc$. 
\[
  \begin{array}{l@{~~~}lcl}\small
    \textit{Store} & \delta & \bnfdef & \cdot\, \bnfalt 
                                        \delta, x \mapsto v
    \\
    \textit{PC Stack} & \kappa & \bnfdef & \emptyset \bnfalt
                                           (\iota_\pc \, \gpc) \bnfalt \kappa_1\rhd\kappa_2
    \\
    \textit{Actions} & \alpha & \bnfdef & \cdot\, \bnfalt (\lab, v)
    \\
    \textit{Commands} & c & \bnfdef & \cdots \bnfalt \{c\} \bnfalt
                                      \eif \ e\, \ethen\, c_1\, \eelse
                                      \, c_2
  \end{array}
\]
The stack is
used for evaluating nested if statements. The operation $\kappa_1 \rhd
\kappa_2$ indicates that $\kappa_1$ is on top of $\kappa_2$ in the stack.
$\alpha$ is an action, which may be silent or a labeled output.
We add two runtime commands. $\{c\}$ is used in evaluating if
statements. The curly braces help the monitor keep track of the scope of a
branch. 
The if statement without the write set is used in an intermediate evaluation state.

\begin{figure}
    \flushleft
    \noindent{\framebox{$ \delta \sepidx{} e \evalsto v$}}
    \begin{mathpar}
      \inferrule*[right=M-Const]{
      }{
        \delta\sepidx{} (\iota\, u)^\glab \evalsto  (\iota\, u)^\glab
      }
      \and
      \inferrule*[right=M-Var]{
      }{
        \delta\sepidx{} x \evalsto  \delta(x)
      }
      \and
      \inferrule*[right=M-Bop]{
        \forall i \in \{1,2\},~\delta\sepidx{} e_i \evalsto (\iota_i\, u_i)^{\glab_i} \\
        \iota = \iota_1\labjoin\iota_2 \\ \glab = \glab_1\cjoin\glab_2 \\ u = u_1\bop u_2
      }{
        \delta\sepidx{} e_1\bop e_2 \evalsto   (\iota\;u)^{\glab}
      }
      \and
      \inferrule*[right=M-Cast]{
        \delta\sepidx{} e \evalsto (\iota\, u)^{\glab} \\
        \iota' = \iota \bowtie E
      }{
        \delta\sepidx{} E^{\glab'}\; e
        \evalsto  (\iota'\; u)^{\glab'}
      }
      \and
       \inferrule*[right=M-Cast-Err]{
        \delta\sepidx{} e \evalsto (\iota\, u)^{\glab} \\
        \iota \bowtie E = \eundef
      }{
        \delta\sepidx{} E^{\glab'}\; e
        \evalsto  \eabort
      }
    \end{mathpar}
        \vspace*{-2mm}
    \caption{Monitor semantics for expressions}
    \label{fig:mon-semantics-exp}
    \vspace*{-2mm}
  \end{figure}

\begin{figure}
    \flushleft
    \noindent{\framebox{$ \kappa,\delta\sepidx{} c \stackrel{\alpha}{\stepsto} \kappa',\delta'\sepidx{} c'$}}
    \begin{mathpar}
      \inferrule*[right=M-Seq]{
        \kappa,\delta\sepidx{} c_1  \stackrel{\alpha}{\stepsto}  \kappa',\delta' \sepidx{} c'_1
      }{
        \kappa,\delta\sepidx{} c_1; c_2
        \stackrel{\alpha}{\stepsto}  \kappa',\delta'\sepidx{} c'_1; c_2
      }
      \and
      \inferrule*[right=M-Pc]{
        \kappa,\delta\sepidx{} c 
        \stackrel{\alpha}{\stepsto}  \kappa',\delta'\sepidx{}c'
      }{
        \kappa\rhd \iota_\pc\;\gpc,\delta\sepidx{} \{c\} 
        \stackrel{\alpha}{\stepsto}  \kappa'\rhd\iota_\pc\;\gpc,\delta'\sepidx{}\{c'\}
      }
      \and
      \inferrule*[right=M-Pop]{
      }{
        \iota_\pc\;\gpc\rhd\kappa,\delta\sepidx{} \{\eskip\} 
        \stepsto  \kappa,\delta\sepidx{} \eskip
      }
      \and
      \inferrule*[right=M-Skip]{
      }{
        \kappa,\delta\sepidx{} \eskip; c
        \stepsto  \kappa,\delta\sepidx{} c
      }
      \and
      \inferrule*[right=M-Assign]{
        \delta\sepidx{} e \evalsto v 
        \\ v' = \reflvof{}(\iota_\pc, v)
        \\ v'' = \updval(\labof(\delta(x)), v')
      }{
        \iota_\pc\;\gpc,\delta \sepidx{} x:= e
        \stepsto  \iota_\pc\;\gpc, 
        \delta[x\mapsto v'']\sepidx{} \eskip
      }
      \and
      \inferrule*[right=M-Out]{
        \delta\sepidx{} e \evalsto v 
        \\ v' = \reflvof{}(\iota_\pc, v)
        \\ v'' = \updval(\labof(\delta(x)), v')
      }{
        \iota_\pc\;\gpc,\delta \sepidx{} \eoutput(\lab, e)
        \stackrel{(\lab, v'')}{\stepsto}  \iota_\pc\;\gpc, 
        \delta\sepidx{} \eskip
      }
      \and
      \inferrule*[right=M-If]{
        \iota'_\pc = \iota_\pc\labjoin \iota
        \\ \gpc' =\gpc\cjoin\glab \\\\
        c_i = c_1~\mbox{if}~b = \etrue\\
        c_i = c_2~\mbox{if}~b = \efalse
      }{
        \iota_\pc\;\gpc,\delta\sepidx{} \eif\; (\iota\;b)^\glab\ \ethen\ c_1\ \eelse\ c_2
        \stepsto \\\\ \iota'_\pc\;\gpc'\rhd \iota_\pc\;\gpc, \delta \sepidx{}  \{c_i\}
      }
      \and
      \inferrule*[right=M-If-Refine]{
        \delta\sepidx{} e \evalsto v \\
        \delta' = \rflof{}(\delta, X, \iota_\pc\labjoin \labof(v))
      }{
        \iota_\pc\;\gpc,\delta\sepidx{} \eif^X\;e\ \ethen\ c_1\ \eelse\ c_2
        \stepsto
        \\\\
        \iota_\pc\;\gpc, \delta' \sepidx{} \eif\;v\ \ethen\ c_1\ \eelse\ c_2
      }
      \and
      \inferrule*[right=M-While]{
      }{
        \iota_\pc\;\gpc,\delta\sepidx{} \ewhile^X\;e\ \edo\ c
        \stepsto   \iota_\pc\;\gpc, \delta \sepidx{} \\\\
        \eif^X\;e\ \ethen\ (c; \ewhile^X\;e\ \edo\ c)\ \eelse\ \m{skip}
      }
    \end{mathpar}
        \vspace*{-2mm}
    \caption{Monitor semantics for commands}
    \label{fig:mon-semantics-cmd}
    \vspace*{-2mm}
  \end{figure}

\begin{figure}[tb]
    \begin{mathpar}
      \inferrule*{ 
        \refineof(\iota_c,\iota) = (\_,\iota')
      }{
        \reflvof{}(\iota_c, (\iota\;u)^\glab) = (\iota'\;u)^\glab
      }
      \and
      \inferrule*{ }{
        \rflof{}(\delta, \cdot, \iota) = \delta
      }
      \and 
      \inferrule*{  
        \delta' = \rflof{}(\delta, X, \iota)
        \\ v' = \reflvof{}(\iota, v) \neq \eundef
      }{
        \rflof{}((\delta, x\mapsto v), (X, x), \iota) 
        = \delta', x \mapsto v'
      }
      \and
       \inferrule*
      {
         \lab_{1l}\labjoin\lab_{2l} \labless \lab_{1r}
      }{
        \newlab([\lab_{1l}, \lab_{1r}], [\lab_{2l}, \lab_{2r}]) =
         [\lab_{1l}\labjoin\lab_{2l}, \lab_{1r}]
      }
       \and
\inferrule*{ }
 {\updval(\iota_o, (\iota_n\, u_n)^\glab) =  (\newlab(\iota_o, \iota_n) u_n)^\glab 
}
\end{mathpar}
    \vspace*{-2mm}
    \caption{Label-interval operations for the monitor}
    \label{fig:store-aux-single}
    \vspace*{-2mm}
  \end{figure}

\Paragraph{Expression monitoring semantics}
Our monitoring semantics for expressions is of the form
$\delta \sepidx{} e \evalsto e'$ as shown in
Fig.~\ref{fig:mon-semantics-exp}.
Rules \rulename{M-Const} and \rulename{M-Var} are standard. To perform a
binary operation on two values, the operation is performed on the raw
values, and the join of their associated intervals and gradual labels
is assigned to the computed value. \rulename{M-Cast} refines a 
value's interval according to the cast evidence. If the refinement is not
valid, the execution aborts (\rulename{M-Cast-Err}). Note that none of
these operations modify the gradual label of the variable (the type of
store locations remain the same); the operations only refine the
intervals of the gradual label.

\Paragraph{Commands monitoring semantics}
Our monitoring semantics for commands is summarized in
Fig.~\ref{fig:mon-semantics-cmd} and has the form 
$\kappa,\delta\sepidx{} c \stackrel{\alpha}{\stepsto} \kappa',\delta'\sepidx{}
c'$. 
Additional 
  set of rules where the monitor aborts can be found in 
  Fig.~\ref{fig:app-mon-semantics-cmd} in the Appendix.
Rules \rulename{M-Pc} and \rulename{M-Pop} manage
commands running in branches or loops. \rulename{M-Pop} pops the
top-most $\pc$ label from the stack, indicating the end of the branch
or loop. We use braces around a command, $\{ c \}$ to indicate that
$c$ is executing in a branch or loop. Such a command is run taking
into account only the specific branch's $\pc$ stack. When 
the command execution finishes, the braces are removed and the current 
$\pc$ label is popped off the stack. 

Rule~\rulename{M-Assign} updates the
label-interval of the value being assigned based on the
assignment's context $\iota_\pc$ to prevent information leaks. The resulting label-interval is further restricted using
the existing label-interval of the variable to ensure that we only
refine the set of possible labels. 
The function $\labof(v)$ 
returns the label-interval 
of $v$.  Formally:
$$\labof((\iota\;u)^\glab) = \iota$$

%
The function $\reflvof{}$ 
refines the lower bound of a value's label-interval and raises it 
based on the interval $\iota_\pc$.
The $\updval$ function is defined in Fig.~\ref{fig:store-aux-single} and uses the
interval $\newlab$ operation. We
  raise the lower bound of the existing interval based on the interval
  of the newly computed value.
Note that $\newlab$ differs from $\refineof$ in the upper-bound of the
computed interval. 
If either of these functions return
an invalid interval, the execution aborts.
 The interval need not be checked against the reference's
 label-interval because the inserted cast would have ruled out unsafe
 programs earlier.  Consider the
 following program, where $\delta= [x\mapsto([L,L]\, 3)^L,
   y\mapsto([H,H]\,5)^H]$
\[x := ([\bot,\top],[L,L])\,([H,H],[\bot,\top])\,y\]
This program tries to cast an $H$ value to ?, then back to L, which is accepted by the type system.
The expression to be assigned to $x$ is first evaluated to
$([\bot,\top],[L,L])\,([H,H],[\bot,\top])\,([H,H]\,5)^H$, then to
$([\bot,\top],[L,L])\,([H,H]\,5)^H$, then aborts, because
$[H,H]\bowtie([\bot,\top],[L,L])$ evaluates to $\refineof([H,H], [L,L])=\eundef$. 


We explain the assignment rule via examples. Consider a three-point
lattice $L\labless M\labless H$, the following command $ x := [H,
  H]\;5^?$, and two stores $\delta_1 = x\mapsto [M, H]\;1^?$ and
$\delta_2 = x\mapsto [L, M]\;2^?$. Assume the following current 
$\pc$-interval $\iota_\pc = [L, H]$.  Here, $v = [H, H]\;5^?$.  The second
premise further refines the interval of $v$ to make sure that the $\pc$
context is lower than or equal to the interval of the value to be
written. This is to prevent low assignments in a high context. For this
example, $\refineof([L,H], [H,H]) = ([L,H], [H,H])$, so the
intervals remain the same.  Next, we narrow down the possible
  label set using the existing value's interval. This is to adhere to
  our design choice that we do not change the type of the variables
  and only narrow down the label choices
  (Section~\ref{sec:overview-g}).  
For $\delta_1$, $v'' =  [H,
H]\;5^?$, so now $x$ stores a secret value with label $H$. For store
$\delta_2$, $\newlab([L,M], [H,H])$ is not defined and the monitor
aborts.

When the old value in the store has a label-interval
  that is lower than the label-interval of the
  new value to be stored, the monitor aborts; for instance, under the store where
  $x\mapsto([L,L] 0)^?$, the monitor aborts the execution of
  $x:=([H,H]1 )^?$, as $\newlab([L,L],[H,H])$ is not defined.
  This is also consistent with our design choice to only
  refine the label set, not update it.

\rulename{M-Out} makes similar
comparisons as \rulename{M-Assign} to ensure that the output is permitted.
Rule \rulename{M-If} is standard. 
The $\pc$ label is determined by joining the current $\pc$ with the
gradual label and interval of the branch-predicate's value. Here, the
$\pc$ stack grows and the branch is placed in the scoping braces.
The rule for \texttt{while} reduces it to \texttt{if}. 
Rule \rulename{M-If-Refine} is the key for
preventing implicit leaks. We refine the intervals for variables in 
both branches according to the write set, $X$, which contains the set
of all variables being updated in either one of the two
branches. Fig.~\ref{fig:store-aux-single} includes the auxiliary  
definitions for refining the intervals of variables in a
write set.
The function $\rflof{}$ refines the
label-interval of values in the store and is defined inductively.
Here, it is used to refine the intervals of the variables in
the write set to be at least as high as the lower label in the interval of
the current $\pc$. 
%
When the functions $\rflof{}$ return $\eundef$, the
execution aborts. 

\Paragraph{Example}
Below, we define two initial memories, $\delta_t$ maps $x$
to $\etrue$ and $\delta_f$ maps $x$ to $\efalse$. Both $y$ and $z$
store $\etrue$ initially with $y$'s label being ? and $z$ being $L$
such that $L \sqsubseteq H$.
\[
  \begin{array}{lcl}
    \delta_y & = & y\mapsto [L, H] \etrue^?
    \\ \delta_z & = & z\mapsto [L, L] \etrue^L
    \\ \delta_t & = & x\mapsto [H, H]\etrue^H 
    \\ \delta_f & = & x\mapsto [H, H]\efalse^H 
    \\  c_1 & = & \eif^{\{y\}}\ x\ \ethen\ y := [L, H]\efalse^?\ \eelse\ \m{skip}
    \\  c_2& = & \eif^{\{z\}}\ y\ \ethen\ z := [L,L]\efalse^L\ \eelse\ \m{skip}
  \end{array}
\]
Below is the execution starting from the state where $x$ is $\etrue$.
\[ \begin{array}{r@{}l}
    &[L, L]\; L,\ (\delta_t, \delta_z, \delta_y)
    \sepidx{} c_1; c_2 
    \\
    \stepsto &{[L, L]}\; L,\ (\delta_t, \delta_z, y\mapsto [H, H] \etrue^? )
    \sepidx{} \\ &~~~~\eif \ x\ \ethen\ y :=\ [L, H]\efalse^?\ \eelse\
    \m{skip}; c_2
    \\ 
    \stepsto &{[H, H]} \; H \rhd [L, L]\; L,\ (\delta_t, \delta_z, y\mapsto [H, H] \etrue^?)
    \sepidx{}  \\ &~~~~\{y := [L, H]\efalse^?\}; c_2 
    \\
    \stepsto &{[H, H]} \; H \rhd [L, L]\; L,\ (\delta_t, \delta_z, y\mapsto [H, H]
    \efalse^?) \sepidx{}  
    \\ &~~~~\{\m{skip}\}; c_2
    \\
    \stepsto &{[H, H]} \; H \rhd [L, L]\; L,\ (\delta_t, \delta_z, y\mapsto [H, H]
    \efalse^?) \sepidx{}  
    \\ &~~~~\eskip; c_2
    \\
    \stepsto &{[L, L]}\; L,\ (\delta_t, \delta_z, y\mapsto [H, H]
    \efalse^?) \sepidx{}  c_2  
    \\
    \stepsto &\m{abort}
  \end{array}
\]
In the last step, $~\rflof{}$ fails, because the operation
$\mi{refine}([H,H], [L,L])$ produces an invalid label-interval.
Now let's see the execution starting from $x \mapsto \efalse$.
\hspace{-10pt}\[
  \begin{array}{r@{}l}
    &[L, L]\; L, (\delta_f, \delta_z,\delta_y) \sepidx{} c_1; c_2 
    \\
    \stepsto &{[L, L]}\; L, (\delta_f, \delta_z,y\mapsto [H, H] \etrue^? ) \sepidx{} 
    \\ &~~~~\eif\ x\ \ethen\ y := [L, H]\efalse^?\ \eelse\
    \m{skip}; c_2 
    \\
    \stepsto &{[H, H]} \; H \rhd [L, L]\; L, (\delta_f, \delta_z,y\mapsto [H, H] \etrue^?) \sepidx{} 
    \\ &~~~~\{\m{skip}\}; c_2
    \\
    \stepsto & {[H, H]} \; H \rhd [L, L]\; L, (\delta_f, \delta_z,y\mapsto [H, H] \etrue^?) \sepidx{} 
               \\ &~~~~\m{skip}; c_2
    \\
    \stepsto &{[L, L]}\; L, (\delta_f, \delta_z,y\mapsto [H, H]
    \etrue^?) \sepidx{}  c_2  
    \\
    \stepsto &\m{abort}
  \end{array}
\]
Notice that the label-intervals of $y$ are changed the same way as when we start
the execution from $\delta_t$. Ultimately, the program aborts for the
same reason.

\section{Noninterference}
\label{sec:noninterference}

To prove noninterference, we extend $\wge$ with pairs of
values, expressions, and commands to simulate two executions  
which differ on secret values. 
This allows us to reduce our noninterference proof to a preservation
proof~\cite{pottier2002}. 

\subsection{Paired Execution}

\noindent{\bf Syntax:}
The augmented syntax with pairs is shown below. 
\[
  \begin{array}{lcll}
    \textit{Values } & v & \bnfdef & (\iota\; u)^\glab \bnfalt  \epair{\iota_1\;  u_1}{\iota_2\; u_2}^\glab
    \\
    \textit{Cmd.}& c & \bnfdef & \cdots \bnfalt \epair{\kappa_1, \iota_1,  c_1}{\kappa_2, \iota_2, c_2}_\glab
  \end{array}
\]
The store $\delta$ is
extended to contain pairs of values.
We also extend
commands to be paired but do not allow pairs to be
nested; an invariant maintained by our operational semantics. We only
use pairs for values and commands whose values and effects are
not observable by the adversary (are ``high'').
%
Pairs of commands are part of the runtime statement, generated as a
result of evaluating a branching statement. Each command
represents an independent execution, capable of changing its own
$\pc$ stack. As a result, we include local $\pc$ stacks in the pair 
with each command. The rationale behind additional $\pc$ stacks in
command pairs is explained with the semantics. 

\Paragraph{Label-interval operations on pairs}
The interval of a paired value is a pair of intervals, defined below. 
$$
\inferrule*{}{\labof(\epair{\iota_1\;u_1}{\iota_2\;u_2}^\glab) = \epair{\iota_1}{\iota_2}}
$$
The
intersection of an interval and a paired value is defined as follows. 
Other extensions to label-interval operations can be found in 
\ifconf{
the full version of the paper~\cite{full-version}. 
}
\iffull{
Appendix~\ref{sec:app-paired}.
}
$$
  \inferrule*{}{\iota\bowtie \epair{\iota_1\;u_1}{\iota_2\;u_2}^\glab = \epair{\iota\bowtie\iota_1\;u_1}{\iota\bowtie\iota_2\;u_2}^\glab}
$$
\noindent{\bf Memory read and update operations:}
As we allow the intervals of values to be refined, the store read
($\eread$) and update ($\eupdate$) operations for paired values
need to make sure that the correct paired value is read or updated.
These functions are shown in Fig.~\ref{fig:app-rdupd} in the Appendix. 

\begin{figure}[t!]
    \flushleft
    \textbf{Expression Semantics: }
    \noindent{\framebox{$ \delta \sepidx{i} e \evalsto v$}}
    \begin{mathpar}
      \inferrule*[right=P-Const]{
      }{
        \delta\sepidx{i} (\iota\, u)^\glab \evalsto  (\iota\, u)^\glab
      }
      \and
      \inferrule*[right=P-Var]{
      }{
        \delta\sepidx{i} x \evalsto  \eread_i\ \delta(x)
      }
      %
      \and
      \inferrule*[right=P-Cast]{
        \delta\sepidx{i} e \evalsto v 
        \\ v' = (E, \glab)\rhd v
      }{
        \delta\sepidx{i} E^{\glab}~ e 
        \evalsto  v'
      }
    \end{mathpar}
    \flushleft
    \textbf{Command Semantics: } 
    \noindent{\framebox{$ \kappa,\delta \sepidx{i} c \stackrel{\alpha}{\stepsto}
        \kappa',\delta'\sepidx{i} c'$}}
    
    \begin{mathpar}
      \inferrule*[right=P-C-Pair]{
        \kappa_i\rhd \iota_\pc\labjoin\iota_i\;\gpc\cjoin\glab,\delta\sepidx{i} c_i 
        \stackrel{\alpha}{\stepsto} \\\\ \kappa'_i\rhd \iota_\pc\labjoin\iota_i\;\gpc\cjoin\glab,\delta'\sepidx{i}c'_i
        \\ c_j = c'_j
        \\ \kappa_j = \kappa'_j
        \\ \{i,j\} = \{1,2\}
      }{
        \iota_\pc\;\gpc,\delta\sepidx{} 
        \epair{\kappa_1, \iota_1, c_1}{\kappa_2, \iota_2, c_2}_\glab
        \stackrel{\alpha}{\stepsto} \\\\ \iota_\pc\;\gpc,\delta'\sepidx{} 
        \epair{\kappa'_1, \iota_1, c'_1}{\kappa'_2, \iota_2, c'_2}_{\glab}
      }
      \and
      \inferrule*[right=P-Lift-If]{
        c_j = c_1~\mbox{if}~u_1 = \etrue
        \\ c_j = c_2~\mbox{if}~u_1 = \efalse
        \\ c_k = c_1~\mbox{if}~u_2 = \etrue
        \\ c_k = c_2~\mbox{if}~u_2 = \efalse
      }{
        \iota_\pc\;\gpc,\delta\sepidx{}
        \eif\; \epair{\iota_1\;u_1}{\iota_2\;u_2}^\glab
        \ \ethen\ c_1\ \eelse\ c_2
        \stepsto \\\\ \iota_\pc\;\gpc,\delta\sepidx{} 
        \epair{\emptyset, \iota_{1}, c_j}{
          \emptyset, \iota_{2}, c_k}_g
      }
      \and
      \inferrule*[right=P-Skip-Pair]{
      }{
        \iota_\pc\; \gpc, \delta\sepidx{} \epair{\emptyset, \iota_1, \eskip}{\emptyset, \iota_2, \eskip}_g
        \stepsto \\ \iota_\pc\; \gpc,\delta\sepidx{} \eskip
      }
      \and
      \inferrule*[right=P-Assign]{
        \delta\sepidx{i} e \evalsto v 
        \\ v' = \reflvof{}(\iota_\pc, v)
      }{
        \iota_\pc\;\gpc,\delta \sepidx{i} x:= e
        \stepsto \\ \iota_\pc\;\gpc, 
        \delta[x\mapsto \eupdate_i\ \delta(x)\  v'] \sepidx{i} \m{skip}
      }
      \and
       \inferrule*[right=P-Out]{
         \delta\sepidx{i} e \evalsto v 
         \\  v' = \reflvof{}(\iota_\pc, v)
         \\  v'' = \updval\ [\lab,\lab]\  v'
       }{
         \iota_\pc\;\gpc,\delta \sepidx{i} \eoutput(\lab, e)
         \stackrel{(i, \lab, v'')}{\stepsto} \iota_\pc\;\gpc, 
         \delta\sepidx{i} \eskip
         }
    \end{mathpar}
    
    \caption{Selected rules of paired executions}
        \vspace*{-2mm}
    \label{fig:pair-semantics-c}
    \vspace*{-2mm}
  \end{figure}  

\Paragraph{Operational semantics for pairs}
The operational semantics are augmented with an index, $i$.  The
judgments now are of the form $\delta \sepidx{i} e \evalsto e'$ and
$ \kappa,\delta \sepidx{i} c \stepsto \kappa',\delta'\sepidx{i} c'$.
The index $i$ indicates which branch of a pair is executing (when
$i\in\{1,2\}$) or if it is a top-level execution (when $i$ is omitted).
Most of the rules can be directly obtained by adding the $i$ to the
monitor semantics shown in Fig.~\ref{fig:mon-semantics-exp}
and~\ref{fig:mon-semantics-cmd}.
Rules that deal with pairs, including read and
write to the store need to be modified. 
We explain important rule changes (shown in
Fig.~\ref{fig:pair-semantics-c}).

Rule \rulename{P-Var} uses the function $\eread_i\; v$ to retrieve the
value indexed by $i$ within $v$. 
To evaluate a cast over a pair of values, we push the cast
inside the pair (\rulename{P-Cast}).

Each command in the pair (\rulename{P-C-Pair}) can make
progress independently and the premise of the rule is indexed by the
corresponding $i$. Here $\kappa_i$ is the $\pc$ stack specific to
$c_i$. Consider a command $c=\epair{c_1}{c_2}$, where both $c_1$ and 
$c_2$ have nested if statements. The execution of $c$ will create
different $\kappa_1$ and $\kappa_2$ when executing $c_1$ and
$c_2$. Next, $\iota_i$ is the $\pc$ label-interval demonstrating that
$c$ is supposed to execute in a ``high'' context (unobservable by the
adversary). The bottom $\pc$ in the stack is joined with
$\iota_i$. We will come back to this point when explaining the typing
rules.  

Rule \rulename{P-Lift-If} lifts the pair that appears as branch
conditions to generate a paired command. The resulting commands on each
side of the pair are determined by the value in the corresponding side
of the branch condition.  The branching context $\iota_i$ is the
runtime interval of the branching condition. 
The initial local $\pc$ stack is empty.

Note that the individual branches do not contain pairs of commands.
The only rule that generates paired command is \rulename{P-Lift-If}.
To see how the semantics prevent
nesting command pairs and how paired execution represents low runs
with different secrets, consider the program in Listing~\ref{eg:pair-branch}.
%
%
%
%
%
%
%
Assume that $a$ and $b$ are variables containing paired values such that
$a = \epair{u_{a1}}{u_{a2}}^H$ and $b = \epair{u_{b1}}{u_{b2}}^H$, meaning both
$a$ and $b$ contain secrets and $u_{a1}$ and $u_{b1}$ are values for the
first execution and $u_{a2}$ and $u_{b2}$ are for the second. We ignore
the intervals in this example for simplicity of exposition. On
line~\ref{if1}, we use the \rulename{P-Lift-If} rule since we branch on 
a pair of values to create paired commands. In the first execution,
if $u_{a1} = \etrue$, we take the $\ethen$ branch. When evaluating $b$
inside the branch we take the first part of the pair using the
expression evaluation rules and $\eread_i$ operation
(Fig.~\ref{fig:app-rdupd}) for $i=1$, i.e., $\eread_1\ b = \proj{b}{1} = u_{b1}$.
Thus, the branch on line~\ref{if2}
becomes: $\eif \;u_{b1} \;\ethen \ldots$ while the remaining parts remain
the same. Similarly in the second execution, based on
the value of $u_{a2}$, either the $\ethen$ branch or the $\eelse$ branch
is chosen. If the $\ethen$ branch is chosen, the branch on line~\ref{if2}
becomes $\eif \;u_{b2} \;\ethen \ldots$ as we are in the second execution of
the branch ($i=2$) on line~\ref{if1} and $\eread_2\ b = \proj{b}{2} = u_{b2}$.
Generating two different runs of the program is sufficient for reasoning about
noninterference, which is what the projection semantics
do. 


  \begin{lstlisting}[caption=Example program to explain branching on pairs,
    label=eg:pair-branch,xleftmargin=0.3\columnwidth,float=t]
  if $a$ then              @\label{if1}@
     if $b$ then $y :=[\bot,\top] 1^?$ @\label{if2}@ @\label{assn1}@
     else $\eskip$ 
  else $y :=[\bot,\top] 2^?$ @\label{assn2}@
\end{lstlisting}

Local $\pc$ refinements in pairs are forgotten when both sides of the
pair finish executing in \rulename{P-Skip-Pair}. This is similar to
the \rulename{P-Pop} rule where $\pc$ for the branch or loop is
forgotten. 

Rule \rulename{P-Assign} deals with the complexity of pairs updating
the store in one branch with the helper function $\eupdate_i~v_o~v_n$
(defined in Fig.~\ref{fig:app-rdupd}). The refinement of labels during
store updates is the same as the monitor semantics. When the update
comes from a specific branch of execution ($i\in\{1,2\}$), the value
for the other branch should be preserved. If the value in the store is
already a pair, only the $i^{th}$ sub-expression is
updated. Reconsider the example in Listing~\ref{eg:pair-branch}. The
assignment on line~\ref{assn1} happens in either of the two branches,
or both the branches depending on the values of $u_{a1}$ to
$u_{b2}$. If it happens in only the first projection, the first part
of the value-pair in $y$ is updated. Suppose that $y = \epair{[H,\top]
  0}{[H,\top] 42}^{?}$, initially, and $u_{a1} = u_{b1} =
\etrue$. Then, the value of $y$ after the assignment on
line~\ref{assn1} becomes $y = \epair{[H,\top] 1}{[H,\top] 42}^{?}$. If
$u_{a2} = \efalse$, then the $\eelse$ branch is taken, and at the end
of the assignment on line~\ref{assn2} the value of $y$ is updated to
$y = \epair{[H,\top] 1}{[H,\top] 2}^{?}$. The first part of the pair
is already updated through the $\ethen$ branch as we evaluate the two
runs one after the other when branching on a pair of values.

If the store value is not a pair, the value becomes a pair where the
$i^{th}$ sub-expression is the updated value, and the other 
sub-expression is the old value.  Considering the same
example as above, if initially $y = ([H,\top] 42)^?$, then at the end
of $\ethen$ branch with $u_{a1} = \etrue$, the updated value of $y$ is $y
= \epair{[H,\top] 1}{[H,\top] 42}^{?}$.
When updates happen at the
top-level, the entire value in the store should be updated. The first
rule applies when either the old or the new value is a pair and the
second rule applies when none of them are pairs. Note that this
is the reason why the intervals in a pair may differ. 

The output rule is mostly the same. The event being output now
includes the index to aid the statement and proof of noninterference.
The \rulename{P-If-Refine} rule (for \rulename{M-If-Refine}) uses an 
augmented version of $\rflof{}$, which only refines
label-intervals for the $i^{\mathit{th}}$ branch.


\subsection{Semantic Soundness and Completeness}

To connect the semantics of the extended language with pairs to
the monitor semantics, we prove soundness and completeness theorems. These
theorems depend on projections of the store, expression- and
command-configurations. Similar to the value 
projection seen before, the goal of these projections
is to obtain one execution from a paired execution. 

The projection of a paired value, a paired
interval, a normal value and interval are straightforward and
defined in Fig.~\ref{fig:app-rdupd}. 
The projection of stores ($\delta$) and traces ($\trace$) is
inductively defined as shown in Fig.~\ref{fig:proj}.
The projection function only keeps the output events produced
by the execution of concern and ignores output performed by the other 
execution. The projection function for
expression configurations is $\proj{\delta\sepidx{} e}{i} = \proj{\delta}{i} \sepidx{} e$ and for command configurations is defined in
Fig.~\ref{fig:proj}. The interesting case is the projection
of a command pair. We reassemble the $\pc$ stack and wrap $c_i$
with curly braces to reflect the fact that these pairs only appear in
an if branch. 


\begin{figure}
    \flushleft
    \textbf{Store projection:}
    \begin{mathpar}
      \inferrule*{}{\proj{\cdot}{i} = \; \cdot}
      \and
      \inferrule*{}{\proj{\delta, x \mapsto v}{i} = \; \proj{\delta}{i}, x \mapsto \proj{v}{i}}
    \end{mathpar}
    \flushleft
    \textbf{Trace projection:}
    \[
      \begin{array}{rlrll}
        \proj{\cdot}{i} & = \; \cdot & \\
        \proj{\trace, (\lab, v)}{i} & = \; \proj{\trace}{i}, (\lab, \proj{v}{i}) & \\
        \proj{\trace, (j, \lab, v)}{i} & = \; \proj{\trace}{i}, (\lab, v) &~ \mathit{if}~i = j\\
        \proj{\trace, (j, \lab, v)}{i} & = \; \proj{\trace}{i} &~ \mathit{if}~ i \neq j
      \end{array}
    \]
    \textbf{Command-configuration projection:}
    \begin{mathpar}
      \inferrule*{ }{
        \proj{\iota_\pc\;\gpc, \delta \sepidx{} \m{skip}}{i} =
        \iota_\pc\;\gpc, \proj{\delta}{i} \sepidx{}\m{skip} 
      }
      \and
      \inferrule*{
        \proj{\kappa, \delta \sepidx{} c_1}{i} = \kappa', \delta' \sepidx{} c'_1}{
        \proj{\kappa, \delta \sepidx{} c_1;c_2}{i} = \kappa', \delta' \sepidx{} c'_1; c_2
      }
      \and
      \inferrule*{ } {
        \proj{\iota_\pc\;\gpc, \delta \sepidx{} x:= e}{i} = 
        \iota_\pc\;\gpc, \proj{\delta}{i} \sepidx{}  x := e
      }
      \and
      \inferrule*{ 
        \proj{\kappa, \delta \sepidx{} c}{i} = \kappa', \delta' \sepidx{} c'
      }{
        \proj{\kappa\rhd \iota_\pc\; \gpc, \delta \sepidx{} \{c\}}{i} = \kappa'\rhd
        \iota_\pc\;\gpc, \delta' \sepidx{} \{c'\}
      }
      \and
      \inferrule*{ } {
        \proj{\iota_\pc\;\gpc, \delta \sepidx{} \m{output}(\lab, e)}{i} = 
        \iota_\pc\;\gpc, \proj{\delta}{i} \sepidx{} \m{output}(\lab, e)
      }
      \and
      \inferrule*{\forall \{i,j\} \in \{1,2\},\ c'_i =
        \left\{\begin{array}{ll}
                 \eskip\, & \m{if}\ c_i=\eskip\ \m{and}\ c_j \neq \eskip \arcr
                            \{c_i\}\, & \m{else} \arcr
               \end{array}\right.
      }
      {
        \proj{\iota_\pc\;\gpc, \delta \sepidx{} \epair{\kappa_1, \iota_1, c_1}{\kappa_2, \iota_2, c_2}_\glab}{i} = \\
        \kappa_i\rhd  (\iota_\pc\labjoin\iota_i)\;(\gpc\cjoin\glab) 
        \rhd \iota_\pc\;\gpc, \proj{\delta}{i} \sepidx{} c_i'
      }
      \and
      \inferrule*{ }
      {
        \proj{\iota_\pc\;\gpc, \delta \sepidx{} \epair{\emptyset, \iota_1, \eskip}{\emptyset, \iota_2, \eskip}_\glab}{i} = \\
        \kappa_i\rhd  (\iota_\pc\labjoin\iota_i)\;(\gpc\cjoin\glab) 
        \rhd \iota_\pc\;\gpc, \proj{\delta}{i} \sepidx{} \{\eskip\}
      }
      \and
      \inferrule*{ }{
        \proj{\iota_\pc\;\gpc, \delta \sepidx{} \eif\; v\; \ethen\; c_1\; \eelse\; c_2 }{i} = \\
        \iota_\pc\;\gpc, \proj{\delta}{i} \sepidx{} \eif\; \proj{v}{i} \; \ethen\; c_1\; \eelse\; c_2 
      }
      \and
      \inferrule*{ }{
        \proj{\iota_\pc\;\gpc, \delta \sepidx{} \eif^X\; e\; \ethen\; c_1\; \eelse\; c_2 }{i} = \\
        \iota_\pc\;\gpc, \proj{\delta}{i} \sepidx{} \eif^X\; e\; \ethen\; c_1\; \eelse\; c_2 
      }
      \and
      \inferrule*{ }{
        \proj{\iota_\pc\;\gpc, \delta \sepidx{} \ewhile^X\; e\; \edo\; c}{i} = 
        \iota_\pc\;\gpc, \proj{\delta}{i} \sepidx{} \ewhile^X\; e\; \edo\; c
      }
    \end{mathpar}
        \vspace*{-2mm}
    \caption{Projections}
    \label{fig:proj}
    \vspace*{-2mm}
  \end{figure}

The Soundness theorem ensures that if a configuration can transition
to another configuration, then its projection can transition to the
projection of the resulting configuration, generating the same trace
modulo projection. The Completeness theorem ensures that if both 
projections of a configuration terminate, then the configuration
terminates in an equivalent state.
We write, $\vdash \kappa, \delta \sepidx{i} c\ \m{wf}$,
to indicate that the configuration is well-formed
(defined in Appendix~\ref{sec:app-wf}).
Theorems~\ref{thm:soundness} and~\ref{thm:completeness}
are the formal soundness and completeness theorem statements.
\ifconf{
  The proofs can be found in the full version of the paper~\cite{full-version}.
}
\iffull{
The proofs can be found in Appendix~\ref{sec:app-sound} and Appendix~\ref{sec:app-comp}.
}
\begin{thm}[Soundness]
  \label{thm:soundness}
  If $\kappa,\delta \sepidx{}  c \stackrel{\trace}{\stepsto^*} \kappa',\delta'\sepidx{} c'$
  where $\vdash \kappa, \delta \sepidx{} c\ \m{wf}$, 
  then $\forall i \in \{1, 2\}$, 
  $\proj{\kappa,\delta \sepidx{}  c}{i} \stackrel{\proj{\trace}{i}}{\stepsto^*} 
  \proj{\kappa',\delta'\sepidx{} c'}{i}$
\end{thm}
\begin{thm}[Completeness]
  \label{thm:completeness}
  If $\forall i\in \{1, 2\}$, $\proj{\kappa, \delta \sepidx{} c}{i} \stackrel{\trace_i}{\stepsto^*} 
  \kappa_i, \delta_i \sepidx{} \m{skip}$ and $\vdash \kappa, \delta \sepidx{} c \ \m{wf}$, then $\exists \kappa', \delta'$ s.t. 
  $\kappa, \delta \sepidx{} c \stackrel{\trace}{\stepsto^*} \kappa', \delta' \sepidx{} \m{skip}$,
  $\proj{\kappa', \delta' \sepidx{} \m{skip}}{i} = \kappa_i, \delta_i \sepidx{} \m{skip}$ and $\trace_i = \proj{\trace}{i}$.
\end{thm}

\subsection{Preservation}

Before we explain the typing rules for the extended configuration, we
define another label relation. A gradual label is said to be
``high'' w.r.t an attacker, if the lower label in the interval is
not lower than or equal to the level of the attacker. 
\[  %
  \inferrule*{
    \iota = [\lab_l, \lab_r]
    \\ \lab_l \not\preccurlyeq \lab_A 
    \\ \iota \sqsubseteq \gamma(g)
  }{
    \iota \vdash g \in H(\lab_A) 
  }
\]
All the pair typing rules are parameterized over attacker's label
$\lab_A$, which we omit from the rules  for simplicity. 
The typing rule for value-pairs is shown below. The second
premise checks that the interval is representative of the gradual type
$U$. The last premise checks if $U$'s security
label is high, meaning this pair of values is non-observable to
the adversary. 
\begin{mathpar}
  \inferrule*[right=R-V-Pair]{ 
    \forall i\in\{1,2\}, ~
    \Gamma  \vdash \iota_i\;u_i : \tau^\glab
    \\ \iota_i \sqsubseteq \gamma(\glab)
    \\ \iota_i \vdash (\glab)\in\ H(\lab_A)
  }{ 
    \Gamma  \vdash \epair{\iota_1\;u_1}{\iota_2\;u_2} : \tau^\glab
  }
\end{mathpar}

The judgement for typing commands with pairs is of the form $\Gamma ;
\kappa \vdash_r c$.  
Fig.~\ref{fig:runtime-typing} summarizes these typing rules.  
\begin{figure}[t!]
    \flushleft
    \noindent\framebox{$\Gamma ; \kappa \vdash_r c $}
    \begin{mathpar}
      \inferrule*[right=R-Pop]{
        \Gamma; \kappa  \vdash_r c
      }{ 
        \Gamma ; \kappa \rhd \iota\; \gpc \vdash_r \{c\} 
      }
      \and
      \inferrule*[right=R-End]{
        \Gamma; \iota\; \gpc \vdash c
      }{ 
        \Gamma ; \iota\;\gpc \vdash_r c
      }
      \and
      \inferrule*[right=R-C-Seq]{ 
        \Gamma; \kappa\rhd \iota_\pc\;\gpc \vdash_r  c_1
        \\     \Gamma; \iota_\pc\;\gpc \vdash  c_2
        \\ \kappa\neq \emptyset 
      }{ 
        \Gamma; \kappa\rhd \iota_\pc\;\gpc \vdash_r  c_1; c_2
      }
      \and
      \inferrule*[right=R-C-Pair]{ 
        \forall i\in{1,2},\  
        \Gamma;\kappa_i\rhd(\iota_\pc\labjoin\iota_i)\; (\gpc\cjoin g)
        \vdash_r c_i
        \\ \iota_i\vdash \glab\in H(\lab_A)
      }{ 
        \Gamma; \iota_\pc\;\gpc \vdash_r 
        \epair{\kappa_1, \iota_1, c_1}{\kappa_2, \iota_2, c_2}_\glab
      }
      \and
      \inferrule*[right=R-C-If]{
        \Gamma \vdash e :  \tbool^\glab
        \\ \iota_g = \gamma(\glab)
        \\ \forall i \in \{1,2\},\ \Gamma ; 
        \iota_\pc\labjoin\iota_g\;\gpc\cjoin \glab \vdash c_i
      }{ 
        \Gamma ; \iota_\pc\;\gpc \vdash_r 
        \eif\; e\ \ethen\  c_1\ \eelse\  c_2 }
    \end{mathpar}
        \vspace*{-2mm}
    \caption{Typing rules for commands with pairs}
    \label{fig:runtime-typing}
    \vspace*{-2mm}
  \end{figure}

Rule \rulename{R-Pop} types the inner command with only the top part
of the $\pc$ stack. When the $\pc$ stack contains only one element,
\rulename{R-End} directly uses command typing. For pairs,
\rulename{R-C-Pair} first checks that each $c_i$ is well-typing,
using the $\pc$ context assembled from the local $\pc$ context. The second
premise makes sure that these commands are typed (executed) in a high
context. Here $\iota_i$ is the witness for $g$, which demonstrates
that $c_i$ are high commands. The sequencing statement types the
second command using only the last $\pc$ on the stack because
the execution order is from left to right. We can only encounter
branches in the first part of a sequencing statement and not the
second part before beginning the execution of the second command in
the sequence. The typing rule for if statements without a write set is
straightforward.

We define store, trace and configuration typing in
Fig.~\ref{fig:stc-typing}. The store $\delta$ types in the typing
environment $\Gamma$ if all variables in $\delta$ are mapped to their
respective type and gradual label in $\Gamma$.
We define top-level configuration typing as $\vdash \kappa, \delta, c$.
To type traces and actions, the output
value needs to be well-typed, and the label-interval of  the value 
has to be lower than or equal to the channel label. 

\begin{figure}[t!]
    \flushleft
    \noindent\textbf{Store typing:}
    \begin{mathpar}
      \inferrule*[right=T-S-Emp]{ }{
        \vdash \cdot: \cdot
      }
      \and
      \inferrule*[right=T-S-Ind]{ 
        \vdash \delta: \Gamma
        \\ \Gamma \vdash v : U
      }{
        \vdash \delta, x\mapsto v : \Gamma, x: U
      }
    \end{mathpar}
    \flushleft
    \noindent\textbf{Configuration typing:}
    \begin{mathpar}
      \inferrule*[right=T-Conf]{
        \vdash \delta: \Gamma
        \\  \Gamma; \kappa \vdash_r c
      }{
        \vdash \kappa,\delta, c
      }
    \end{mathpar}
    \flushleft
    \noindent\textbf{Trace typing:}
    \begin{mathpar}
      \inferrule*[right=T-A-Out]{ 
        \vdash v: U
        \\ \vdash \labof(v) \labless [\lab, \lab]
      }{
        \vdash (\lab, v)}
      \and
      \inferrule*[right=T-A-OutI]{ 
        \vdash v: U
        \\  \lab \not\labless \lab_A 
        \\ \vdash  \labof(v) \labless [\lab, \lab]
      }{
        \vdash (i, \lab, v)
      }
    \end{mathpar}
    \begin{mathpar}
      \inferrule*[right=T-T-Emp]{ }{
        \vdash \cdot
      }\and
      \inferrule*[right=T-T-Ind]{ \vdash \alpha\\ \vdash \trace}{
        \vdash \alpha, \trace
      }
    \end{mathpar}
        \vspace*{-2mm}
    \caption{Store, trace and configuration typing}
    \label{fig:stc-typing}
    \vspace*{-2mm}
  \end{figure}


Using these definitions, we prove that our paired execution semantics
preserve the configuration typing and generate a well-typed trace 
(Theorem~\ref{thm:preservation}).
We write, $\vdash \kappa, \delta \sepidx{i} c\ \m{sf}$ for
$i \in \{\cdot, 1, 2\}$, to indicate that the configuration is safe.
We say a configuration is safe if all of the following hold:
\begin{enumerate}
\item if $i\in\{1,2\}$, then $\kappa\in H(\lab_A)$, 
  $\forall x \in \wtsetof(c)$, $\labof(\delta(x))\in H(\lab_A)$ 
\item if $c=\eif\;\epair{\_}{\_}\;\ethen\; c_1\;\eelse\;c_2$,
  then $\forall x\in\wtsetof(c)$, $\labof(\delta(x))\in H(\lab_A)$
\item if $c=\epair{\kappa_1,\iota_1,c_1}{\kappa_2,\iota_2,c_2}_g$,
  then $\forall i\in\{1,2\}$, 
  $\iota_i\vdash g\in H(\lab_A)$,
  and $\forall x\in\wtsetof(c)$, $\labof(\delta(x))\in H(\lab_A)$ 
\end{enumerate}
\iffull{
  The lemmas and the proofs are shown in Appendix~\ref{sec:app-pres}.
}
\begin{thm}[Preservation]
  ~\label{thm:preservation}
  If $\kappa, \delta \sepidx{} c \stackrel{\trace}{\stepsto^*} \kappa', \delta' \sepidx{} c'$
  with $\vdash \kappa, \delta, c$ and $\vdash \kappa, \delta \sepidx{} c \ \m{sf}$, then $\vdash \kappa', \delta', c'$
  and $\vdash \trace$
\end{thm}

\subsection{Noninterference}
\begin{figure}
    \flushleft
    \noindent\textbf{Store equivalence:}
    \begin{mathpar}
      \inferrule*[right=EqV-L]{
        \glab\clabless \lab_A
        \\ \iota \sqsubseteq \gamma(\glab)
        \\ \vdash (\iota\; u)^\glab : U
      }{ 
        \vdash (\iota\; u)^\glab \approx_{\lab_A} (\iota\; u)^\glab :
        U
      }
      \and
      \inferrule*[right=EqV-H]{
        \forall i\in\{1,2\}, 
        \iota_i \vdash \glab\in H(\lab_A) 
        \\ \vdash (\iota_i\; u_i)^\glab : U
      }{ 
        \vdash (\iota_1\; u_1)^\glab \approx_{\lab_A} (\iota_2\; u_2)^\glab
        : U
      }
      \and
      \inferrule*[right=EqS-Emp]{  }{
        \vdash \cdot \approx_{\lab_A} \cdot:\cdot
      }
      \and
      \inferrule*[right=EqS-Ind]{ 
        \vdash \delta_1 \approx_{\lab_A} \delta_2: \Gamma
        \\ \vdash v_1\approx_{\lab_A} v_2 : U
      }{
        \vdash \delta_1, x\mapsto v_1 \approx_{\lab_A} \delta_2, x\mapsto
        v_2 :\Gamma, x:U
      }
    \end{mathpar}
    \flushleft
    \noindent\textbf{Trace equivalence:}
    \begin{mathpar}
      \inferrule*[right=EqT-E]{
      }{ 
        \vdash [] \approx_{\lab_A} []
      }
      \and
      \inferrule*[right=EqT-L]{
        \vdash \trace_1 \approx_{\lab_A} \trace_2 \\
        \lab_1 = \lab_2 \labless \lab_A \\ v_1 = v_2
      }{ 
        \vdash (\lab_1, v_1)::\trace_1 \approx_{\lab_A} (\lab_2, v_2)::\trace_2
      }
      \and
      \inferrule*[right=EqT-Hl]{
        \vdash \trace_1 \approx_{\lab_A} \trace_2 \\
        \lab_1 \not\labless \lab_A 
      }{ 
        \vdash (\lab_1, v_1)::\trace_1 \approx_{\lab_A} \trace_2
      }
      \and
      \inferrule*[right=EqT-Hr]{
        \vdash \trace_1 \approx_{\lab_A} \trace_2 \\
        \lab_2 \not\labless \lab_A
      }{ 
        \vdash \trace_1 \approx_{\lab_A} (\lab_2, v_2)::\trace_2
      }
    \end{mathpar}
        \vspace*{-2mm}
    \caption{Equivalence definitions}
    \label{fig:eq-def}
    \vspace*{-2mm}
  \end{figure}

We show that the gradual type system presented above satisfies
termination-insensitive noninterference. We start by defining
equivalence for values and stores (Fig.~\ref{fig:eq-def}).
Two values are said to be equivalent to an adversary at level $\lab_A$
if either they are both visible to the adversary and are the same, or neither are
observable by the adversary. We also define
equivalence of traces w.r.t an adversary at level $\lab_A$ in Fig.~\ref{fig:eq-def}. 
The following noninterference theorem (Theorem~\ref{thm:security}) states that
given a program and two stores equivalent for an adversary at level
$\lab_A$, if the program terminates in both the runs, then the
$\lab_A$-observable actions on both runs are the same. 

\begin{thm}[Noninterference]~\label{thm:security}
  Given an adversary label $\lab_A$, a program $c$, and two stores $\delta_1$,
  $\delta_2$, s.t., $\vdash \delta_1\approx_{\lab_A}\delta_2: \Gamma$,
  and $\Gamma; [\bot, \bot]\; \bot\vdash c$, and $\forall i\in\{1,2\}$, $[\bot,\bot]\; \bot,
  \delta_i \sepidx{} c \stackrel{\trace_i}{\stepsto^*}
  \kappa_i,\delta'_i,\m{skip}$, then $\vdash \trace_1\approx_{\lab_A}\trace_2$.
\end{thm}

\iffull{The lemmas and proofs can be found in Appendix~\ref{sec:app-ni}.}
We know that only when $\lab \not\labless \lab_A$, the individual runs
can produce $(i, \lab, v)$ because the two runs diverge only when
branching on pairs. Similarly, $(\lab, \epair{v_1}{v_2})$ can only be
produced if $\lab \not\labless \lab_A$, because pairs can only be
typed if each individual interval in the pair is high and rule
\rulename{P-Out} makes sure $\lab$ is lower than or equal to the pair's
interval.  We prove a simple lemma that establishes that
given a well-typed trace of actions $\trace$, $\vdash \proj{\trace}{1}
\approx_{\lab_A} \proj{\trace}{2}$. 
By combining the Preservation, Soundness, and
Completeness Theorems, it follows that
our gradual type system satisfies termination-insensitive noninterference. 

\section{Gradual Guarantees}
\label{sec:graduality}
\begin{figure}
    \flushleft
    \noindent\textbf{Labels and intervals:}
    \begin{mathpar}
      \inferrule*{
      }{
        \lab \gsubtp \; ?
      }
      \and
      \inferrule*{
      }{
        \lab \gsubtp \lab
      }
      \and
      \inferrule*{
        \lab_1'\labless \lab_1 \\
        \lab_2 \labless \lab_2'
      }{
        [\lab_1, \lab_2] \gsubtp [\lab_1', \lab_2']
      }
    \end{mathpar}
    \flushleft
    \noindent\textbf{Expressions:}
    \begin{mathpar}
      \inferrule*{
        \iota_1\gsubtp \iota_2 \\ \glab_1\gsubtp \glab_2
      }{
        (\iota_1\; u)^{\glab_1} \gsubtp (\iota_2\; u)^{\glab_2}
      }
      \and
      \inferrule*{
      }{
        x \gsubtp x
      }
      \and
      \inferrule*{
        e_1\gsubtp e_1' \\ e_2\gsubtp e_2'
      }{
        e_1 \bop e_2 \gsubtp e_1' \bop e_2'
      }
      \and
      \inferrule*{
        E\gsubtp E' \\ \glab_1\gsubtp \glab_2 \\ e_1 \gsubtp e_2
      }{
        E_{\glab_1} e_1 \gsubtp E_{\glab_2}' e_2
      }
    \end{mathpar}
    \flushleft
    \noindent\textbf{Commands:}
    \begin{mathpar}
      \inferrule*{
      }{
        \eskip \gsubtp \eskip
      }
      \and
      \inferrule*{
        c_1\gsubtp c_1' \\ c_2 \gsubtp c_2'
      }{
        c_1;c_2 \gsubtp c_1';c_2'
      }
      \and
      \inferrule*{
        e_1\gsubtp e_2
      }{
        x := e_1 \gsubtp x := e_2
      }
      \and
      \inferrule*{
        e_1\gsubtp e_2 
      }{
        \eoutput (\lab, e_1) \gsubtp \eoutput (\lab, e_2)
      }
      \and
      \inferrule*{
        e_1\gsubtp e_2 \\ c_1 \gsubtp c_1' \\ c_2 \gsubtp c_2'
      }{
        \eif^X \;e_1\; \ethen\; c_1 \;\eelse\; c_2 \gsubtp    \eif^X\; e_2\;\ethen\; c_1'\; \eelse\; c_2'
      }
      \and
      \inferrule*{
        e_1\gsubtp e_2 \\ c_1 \gsubtp c_2
      }{
        \ewhile^X\; e_1\;\edo\; c_1 \gsubtp    \ewhile^X\; e_2\;\edo\; c_2
      }
    \end{mathpar}
    \flushleft
    \noindent\textbf{Store, Types and Typing-context:}
    \begin{mathpar}
      \inferrule*{
        \glab \gsubtp \glab'
      }{
        \tau^\glab \gsubtp \tau^{\glab'}
      }
      \and
      \inferrule*{
        \forall x \in \Gamma.~ \Gamma(x) \sqsubseteq \Gamma'(x)
      }{
        \Gamma \gsubtp \Gamma'
      }
    \end{mathpar}
    \flushleft
    \noindent\textbf{PC-stack and Configurations:}
    \begin{mathpar}
      \inferrule*{
      }{
        \emptyset \gsubtp \emptyset
      }
      \and
      \inferrule*{
        \iota \gsubtp \iota' \\ \glab \gsubtp \glab'
      }{
        \iota\; \glab \gsubtp \iota'\; \glab' 
      }
      \and
      \inferrule*{
        \kappa_1 \gsubtp \kappa_1' \\         \kappa_2 \gsubtp \kappa_2'
      }{
        \kappa_1 \rhd \kappa_2 \gsubtp \kappa_1' \rhd \kappa_2'
      }
      \and
      \inferrule*{
        \forall x \in \delta.~ \delta(x) \sqsubseteq \delta'(x)
      }{
        \delta \sqsubseteq \delta' 
      }
      \and
      \inferrule*{
        \kappa \gsubtp \kappa' \\ \delta \gsubtp \delta' \\ c \gsubtp c'
      }{
        \kappa, \delta \sepidx{} c \gsubtp \kappa', \delta' \sepidx{} c'
      }
    \end{mathpar}
        \vspace*{-2mm}
    \caption{Precision relations}
    \label{fig:precision}
  \vspace*{-2mm}
\end{figure}

The gradual guarantees state that if a program with more precise
labels type-checks and is accepted by the runtime semantics of the
gradual type system, then the same program with less precise labels
is also accepted by the gradual type system. To establish these
guarantees, we define a precision relation between labels, expressions, and
commands. The precision relations are shown in Fig.~\ref{fig:precision}. 

Our type system with gradual labels satisfies the static gradual
guarantee.  The dynamic gradual guarantee is also ensured by our calculus,
i.e., if a command takes a step under a store and $\pc$ stack, then a
less precise command can also take a step under a less precise store
and $\pc$ stack.
\iffull{
  The proofs are shown in Appendix~\ref{sec:app-sgg} and~\ref{sec:app-dgg}.}

\begin{thm}[Static Guarantee]
If~ $\Gamma_1; \glab_1 \vdash c_1$, $\Gamma_1\sqsubseteq\Gamma_2$, 
  $\glab_1\sqsubseteq\glab_2$, and $c_1\sqsubseteq c_2$, 
then $\Gamma_2; \glab_2 \vdash c_2$. 
\end{thm}

\begin{thm}[Dynamic Guarantee]
  \label{thm:dgg}
  If~ $\kappa_1, \delta_1 \sepidx{} c_1 \stackrel{{\alpha_1}}{\stepsto}
  \kappa'_1,\delta'_1\sepidx{} c'_1$ and 
  $\kappa_1, \delta_1 \sepidx{} c_1 \sqsubseteq \kappa_2, \delta_2 \sepidx{} c_2$, 
  then $\kappa_2, \delta_2 \sepidx{} c_2 \stackrel{{\alpha_2}}{\stepsto}
  \kappa'_2,\delta'_2\sepidx{} c'_2$ such that
  $\kappa'_1, \delta_1' \sepidx{} c_1' \sqsubseteq\kappa'_2, \delta_2' \sepidx{} c_2'$ 
  and $\alpha_1 = \alpha_2$.
\end{thm}

\section{Discussion}
\label{sec:discussion}

\Paragraph{Monitor comparison}
Our monitor is a hybrid monitor.  The differences between our monitor
and traditional hybrid
monitors (c.f.,~\cite{russo2010,moore2011,buiras2015,hedin2015,bedford2017})
are that (1), we update the memory before executing the branch, while
other hybrid monitors update the labels at the merge points; (2), our
monitor will not upgrade variables with fixed static labels (the
monitor will abort) and only refine label intervals for dynamically
labeled variables.
Fig.~\ref{fig:mon-comp} highlights the differences between our
monitor and other information flow monitors using the example program
in Listing~\ref{motvEg} and a two point lattice ($L\labless H$).
We show two cases for our monitor, differing
in $y$'s security label (?, or $L$). Note that
hybrid monitors including ours have the same behavior regardless of $x$'s value.
When $y$ has a dynamic label, our
monitor is more precise than NSU, as precise as
permissive-upgrade~\cite{austin2010}, and less precise than a traditional hybrid
monitor~\cite{russo2010}.  Ours can be as precise as a traditional hybrid monitor
if $z$'s initial value is $([L, H]\etrue)^?$. This would allow
us to refine $z$'s label interval, and only abort at the output. When
$y$ has a static label, or an effectively static label (i.e., $([L,
  L]\etrue)^L$), our monitor is the least precise and will abort at the
first branch. This rigidity is due to our decision to not update
variables' label-intervals, except by refinement.

\begin{figure*}[tbp]
  \small
  \begin{center}
      \begin{tabular}{|l|c|c|c|c|c|c|} 
\hline
 $y =\etrue^L$, $z =\etrue^L$  & $x =\efalse^H$         & \multicolumn{2}{c|}{$x =\etrue^H$}           &                 & \multicolumn{2}{c|}{$\wge$ }                                   \\ 
\cline{1-4}\cline{6-7}
\multicolumn{1}{|c|}{Program}                        & NSU/Permissive         & NSU                                                   & Permissive                      & Hybrid          & $y = [L,H]\etrue^?$           & $y = [L,L]\etrue^L$            \\ 
\hline
 $\eif$ $x$                    & branch not taken       & branch taken~                                         & branch taken                    &                 & $y \uparrow [H,H]$            & try $y\uparrow [H,H]$, abort   \\ 
 $\ethen$ $y := \efalse^L$     &                        & $\pc=H$, abort                                        & $y \uparrow P$ & $y \uparrow H$  &                               & \multicolumn{1}{c|}{}           \\ 
 $\eif$ $y$                    & branch taken           & \multicolumn{1}{l|}{}                                 & abort                           &                 & try $z\uparrow [H,H]$, abort  & \multicolumn{1}{c|}{}           \\ 
 $\ethen$ $z := \efalse^L$     &                        & \multicolumn{1}{l|}{}                                  & \multicolumn{1}{l|}{}           & $z\uparrow H$   & \multicolumn{1}{l|}{}          & \multicolumn{1}{l|}{}           \\
$\eoutput(L, z)$               & $\eoutput(L,\efalse)$  & \multicolumn{1}{l|}{}                                  & \multicolumn{1}{l|}{}           & abort           & \multicolumn{1}{l|}{}          & \multicolumn{1}{l|}{}           \\\hline
\end{tabular}
\end{center}
\caption{Comparison of monitor behavior. $y \uparrow \lab$ denotes
  monitor's attempt to update $y$'s label (interval).}
\label{fig:mon-comp}
    \vspace*{-2mm}
\end{figure*}


\Paragraph{Implicit leaks manifested in noninterference proofs}
Let's revisit the example at the end of
Section~\ref{sec:monitor-semantics} to see how the implicit
leak manifests in the paired execution and why it leads to our current
design; where we do not use the write sets in the
$\eif$ statements and simply refine the label-intervals. Because $x$
is $H$, it's initialized with a paired value.
\[
\begin{array}{lcl}
  \delta & = & x\mapsto\epair{ [H, H]\etrue^H}{ [H, H]\efalse^H},\\
         &  & y\mapsto [L, H] \etrue^?,\ z\mapsto [L, L] \etrue^L \\
  c_1 & = & \eif\ x\ \ethen\ y := [L, H]\efalse^?\ \eelse\ \m{skip} \\
  c_2& = & \eif\ y\ \ethen\ z := [L,L]\efalse^L\ \eelse\ \m{skip}
\end{array}
\]
\[
\begin{array}{rl}
[L, L]\; L,\ \delta &\sepidx{} c_1; c_2 \stepsto \\
\cdots &\sepidx{} \eif\ \epair{ [H, H]\etrue^H}{ [H, H]\efalse^H}\  \ethen\\
& ~~~~~~~y := [L, H]\efalse^?\ \eelse\ \m{skip}; c_2 \stepsto \\ 
\cdots &\sepidx{}  \epair{\cdots, y := [L, H]\efalse^?; c_2}{\cdots, \m{skip}; c_2} \stepsto
\end{array}
\]
The variable $y$ is updated only in the left branch. To prove
soundness and completeness of the paired semantics, the two executions
should be independent. Therefore, we try to update $y$ in the 
store as $\epair{[H, H]\efalse}{ [L, H] \etrue}^?$. However, this pair
is not well-formed because pairs are only well-typed if both intervals
are in $H$. Clearly, the right branch of $y$ does not satisfy this
requirement. Therefore, we cannot prove preservation for the
assignment case. For preservation to succeed, we would need to refine
the right branch to be $[H, H] \etrue$ when assigning to $y$ in the
left execution. But then the two executions are no longer independent,
which breaks soundness (i.e., the projected execution is not
guaranteed to make progress or stay in the same state). 

With these constraints in place, we need to refine $y$ before 
the branch, which ultimately leads to our final design.


\section{Related Work}
\label{sec:related}

\Paragraph{Static Information Flow Type Systems}
Quite a few type-systems have been
proposed to statically enforce noninterference by annotating variables
with labels. Volpano et al.~\cite{volpano1996} present the first
type-system with information flow labels that satisfies a variant of
noninterference, also known as termination-insensitive
noninterference. 
If all variables are annotated with concrete security labels, our
  type system behaves the same as a flow-insensitive information flow
  type system. Being a gradual type system, we can additionally accept
  programs with no security labels and enforce termination-insensitive
  noninterference at runtime. Our formalization borrows the
proof-technique from FlowML, presented by Pottier and
Simonet~\cite{pottier2002}, for enforcing noninterference using
pairs. 

\Paragraph{Static type systems that resemble gradual typing}
JFlow~\cite{myers1999} (and later Jif) includes polymorphic
  labels, for which programmers can specify the upper bound of a
  polymorphic label. Polymorphic labels are essentially 
  (bounded) universally quantified labels and these labels are
  instantiated by concrete labels at runtime. There is no label
  refinement associated with polymorphic labels. Jif also allows
  run-time labels (also called dynamic labels). These are runtime
  representation of label objects that users can generate and perform
  tests on. This is not to be confused with the dynamic label (?) in
  gradual typing. Runtime labels do not mean unknown security labels,
  nor are they refined at runtime. Finally, label inference is a
  widely used compile-time algorithm to reduce programmers' annotation
  burdens. A flow-insensitive type system with the most powerful
  inference algorithm is less permission than our system.  Ill-typed
  programs, rejected by a type system, can be accepted by our type
  system. Their safety is ensured by our runtime monitors.  

\Paragraph{Purely dynamic monitors}
Dynamic approaches use a runtime monitor to track the flow of
information through the program. The labels are mostly flow-sensitive
in nature. Austin and Flanagan~\cite{austin2009} present a
purely dynamic information flow monitoring approach that disallows
assignments to public values in secret contexts. Our monitor semantics
follows a similar approach to prevent information leaks at
runtime. Subsequent work presents approaches to make the analysis more
permissive and amenable to dynamic
languages~\cite{austin2010,hedin2012,plas2014,inlining-js}.  A recent
paper shows the equivalence between coarse-grained and fined-grained
dynamic monitors~\cite{fine-coarse-popl19}.  Detailed comparisons
  between our monitor and dynamic monitors can be found in
  Fig.~\ref{fig:mon-comp}.  

\Paragraph{Hybrid monitors}
To leverage the benefits of static and dynamic approaches for
precision and permissiveness, researchers have also proposed hybrid
approaches to enforce
noninterference~\cite{chandra2007,russo2010,moore2011,just2011,buiras2015,hedin2015,hybrid-csf15,bedford2017}. We
demonstrate that the hybrid monitoring approach is suitable for
generating runtime behavior of gradual types that rely on refining
label intervals. Gradual typing has the added benefit of allowing
programmers to reject ill-typed programs.  As shown in
  Fig.~\ref{fig:mon-comp}, our monitor aborts earlier than a
  typical hybrid monitor because of the lack of support for
  label updating. We support termination-insensitive 
  noninterference, while others support progress-sensitive 
  noninterference~\cite{hybrid-csf15,bedford2017}.


\Paragraph{Gradual information flow type systems}
More closely related to our work are works on gradual security
types. Disney and Flanagan~\cite{disney2011} study gradual security
types for a pure lambda calculus, and Fennell and
Thiemann~\cite{fennell2013} present a gradual type system for a
calculus with ML-style references. However, these works are based on
adding explicit programmer-provided checks and casts to the
code. Fennell and Thiemann~\cite{fennell2016} extend their prior work
to object-oriented programs in a flow-sensitive setting for a
Java-like language. They use a hybrid approach to perform effect
analysis that upgrades the labels of variables similar to the
write set used in our analysis. At runtime, these systems cast the
dynamic label to a fixed label, rather than a set of possible labels,
and the monitors updates labels of memory locations, which we do not do.
On the other hand, our approach has fixed gradual labels and refines
only the label-intervals associated with the value to satisfy dynamic
gradual guarantee. More recently, Toro et al.~\cite{toro2018}
presented a type-driven gradual type system for a higher-order
language with references based on abstract gradual
typing~\cite{garcia2016}. Their formalization satisfies the static
gradual guarantee, but sacrifices the dynamic gradual guarantee for
noninterference. They briefly discuss the idea of using hybrid
approaches and faceted evaluation for regaining the dynamic gradual
guarantee. The language presented in this paper is simpler than their
language but has mutable global variables and hence, a similar
issue with proving noninterference while satisfying the dynamic
gradual guarantee.

\textit{GLIO}~\cite{gifc-lics2020} presents another
interpretation of gradual information flow types that enjoys both
noninterference and gradual guarantees. \textit{GLIO} is the most
expressive among the above-mentioned projects; it includes
higher-order functions, general references, coarse-grained
information flow control, and first-class labels. \textit{GLIO}'s monitor
decides the concrete label for a dynamically labeled reference at
allocation time. While avoiding problems stemmed from refining label
intervals, the concrete label results in a less permissive
approach.

Our work explores yet another design space of gradual information
  flow types and highlights the necessity of a hybrid approach for a
  system that refines label intervals to ensure both noninterference
  and the gradual guarantees.
Extending our static type system to include higher-order
  functions and references would require a precise static analysis to
  determine the write set accurately, 
   as pointed out by prior
  work~\cite{moore2011}.  This is common for hybrid
  approaches. For instance, LJGS~\cite{fennell2016} uses a
  sophisticated points-to analysis. Moore and Chong have identified
  sufficient conditions for safely incorporating memory abstractions
  and static analyses into a hybrid information-flow
  monitor~\cite{moore2011}. An interesting future direction is to
  investigate such conditions and abstractions for a higher-order
  language. Another possible direction for handling languages with
  first-class functions and references can be using the 
  ideas proposed by Nielson et al.~\cite{nielson1999} and
  Foster et al.~\cite{foster2002}, who use
  regions and side-effect analysis to determine aliases.

\section{Conclusion}
We presented a gradual information flow type system for
a simple imperative language that enforces
termination-insensitive noninterference and ensures the gradual 
guarantee at the same time. 
We demonstrated that our hybrid monitor can stop
 implicit flows by refining the labels for references in
  the write-sets of both branches,
 regardless of which branch is taken.
The non-conventional proof
technique of noninterference that we used
helps us identify the conditions for ensuring the gradual guarantees.

\section*{Acknowledgment}
We would like to thank the anonymous reviewers for their
insightful comments and feedback. This work was supported
in part by the National Science Foundation via grant
CNS1704542 and the CyLab Presidential Fellowship at
Carnegie Mellon University.   

\bibliographystyle{IEEEtran}
\bibliography{main-full}

\appendix
\subsection{From $\wg$ to $\wge$}
\label{app:trans-lemmas}
\begin{lem}
  \label{lem:i-sub-g}
  If $\iota = \gamma(g)$, then $\iota \sqsubseteq \gamma(g)$
\end{lem}
\begin{proof}
  \begin{description}
  \item [Case:] $g = ?$\\
    $\iota = [\bot, \top]$ and $\gamma(g) = [\bot, \top]$. \\
    By the definition of $\labless$, $\bot \labless \bot$ and $\top \labless \top$.\\
    Thus by the definition of $\sqsubseteq$, $\iota \sqsubseteq \gamma(g)$
  \item [Case:] $g = \ell$\\
    $\iota = [\ell, \ell]$ and $\gamma(g) = [\ell, \ell]$. \\
    By the definition of $\labless$, $\ell \labless \ell$ and $\ell \labless \ell$.\\
    Thus by the definition of $\sqsubseteq$, $\iota \sqsubseteq \gamma(g)$
  \end{description}      
\end{proof}

\begin{lem}[Interval refine]\label{lem:interval-refine}
~~
$\forall i\in\{1,2\}$, $\refineof(\iota_1, \iota_2) \sqsubseteq \iota_i$
\end{lem}
\begin{proofsketch} By examining the definitions of the operations.
\end{proofsketch}

\begin{lem}[Intersect refine]\label{lem:bowtie-refine}
~~
$\forall i\in\{1,2\}$, $\iota_1\bowtie \iota_2 \sqsubseteq \iota_i$

\end{lem}
\begin{proofsketch} By examining the definitions of the operations.
\end{proofsketch}

\begin{lem}[Join refine]\label{lem:join-refine}
$\forall i\in\{1,2\}$, $\iota_i\sqsubseteq \gamma(\glab_i)$ imply
$\iota_1\labjoin\iota_2 \sqsubseteq \gamma(\glab_1\cjoin\glab_2)$.
\end{lem}
\begin{proofsketch} By examining the definitions of the operations.
\end{proofsketch}

\begin{lem}[Evidence subtyping]\label{lem:evd-subtyping}
 $E \vdash \glab_1\csubtp \glab_2$, $\iota\sqsubseteq \gamma(\glab_1)$ imply 
$\iota \bowtie E \sqsubseteq \gamma(\glab_2)$.
\end{lem}
\begin{proof}
\begin{tabbing}
\\\quad\=
Let $E = (\iota_1, \iota_2)$. By inversion, 
  $\iota_2 \sqsubseteq \gamma(g_2)$,
\\\>By Lemma~\ref{lem:interval-refine}, $\iota \bowtie E\sqsubseteq \iota_2$
\\\>By $\sqsubseteq$ is transitive, $\iota \sqsubseteq \gamma(g_2)$,
\end{tabbing}
\end{proof}

\begin{lem} 
  \label{lem:translation}
  If $\Gamma \vdash e : U$ and $\Gamma \vdash e \leadsto e' : U$, then $\Gamma \vdash e' : U$
\end{lem}
\begin{proof}
  By induction on the expression typing derivation.
  \begin{description}
  \item [Case:] \rulename{Bool, Int}
    \begin{tabbing}
      By Lemma~\ref{lem:i-sub-g}, $\gamma(g) \sqsubseteq \gamma(g)$.\\
      By \rulename{T-Bool, T-Int} and \rulename{G-Bool, G-Int},\\
      the conclusion holds
    \end{tabbing}
  \item [Case:] \rulename{Var}
    \begin{tabbing}
      By definition of \rulename{T-Var} and \rulename{G-Var}
    \end{tabbing}
  \item [Case:] \rulename{Bop}
    \begin{tabbing}
      $
      \inferrule*{
        \forall i\in\{1,2\}, ~\Gamma\vdash e_i: \tau^{\glab_i}
        \\ \glab = \glab_1\cjoin\glab_2
      }{ 
        \Gamma  \vdash e_1 \bop e_2 :  \tau^\glab
      }
      $
      \\By IH,
      \\\,\,\,\,\=(1) \,\,\,\,\= $\Gamma \vdash e_i' : \tau^{\glab_i}$
      \\By (1), \rulename{T-Bool} and \rulename{G-Bool}, the conclusion holds
    \end{tabbing}
  \item [Case:] \rulename{Cast}
    \begin{tabbing}
      $\inferrule*{
        \Gamma  \vdash e : U' \\
        U' \csubtp U
      }{ 
        \Gamma  \vdash e::U : U
      }$
      \\Assume $U = \tau^\glab$.
      \\ By assumption
      \\\quad \=(1)\quad\=  $(\iota_1, \iota_2) = \refineof(\gamma(\glab'), \gamma(\glab))$ 
      \\ By IH,
      \\\,\,\,\,\>(2) \,\,\,\,\= $\Gamma \vdash e' : \tau^{\glab'}$ 
      \\ By Lemma~\ref{lem:interval-refine} and $\sqsubseteq$ is transitive
       \\\>(3) \> $(\iota_1, \iota_2)\vdash \tau^{\glab'}\csubtp \tau^\glab$
      \\By (2) and (3), 
      subtyping and $\sqsubseteq$, $\Gamma \vdash (\iota_1, \iota_2)^\glab e' : \tau^\glab$
    \end{tabbing}
  \end{description}
\end{proof}

\begin{lem}
  \label{lem:wtset-trans}
  If ~$\Gamma;\gpc \vdash c$ and $\Gamma\;\gpc \vdash c \leadsto c'$, then $\wtsetof(c) = \wtsetof(c')$
\end{lem}
\begin{proofsketch}
  By induction on the command typing derivation. Follows from the definition of $\wtsetof$ for \rulename{Skip, Assign, Out} and additionally uses the IH for \rulename{Seq, If, While}
\end{proofsketch}

\begin{thm}
  \label{thm:translation}
  If ~$\Gamma;\gpc \vdash c$ and $\Gamma;\gpc \vdash c \leadsto c'$, then $\forall \iota_\pc \sqsubseteq \gamma(\gpc).~ \Gamma;\iota_\pc\;\gpc \vdash c'$
\end{thm}
\begin{proof}
  By induction on the typing derivation - $\Gamma;\gpc \vdash c$.
  \begin{description}
  \item [Case:] \rulename{Skip}
    \begin{tabbing}
      By definition of \rulename{T-Skip, G-Skip}
    \end{tabbing}
  \item [Case:] \rulename{Seq}
    \begin{tabbing}
      $\inferrule*{ 
        \Gamma;\; \gpc \vdash c_1
        \\   \Gamma;\; \gpc \vdash c_2
      }{ 
        \Gamma; \;\gpc  \vdash c_1; c_2
      }
      $
      \\By IH and inversion of \rulename{T-Seq},
      \\\,\,\,\,\=(1) \,\,\,\,\= $\Gamma;\iota_\pc\;\gpc \vdash c_1'$ and $\Gamma;\iota_\pc\;\gpc \vdash c_2'$
      \\such that $\iota_\pc \sqsubseteq \gamma(\gpc)$
      \\By (1), \rulename{G-Seq}, the conclusion holds
    \end{tabbing}
  \item [Case:] \rulename{Assign}
    \begin{tabbing}
      By Lemma~\ref{lem:translation},
      \\\,\,\,\,\=(1) \,\,\,\,\= $\Gamma \vdash x : \tau^{\glab}$ and $\Gamma \vdash e' : \tau^{\glab'}$
      \\By assumption,
      \\\>(2)\>$\iota_\pc \sqsubseteq \gamma(\gpc)$ 
      \\By inversion of \rulename{Assign},
      \\\>(3)\> $\gpc \clabless \glab$ and $\glab'\clabless\glab$
    \\ By assumption
      \\\ \>(4)\>  $(\iota_1, \iota_2) = \refineof(\gamma(\glab'), \gamma(\glab))$ 
      \\ By Lemma~\ref{lem:interval-refine} and $\sqsubseteq$ is transitive
       \\\>(5) \> $(\iota_1, \iota_2)\vdash \tau^{\glab'}\csubtp \tau^\glab$
       \\By (1) and (5),
       \\\>(6)\>$\Gamma \vdash (\iota_1, \iota_2)^\glab e' : \tau^\glab$
      \\By (1), (6),      the conclusion holds
    \end{tabbing}
  \item [Case:] \rulename{Out} This proof is similar to \rulename{Assign}.
  \item [Case:] \rulename{If} (similar for \rulename{While})
    \begin{tabbing}
      $\inferrule*[right=If]{
        \Gamma \vdash e :  \tbool^\glab 
        \\    \Gamma ; \;\gpc\cjoin \glab \vdash c_1 
        \\ \Gamma ; \;\gpc\cjoin \glab \vdash c_2
      }{ 
        \Gamma ; \;\gpc \vdash 
        \eif \; e\ \ethen\  c_1\ \eelse\  c_2 
      }$
      \\By Lemma~\ref{lem:translation},
      \\\,\,\,\,\=(1) \,\,\,\,\= $\Gamma \vdash e' : \tbool^{\glab}$
      \\By IH,
      \\\>(2)\>$\forall \iota_\pc \sqsubseteq \gamma(\gpc). ~\Gamma;\iota_\pc \labjoin \iota_c \;\gpc\cjoin \glab \vdash c_i$ \\
      such that $\iota_c = \gamma(\glab)$ and $i \in \{1,2\}$
      \\By inversion of \rulename{T-If}, 
      \\\>(3)\>$X = \wtsetof(c_1')\cup \wtsetof(c_2')$
      \\By (1), (2), (3), and Lemma~\ref{lem:i-sub-g}, the conclusion holds
    \end{tabbing}
  \end{description}  
\end{proof}

\subsection{Operational Semantics}
\label{sec:app-monitor-semantics}
Figures~\ref{fig:app-mon-semantics-exp} and~\ref{fig:app-mon-semantics-cmd} summarize the operational semantics of the monitor.
Those rules use auxiliary definitions in Figure~\ref{fig:app-store-aux-single}.
  

  

\begin{figure}[htbp]
    \flushleft
    \noindent{\framebox{$ \delta \sepidx{} e \evalsto v$}}
    \begin{mathpar}
      \inferrule*[right=M-Const]{
      }{
        \delta\sepidx{} (\iota\, u)^\glab \evalsto  (\iota\, u)^\glab
      }
      \and
      \inferrule*[right=M-Var]{
      }{
        \delta\sepidx{} x \evalsto  \delta(x)
      }
      \and
      \inferrule*[right=M-Bop]{
        \forall i \in \{1,2\},~\delta\sepidx{} e_i \evalsto (\iota_i\, u_i)^{\glab_i} \\
        \iota = \iota_1\labjoin\iota_2 \\ \glab = \glab_1\cjoin\glab_2 \\ u = u_1\bop u_2
      }{
        \delta\sepidx{} e_1\bop e_2 \evalsto   (\iota\;u)^{\glab}
      }
      \and
      \inferrule*[right=M-Cast]{
        \delta\sepidx{} e \evalsto (\iota\, u)^{\glab} \\
        \iota' = \iota \bowtie E
      }{
        \delta\sepidx{} E^{\glab'}\; e
        \evalsto  (\iota'\; u)^{\glab'}
      }
      \and
      \inferrule*[right=M-Cast-Err]{
        \delta\sepidx{} e \evalsto (\iota\, u)^{\glab} \\
        \iota \bowtie E = \eundef
      }{
        \delta\sepidx{} E^{\glab'}\; e
        \evalsto  \eabort
      }
    \end{mathpar}
    \caption{Monitor semantics for expressions}
    \label{fig:app-mon-semantics-exp}
\end{figure}

\begin{figure*}[b]
    \flushleft
    \noindent{\framebox{$ \kappa,\delta\sepidx{} c \stackrel{\alpha}{\stepsto} \kappa',\delta'\sepidx{} c'$}}
    \begin{mathpar}
      \inferrule*[right=M-Seq]{
        \kappa,\delta\sepidx{} c_1  \stackrel{\alpha}{\stepsto}  \kappa',\delta' \sepidx{} c'_1
      }{
        \kappa,\delta\sepidx{} c_1; c_2
        \stackrel{\alpha}{\stepsto}  \kappa',\delta'\sepidx{} c'_1; c_2
      }
      \and
      \inferrule*[right=M-Pc]{
        \kappa,\delta\sepidx{} c 
        \stackrel{\alpha}{\stepsto}  \kappa',\delta'\sepidx{}c'
      }{
        \kappa\rhd \iota_\pc\;\gpc,\delta\sepidx{} \{c\} 
        \stackrel{\alpha}{\stepsto}  \kappa'\rhd\iota_\pc\;\gpc,\delta'\sepidx{}\{c'\}
      }
      \and
      \inferrule*[right=M-Pop]{
      }{
        \iota_\pc\;\gpc\rhd\kappa,\delta\sepidx{} \{\eskip\} 
        \stepsto  \kappa,\delta\sepidx{} \eskip
      }
      \and
      \inferrule*[right=M-Skip]{
      }{
        \kappa,\delta\sepidx{} \eskip; c
        \stepsto  \kappa,\delta\sepidx{} c
      }
      \and
      \inferrule*[right=M-Assign]{
        \delta\sepidx{} e \evalsto v 
        \\ v' = \reflvof{}(\iota_\pc, v)
        \\ v'' = \updval(\labof(\delta(x)), v')
      }{
        \iota_\pc\;\gpc,\delta \sepidx{} x:= e
        \stepsto  \iota_\pc\;\gpc, 
        \delta[x\mapsto v'']\sepidx{} \eskip
      }
      \and
       \inferrule*[right=M-Out]{
        \delta\sepidx{} e \evalsto v 
        \\ v' = \reflvof{}(\iota_\pc, v)
        \\ v'' = \updval(\labof(\delta(x)), v')
      }{
        \iota_\pc\;\gpc,\delta \sepidx{} \eoutput(\lab, e)
        \stackrel{(\lab, v'')}{\stepsto}  \iota_\pc\;\gpc, 
        \delta\sepidx{} \eskip
      }
      \and
      \inferrule*[right=M-If]{
        \\ \iota'_\pc = \iota_\pc\labjoin \iota
        \\ \gpc' =\gpc\cjoin\glab \\
        c_i = c_1~\mbox{if}~b = \etrue\\
        c_i = c_2~\mbox{if}~b = \efalse
      }{
        \iota_\pc\;\gpc,\delta\sepidx{} \eif\; (\iota\;b)^\glab\ \ethen\ c_1\ \eelse\ c_2
        \stepsto  \iota'_\pc\;\gpc'\rhd \iota_\pc\;\gpc, \delta \sepidx{}  \{c_i\}
      }
      %
      \and
      \inferrule*[right=M-If-Refine]{
        \delta\sepidx{} e \evalsto v \\
      \delta' = \rflof{}(\delta, X, \iota_\pc\labjoin \labof(v))
      }{
        \iota_\pc\;\gpc,\delta\sepidx{} \eif^X\;e\ \ethen\ c_1\ \eelse\ c_2
        \stepsto  \iota_\pc\;\gpc, \delta' \sepidx{} \eif\;v\ \ethen\ c_1\ \eelse\ c_2
      }
      \and
      \inferrule*[right=M-While]{
      }{
        \iota_\pc\;\gpc,\delta\sepidx{} \ewhile^X\;e\ \edo\ c
        \stepsto   \iota_\pc\;\gpc, \delta \sepidx{}
        \eif^X\;e\ \ethen\ (c; \ewhile^X\;e\ \edo\ c)\ \eelse\ \m{skip}
      }
      \and
      \inferrule*[right=M-Seq-Err]{
        \kappa,\delta\sepidx{} c_1   \stepsto  \eabort
      }{
        \kappa,\delta\sepidx{} c_1;c_2
        \stepsto  \eabort
      }
      \and
      \inferrule*[right=M-If-Refine-Err]{
        \delta\sepidx{} e \evalsto v \\
        \rflof{}(\delta, X, \iota_\pc\labjoin \labof(v)) = \eundef
      }{
        \iota_\pc\;\gpc,\delta\sepidx{} \eif^X\;e\ \ethen\ c_1\ \eelse\ c_2
        \stepsto  \eabort
      }
      \and
      \inferrule*[right=M-Assign-Err]{
        \delta\sepidx{} e \evalsto v \\
        \reflvof{}(\iota_\pc, v) = \eundef
      }{
        \iota_\pc\;\gpc,\delta \sepidx{} x:= e
        \stepsto  \eabort
      }
       \and
      \inferrule*[right=M-Assign-Err2]{
        \delta\sepidx{} e \evalsto v \\
        v'= \reflvof{}(\iota_\pc, v) \\
        \updval(\labof(\delta(x)),v') = \eundef
      }{
        \iota_\pc\;\gpc,\delta \sepidx{} x:= e
        \stepsto  \eabort
      }
      
      \and
      \inferrule*[right=M-Out-Err]{
        \delta\sepidx{} e \evalsto v\\
        \reflvof{}(\iota_\pc, v) = \eundef 
      }{
        \iota_\pc\;\gpc,\delta \sepidx{} \eoutput(\lab, e)
        \stackrel{}{\stepsto}  \eabort
      }
 \and
      \inferrule*[right=M-Out-Err2]{
        \delta\sepidx{} e \evalsto v\\
         \reflvof{}(\iota_\pc, v) = v' \\
             \updval([\lab,\lab],v') = \eundef
      }{
        \iota_\pc\;\gpc,\delta \sepidx{} \eoutput(\lab, e)
        \stackrel{}{\stepsto}  \eabort
      }
      
    \end{mathpar}
    \caption{Monitor semantics for commands}
    \label{fig:app-mon-semantics-cmd}
\end{figure*}

\begin{figure}[h]

    \begin{mathpar}
      \inferrule*{ 
        \refineof(\iota_c,\iota) = (\_,\iota')
      }{
        \reflvof{}(\iota_c, (\iota\;u)^\glab) = (\iota'\;u)^\glab
      }
      \and
      \inferrule*{ 
        \refineof(\iota_c,\iota) = \eundef
      }{
        \reflvof{}(\iota_c, (\iota\;u)^\glab) = \eundef
      }
      \and
      \inferrule*{ }{
        \rflof{}(\delta, \cdot, \iota) = \delta
      }
      \and 
      \inferrule*{  
        \delta' = \rflof{}(\delta, X, \iota)
        \\ v' = \reflvof{}(\iota, v) \neq \eundef
      }{
        \rflof{}((\delta, x\mapsto v), (X, x), \iota) 
        = \delta', x \mapsto v'
      }
      \and 
      \inferrule*{  
        \reflvof{}(\iota, v) = \eundef
      }{
        \rflof{}((\delta, x\mapsto v), (X, x), \iota) 
        = \eundef
      }
      \and
      \and
       \inferrule*
      {
         \lab_{1l}\labjoin\lab_{2l} \labless \lab_{1r}
      }{
        \newlab([\lab_{1l}, \lab_{1r}], [\lab_{2l}, \lab_{2r}]) =
         [\lab_{1l}\labjoin\lab_{2l}, \lab_{1r}]
      }
      \and
       \inferrule*
      {
         \lab_{1l}\labjoin\lab_{2l} \nlabless \lab_{1r}
      }{
        \newlab([\lab_{1l}, \lab_{1r}], [\lab_{2l}, \lab_{2r}]) = \eundef
      }
    \end{mathpar}
    
  \[
 \updval(\iota_o, (\iota_n\, u_n)^\glab) =  (\newlab(\iota_o, \iota_n)\; u_n)^\glab 
\]
    \caption{Label-interval refinement operations for the monitor}
    \label{fig:app-store-aux-single}
\end{figure}


\subsection{Paired Execution}
\label{sec:app-paired}
Figure~\ref{fig:app-pair-semantics-c} shows the paired semantic rules,
which use auxiliary definitions in
Figures~\ref{fig:reflvof-pair},~\ref{fig:app-pair-lab-op},~\ref{fig:app-rdupd},~\ref{fig:app-pair-semantics-e}.
\begin{figure}[h]
  \flushleft
  \textbf{$\reflvof{i}~i\in\{1,2\}$: }
  \begin{mathpar}
    \inferrule{ 
      \refineof(\iota_c,\iota) = (\_,\iota')
      \\  \iota_i = \iota'
      \\\\ \iota_j=\iota
      \\ \{i,j\} = \{1, 2\}
    }{
      \reflvof{i}(\iota_c, (\iota\;u)^\glab) = \epair{\iota_1\;u}{\iota_2\; u}^\glab
    }
    \and
    \inferrule{ 
      \refineof(\iota_c,\iota_i) = (\_,\iota'_i)
      \\\\ \iota'_j = \iota_j
      \\ \{i,j\} = \{1, 2\}
    }{
      \reflvof{i}(\iota_c, \epair{\iota_1\;u_1}{\iota_2\;u_2}^\glab)
      = \epair{\iota'_1\;u_1}{\iota'_2\;u_2}^\glab
    }
    \and
    %
    \inferrule{ 
      \refineof(\iota_c,\iota) = \eundef
    }{
      \reflvof{i}(\iota_c, (\iota\;u)^\glab) = \eundef
    }
    \and
    \inferrule{ 
      \refineof(\iota_c,\iota_1) = \eundef \vee         \refineof(\iota_{c},\iota_2) = \eundef
    }{
      \reflvof{i}(\iota_c, \epair{\iota_1\;u_1}{\iota_2\;u_2}^\glab)
      = \eundef
    }
    %
  \end{mathpar}

  \flushleft
  \textbf{$\reflvof{}$: }
  \begin{mathpar}
    \inferrule{ 
      \refineof(\iota_c,\iota) = (\_,\iota')\\
    }{
      \reflvof{}(\iota_c, (\iota\;u)^\glab) = (\iota'\;u)^\glab
    }
    \and
    \inferrule{ 
      \forall i\in\{1,2\}. \,
      \refineof(\iota_{c},\iota_i) = (\_,\iota'_i)
    }{
      \reflvof{}(\iota_c, \epair{\iota_1\;u_1}{\iota_2\;u_2}^\glab)
      = \epair{\iota'_1\;u_1}{\iota'_2\;u_2}^\glab
    }
    \and
    \inferrule{ 
      \refineof(\iota_i,\iota) = (\_,\iota'_i)
    }{
      \reflvof{}(\epair{\iota_1}{\iota_2}, (\iota\;u)^\glab) = 
      \epair{\iota'_1\;u}{\iota'_2\;u}^\glab
    }
    \and
    \inferrule{ 
      \forall i\in\{1,2\}. \,
      \refineof(\iota_{ci},\iota_i) = (\_,\iota'_i)
    }{
      \reflvof{}(\epair{\iota_{c1}}{\iota_{c2}}, \epair{\iota_1\;u_1}{\iota_2\;u_2}^\glab)
      = \epair{\iota'_1\;u_1}{\iota'_2\;u_2}^\glab
    }
  \end{mathpar}
  \begin{mathpar}
    \inferrule{ 
      \refineof(\iota_c,\iota) = \eundef
    }{
      \reflvof{}(\iota_c, (\iota\;u)^\glab) = \eundef
    }
    \and
    \inferrule{ 
      \exists i\in\{1,2\}. \,
      \refineof(\iota_{c},\iota_i) = \eundef
    }{
      \reflvof{}(\iota_c, \epair{\iota_1\;u_1}{\iota_2\;u_2}^\glab)
      = \eundef
    }
    \and
    \inferrule{ 
      \exists i\in\{1,2\}. \,
      \refineof(\iota_i,\iota) = \eundef
    }{
      \reflvof{}(\epair{\iota_1}{\iota_2}, (\iota\;u)^\glab) = 
      \eundef
    }
    \and
    \inferrule{ 
      \exists i\in\{1,2\}. \,
      \refineof(\iota_{ci},\iota_i) = \eundef
    }{
      \reflvof{}(\epair{\iota_{c1}}{\iota_{c2}}, \epair{\iota_1\;u_1}{\iota_2\;u_2}^\glab)
      = \eundef
    }
  \end{mathpar}
    \caption{ $\reflvof{}$ for paired executions}
   \label{fig:reflvof-pair}
 \end{figure}

\begin{figure}[h]
  \flushleft
  \begin{mathpar}
    \inferrule{ }{
      \rflof{}(\delta, \cdot, \iota) = \delta
    }
    \and
    \inferrule{  
      \delta' = \rflof{}(\delta, X, \iota)
      \\ v' = \reflvof{}(\iota, v)
    }{
      \rflof{}((\delta, x\mapsto v), (X, x), \iota) 
      = \delta', x \mapsto v'
    }
    \and
    \inferrule{  
      \delta' = \rflof{}(\delta, X, \iota)
      \\ \reflvof{}(\iota, v) = \eundef
    }{
      \rflof{}((\delta, x\mapsto v), (X, x), \iota) 
      = \eundef
    }
  \end{mathpar}
  \begin{mathpar}
    \inferrule*{ }{
      \rflof{i}(\delta, \cdot, \Pi) = \delta
    }
    \and
    \inferrule*{  
      \delta' = \rflof{i}(\delta, X, \Pi)
      \\ v' = \reflvof{i}(\Pi, v)
    }{
      \rflof{i}((\delta, x^g\mapsto v), (X, x), \Pi) 
      = \delta', x^g\mapsto v'
    }
    \and
    \inferrule*{  
      \delta' = \rflof{i}(\delta, X, \Pi) \\
      \reflvof{i}(\Pi, v) = \eundef
    }{
      \rflof{i}((\delta, x^g\mapsto v), (X, x), \Pi) 
      = \eundef
    }
  \end{mathpar}
  \caption{$\rflof{}$ for  paired executions}
  \label{fig:app-pair-lab-op}
\end{figure}

\begin{figure}[h]
  
  \flushleft
  \textbf{Value and interval projections: }
  \begin{mathpar}
    \inferrule*{}{
      \proj{\epair{\iota_1\;u_1}{\iota_2\;u_2}^\glab}{i} = (\iota_i\;u_i)^\glab
      \quad \proj{(\iota\;u)^\glab}{i} = (\iota\;u)^\glab
      \quad \proj{\epair{\iota_1}{\iota_2}}{i} = \iota_i
      \quad \proj{\iota}{i} = \iota
      \quad \text{where} \; i = \{1, 2\}
    }
  \end{mathpar}
  \flushleft    
  \textbf{Read operations: }
  \begin{mathpar}
    \inferrule*{}{\eread ~v = v}
    \and
    \inferrule*{}{\eread_1 ~v = \proj{v}{1}}
    \and
    \inferrule*{}{\eread_2 ~v = \proj{v}{2}}
  \end{mathpar}
  \flushleft
  \textbf{Write operations:}


  \begin{mathpar}
    \inferrule*{ 
      v_n=\epair{\iota_1\; u_1}{\iota_2\; u_2}^\glab
      \\ \forall i\in{1,2},  \newlab(\iota_o, \iota_i) = \iota'_i
    }{
      \updval~\iota_o~v_n  = \epair{\iota'_1\; u_1}{\iota'_2\; u_2}^\glab
    }
    \and
    \inferrule*{ 
      \newlab(\iota_o, \iota_n) = \iota
    }{
      \updval~\iota_o~(\iota_n\;u_n)^\glab  = (\iota\;u_n)^\glab
    }
    \and 
    \inferrule*{
      v_n=\epair{\iota_1\; u_1}{\iota_2\; u_2}^\glab
      \\ \exists j\in[1,2], \newlab(\iota_o, \iota_j) = \eundef 
    }{
      \updval~\iota_o~v_n  = \eundef
    }
    \and
    \inferrule*{ 
      \newlab(\iota_o, \iota_n) = \eundef
    }{
      \updval~\iota_o~(\iota_n\;u_n)^\glab  = \eundef
    }
  \end{mathpar}
\end{figure}
  \begin{figure}[h]
  \flushleft
  \textbf{Simple write/output operations: }
   \begin{mathpar}
    \inferrule*{ 
      \labof(v_o) = \epair{\iota_1}{\iota_2}
      ~\mbox{or}~ v_n=\epair{\iota'_1\; u_1}{\iota'_2\; u_2}^\glab
     \\\forall i\in{1,2},  \newlab(\proj{\labof(v_o)}{i}, \proj{\labof(v_n)}{i}) = \iota''_i
    }{
      \eupdate~v_o~v_n  = \epair{\iota''_1\; u_1}{\iota''_2\; u_2}^\glab
    }
    \and
    \inferrule*{ 
        \newlab(\iota_o, \iota_n) = \iota
    }{
      \eupdate~(\iota_o\;u_o)^\glab~(\iota_n\;u_n)^\glab  = (\iota\;u_n)^\glab
    }
    \and
    \inferrule*{ 
     \labof(v_o) = \epair{\iota_1}{\iota_2}
      ~\mbox{or}~ v_n=\epair{\iota'_1\; u_1}{\iota'_2\; u_2}^\glab
     \\\\\exists i\in{1,2},  \newlab(\proj{\labof(v_o)}{i}, \proj{\labof(v_n)}{i}) = \eundef
    }{
      \eupdate~v_o~v_n  = \eundef
    }
    \and
    \inferrule*{ 
      \newlab(\iota_o, \iota_n) = \eundef
    }{
      \eupdate~(\iota_o\;u_o)^\glab~(\iota_n\;u_n)^\glab  = \eundef
    }
    \and
    \inferrule*{
      \proj{v_n}{1} = (\iota_1\; u_n)^{\glab'}
      \\ \newlab(\proj{\labof(v_o)}{1}, \iota_1) = \iota'_1
      \\ \proj{v_o}{2} = (\iota_2\; u_2)^{\glab}}{
      \eupdate_1~v_o~(\iota_n\; u_n)^\glab  = \epair{\iota'_1\; u_n}{\iota_2\; u_2}^\glab
    }
    \and
     \inferrule*{
      \proj{v_n}{2} = (\iota_2\; u_n)^{\glab'}
      \\ \newlab(\proj{\labof(v_o)}{2}, \iota_2) = \iota'_2
      \\ \proj{v_o}{1} = (\iota_1\; u_1)^{\glab}}{
      \eupdate_2~v_o~(\iota_n\; u_n)^\glab  = \epair{\iota_1\; u_1}{\iota'_2\; u_n}^\glab
    }
     \and
     \inferrule*{
      \proj{v_n}{1} = (\iota_1\; u_n)^{\glab'}
      \\ \newlab(\proj{\labof(v_o)}{1}, \iota_1) = \eundef
      }{
      \eupdate_1~v_o~(\iota_n\; u_n)^\glab  = \eundef
    }
    \and
     \inferrule*{
      \proj{v_n}{2} = (\iota_2\; u_n)^{\glab'}
      \\ \newlab(\proj{\labof(v_o)}{2}, \iota_2) = \eundef
}{
      \eupdate_2~v_o~(\iota_n\; u_n)^\glab  = \eundef
    }
   \end{mathpar}
   
  \flushleft
  \textbf{Binary operations: }
  \begin{mathpar}
    \inferrule*{
      v_1 = \epair {\iota_1\; u_1}{\iota'_1\; u'_1}^{\glab_1} \\\\
      v_2 = \epair {\iota_2\; u_2}{\iota'_2\; u'_2}^{\glab_2} \\\\
      \iota = (\iota_1\labjoin\iota_2) \\ \iota' = (\iota'_1\labjoin\iota'_2) \\\\
      u = (u_1\bop u_2) \\\\
      u' = (u'_1\bop u'_2) \\\\
      \glab = \glab_1\cjoin\glab_2 \\\\ v = \epair {\iota\; u}{\iota'\; u'}^{\glab} \\
    }{
      v_1 \bop v_2 = v
    }
    \and
    \inferrule*{
      v_i = \epair {\iota_i\; u_i}{\iota'_i\; u'_i}^{\glab_i} \\\\
      v_j = (\iota_j\; u_j)^{\glab_j} \\\\
      \iota = (\iota_i\labjoin\iota_j) \\     \iota' = (\iota'_i\labjoin\iota_j) \\\\
      u = (u_i\bop u_j) \\\\ u' = (u'_i\bop u_j) \\\\
      \glab = \glab_i\cjoin\glab_j \\ \{i,j\} = \{1,2\} \\\\ v = \epair {\iota\; u}{\iota'\; u'}^{\glab} \\
    }{
      v_1 \bop v_2 = v
    }
    \and
    \inferrule*{
      v_1 = (\iota_1\; u_1)^{\glab_1} \\\\
      v_2 = (\iota_2\; u_2)^{\glab_2} \\\\
      \iota = (\iota_1\labjoin\iota_2) \\\\ 
      u = (u_1\bop u_2) \\\\
      \glab = \glab_1\cjoin\glab_2 \\\\ v = (\iota\; u)^{\glab} \\
    }{
      v_1 \bop v_2 = v
    }
  \end{mathpar}
  \textbf{Cast:}
  \begin{mathpar}
    \inferrule*{ 
      v = (\iota\, u)^{\glab}  \\
      \iota' = \iota \bowtie E
    }{
      (E, \glab')\rhd v = (\iota'\; u)^{\glab'}
    }
    \and
    \inferrule*{ 
      v = \epair{\iota_1\;u_1}{\iota_2\;u_2}^\glab \\
      \forall i\in\{1,2\},~
      \iota'_i = \iota_i \bowtie E
    }{
      (E, \glab')\rhd v = \epair{\iota'_1\;u_1}{\iota'_2\;u_2}^{\glab'}
    }
    \and
    \inferrule*{ 
      v = (\iota\, u)^{\glab}  \\
      \iota \bowtie E = \eundef
    }{
      (E, \glab')\rhd v = \eundef
    }
    \and
    \inferrule*{ 
      v = \epair{\iota_1\;u_1}{\iota_2\;u_2}^\glab \\
      \exists i\in\{1,2\},~ \iota_i \bowtie E = \eundef
    }{
      (E, \glab')\rhd v = \eundef
    }
  \end{mathpar}
  \caption{Operations with pairs}
  \label{fig:app-rdupd}
  
\end{figure}

\begin{figure}[ht]
  \flushleft
  \textbf{Expression Semantics: }
  \noindent{\framebox{$ \delta \sepidx{i} e \evalsto v$}}
  \begin{mathpar}
    \inferrule*[right=P-Const]{
    }{
      \delta\sepidx{i} (\iota\, u)^\glab \evalsto  (\iota\, u)^\glab
    }
    \and
    \inferrule*[right=P-Var]{
    }{
      \delta\sepidx{i} x \evalsto  \eread_i\ \delta(x)
    }
    \and
    \inferrule*[right=P-Cast]{
      \delta\sepidx{i} e \evalsto v 
      \\ v' = (E, \glab)\rhd v
    }{
      \delta\sepidx{i} E^{\glab}~ e 
      \evalsto  v'
    }
    \and
    \inferrule*[right=P-Cast-Err]{
      \delta\sepidx{i} e \evalsto v 
      \\ (E, \glab)\rhd v = \eundef
    }{
      \delta\sepidx{i} E^{\glab}~ e 
      \evalsto  \eabort
    }
    \and
    \inferrule*[right=P-Bop]{
      \delta\sepidx{i} e_1 \evalsto v_1 \\
      \delta\sepidx{i} e_2 \evalsto v_2 \\
    }{
      \delta\sepidx{i} e_1\bop e_2 \evalsto v_1\bop v_2 
    }
  \end{mathpar}
  \caption{Operational semantics for expression evaluation in paired executions}
  \label{fig:app-pair-semantics-e}
  
\end{figure}

\begin{figure}[htbp]
  
  \flushleft
  \textbf{Command Semantics: } 
  \noindent{\framebox{$ \kappa,\delta \sepidx{i} c \stackrel{\alpha}{\stepsto}
      \kappa',\delta'\sepidx{i} c'$}}
  
  \begin{mathpar}
    \inferrule*[right=P-C-Pair]{
      \kappa_i\rhd \iota_\pc\labjoin\iota_i\;\gpc\cjoin\glab,\delta\sepidx{i} c_i 
      \stackrel{\alpha}{\stepsto}  \kappa'_i\rhd \iota_\pc\labjoin\iota_i\;\gpc\cjoin\glab,\delta'\sepidx{i}c'_i
      \\ c_j = c'_j
      \\ \kappa_j = \kappa'_j
      \\ \{i,j\} = \{1,2\}
    }{
      \iota_\pc\;\gpc,\delta\sepidx{} 
      \epair{\kappa_1, \iota_1, c_1}{\kappa_2, \iota_2, c_2}_\glab
      \stackrel{\alpha}{\stepsto}  \iota_\pc\;\gpc,\delta'\sepidx{} 
      \epair{\kappa'_1, \iota_1, c'_1}{\kappa'_2, \iota_2, c'_2}_{\glab}
    }
    \and
    \inferrule*[right=P-C-Pair-Err]{
      \kappa_i\rhd \iota_\pc\labjoin\iota_i\;\gpc\cjoin\glab,\delta\sepidx{i} c_i 
      \stackrel{}{\stepsto}  \eabort
      \\ \{i,j\} = \{1,2\}
    }{
      \iota_\pc\;\gpc,\delta\sepidx{} 
      \epair{\kappa_1, \iota_1, c_1}{\kappa_2, \iota_2, c_2}_\glab
      \stackrel{}{\stepsto}  \eabort
    }
    \and
    \inferrule*[right=P-Lift-If]{
      i = \{1, 2\}
      \\ c_j = c_1~\mbox{if}~u_1 = \etrue
      \\ c_j = c_2~\mbox{if}~u_1 = \efalse
      \\ c_k = c_1~\mbox{if}~u_2 = \etrue
      \\ c_k = c_2~\mbox{if}~u_2 = \efalse
    }{
      \iota_\pc\;\gpc,\delta\sepidx{}
      \eif\; \epair{\iota_1\;u_1}{\iota_2\;u_2}^\glab
      \ \ethen\ c_1\ \eelse\ c_2
      \stepsto  \iota_\pc\;\gpc,\delta\sepidx{} 
      \epair{\emptyset, \iota_{1}, c_j}{
        \emptyset, \iota_{2}, c_k}_g
    }
    \and
    \inferrule*[right=P-Skip-Pair]{
    }{
      \iota_\pc\; \gpc, \delta\sepidx{} \epair{\emptyset, \iota_1, \eskip}{\emptyset, \iota_2, \eskip}_g
      \stepsto \iota_\pc\; \gpc,\delta\sepidx{} \eskip
    }
    \and
    \inferrule*[right=P-Assign]{
        \delta\sepidx{i} e \evalsto v \\
        \\ v' = \reflvof{}(\iota_\pc, v)
      }{
        \iota_\pc\;\gpc,\delta \sepidx{i} x:= e
        \stepsto \\ \iota_\pc\;\gpc, 
        \delta[x\mapsto \eupdate_i\ \delta(x)\  v'] \sepidx{i} \m{skip}
      }
    \and
    \inferrule*[right=P-Assign-Err]{
      \delta\sepidx{i} e \evalsto v \\
        \reflvof{}(\iota_\pc, v) = \eundef 
    }{
      \iota_\pc\;\gpc,\delta \sepidx{i} x:= e
      \stepsto  \eabort
    }
   \and
    \inferrule*[right=P-Assign-Err2]{
      \delta\sepidx{i} e \evalsto v \\
      \reflvof{}(\iota_\pc, v) = v' \\
      \eupdate_i\ \delta(x)\  v' = \eundef
    }{
      \iota_\pc\;\gpc,\delta \sepidx{i} x:= e
      \stepsto  \eabort
    } 
    \and
    \inferrule*[right=P-Seq]{
      \kappa,\delta\sepidx{i} c_1  \stackrel{\alpha}{\stepsto}  \kappa',\delta' \sepidx{i} c'_1
    }{
      \kappa,\delta\sepidx{i} c_1; c_2
      \stackrel{\alpha}{\stepsto}  \kappa',\delta'\sepidx{i} c'_1; c_2
    }
    \and
    \inferrule*[right=P-Seq-Err]{
      \kappa,\delta\sepidx{i} c_1   \stepsto  \eabort
    }{
      \kappa,\delta\sepidx{i} c_1;c_2
      \stepsto  \eabort
    }
    \and
    \inferrule*[right=P-Pop]{
    }{
      \iota\;\glab\rhd\kappa,\delta\sepidx{i} \{\eskip\} 
      \stepsto  \kappa,\delta\sepidx{i} \eskip
    }
    \and
    \inferrule*[right=P-Skip]{
    }{
      \kappa,\delta\sepidx{i} \eskip; c
      \stepsto  \kappa,\delta\sepidx{i} c
    }
    \and
    \inferrule*[right=P-Pc]{
      \kappa,\delta\sepidx{i} c 
      \stackrel{\alpha}{\stepsto}  \kappa',\delta'\sepidx{i} c'
    }{
      \kappa\rhd \iota_\pc\;\gpc,\delta\sepidx{i} \{c\} 
      \stackrel{\alpha}{\stepsto}  \kappa'\rhd\iota_\pc\;\gpc,\delta'\sepidx{i}\{c'\}
    }
    \and
    \inferrule*[right=P-Pc-Err]{
      \kappa,\delta\sepidx{i} c 
      \stackrel{}{\stepsto}  \eabort
    }{
      \kappa\rhd \iota_\pc\;\gpc,\delta\sepidx{i} \{c\} 
      \stackrel{}{\stepsto}  \eabort
    }
    \and
    \inferrule*[right=P-If-Refine]{
      \delta\sepidx{i} e \evalsto v \\
      \rflof{i}(\delta, X, \iota_\pc\labjoin\labof(v)) = \eundef
    }{
      \iota_\pc\;\gpc,\delta\sepidx{i} \eif^X\;e\ \ethen\ c_1\ \eelse\ c_2
      \stepsto  \eabort
    }
    \and
     \inferrule*[right=P-Out]{
         \delta\sepidx{i} e \evalsto v 
         \\  v' = \reflvof{}(\iota_\pc, v)
         \\  v'' = \updval\ [\lab,\lab]\  v'
       }{
         \iota_\pc\;\gpc,\delta \sepidx{i} \eoutput(\lab, e)
         \stackrel{(i, \lab, v'')}{\stepsto} \iota_\pc\;\gpc, 
         \delta\sepidx{i} \eskip
         }
    \and
    \inferrule*[right=P-Out-Err]{
      \delta\sepidx{i} e \evalsto v \\
      \reflvof{}(\iota_\pc, v) = \eundef 
    }{
      \iota_\pc\;\gpc,\delta \sepidx{i} \eoutput(\lab, e)
      \stackrel{}{\stepsto}  \eabort
    }
    \and
    \inferrule*[right=P-Out-Err2]{
      \delta\sepidx{i} e \evalsto v \\
      \reflvof{}(\iota_\pc, v) = v' \\
      \updval\ [\lab,\lab]\  v' = \eundef
    }{
      \iota_\pc\;\gpc,\delta \sepidx{i} \eoutput(\lab, e)
      \stackrel{}{\stepsto}  \eabort
    }
    \and
    \inferrule*[right=P-If]{
      \\ \iota'_\pc = \iota_\pc\labjoin \iota
      \\ \gpc' =\gpc\cjoin\glab
      \\ c_j = c_1~\mbox{if}~b = \etrue
      \\ c_j = c_2~\mbox{if}~b = \efalse
    }{
      \iota_\pc\;\gpc,\delta\sepidx{i} \eif\; (\iota\;b)^\glab\ \ethen\ c_1\ \eelse\ c_2
      \stepsto  \iota'_\pc\;\gpc'\rhd \iota_\pc\;\gpc, \delta \sepidx{i}  \{c_j\}
    }
    \and
    \inferrule*[right=P-While]{
    }{
      \iota_\pc\;\gpc,\delta\sepidx{i} \ewhile^X\;e\ \edo\ c
      \stepsto   \iota_\pc\;\gpc, \delta \sepidx{i}
      \eif^X\;e\ \ethen\ (c; \ewhile^X\;e\ \edo\ c)\ \eelse\ \m{skip}
    }
  \end{mathpar}
  \caption{Operational semantics for paired executions}
  \label{fig:app-pair-semantics-c}
  
\end{figure}  

\subsection{Write-set}
\label{sec:app-write-set}
The function $\wtsetof$ is defined as:
\[
  \begin{array}{l}
  \mi{WtSet}(x) = \{x\} \\
  \mi{WtSet}(\m{skip}) = \emptyset \\
  \mi{WtSet}(\m{output}(\lab, e)) = \emptyset \\
  \mi{WtSet}(\{c\}) = \mi{WtSet}(c)\\
  \mi{WtSet}(E\; e) = \mi{WtSet}(e) \\ 
  \mi{WtSet}(c_1; c_2) = \mi{WtSet}(c_1)\cup  \mi{WtSet}(c_2) \\
  \mi{WtSet}(\ewhile^X \; e\; \edo\ c) = \mi{WtSet}(c)\\  
  \mi{WtSet}(\epair{\kappa_1,\iota_1,c_1}{\kappa_2, \iota_2, c_2}_g) = \mi{WtSet}(c_1)\cup  \mi{WtSet}(c_2)\\
  \mi{WtSet}(x := e_2) = \mi{WtSet}(x)\\
    \mi{WtSet}(\eif^X\; e\; \ethen\ c_1\ \eelse\; c_2) = \mi{WtSet}(c_1)\cup \mi{WtSet}(c_2)
  \end{array}
\]

\subsection{Well-formedness}
\label{sec:app-wf}
\begin{enumerate}
\item $\vdash v\ \m{wf}$, if
  \begin{enumerate}
  \item $v = \epair{v_1}{v_2}$, then $\forall i \in \{1,2\}. v_i = (\iota_i\;u_i)$
  \item $v = (\iota\;u)^\glab$
  \end{enumerate}
\item $\vdash \delta\ \m{wf}$, if $\forall x \in \delta, \vdash \delta(x)\ \m{wf}$
\item $\vdash\delta \sepidx{i} e\ \m{wf}$ for $i \in \{\cdot, 1, 2\}$ if $\vdash \delta\ \m{wf}$
\item $\vdash c\ \m{wf}$, when the following hold:
  \begin{enumerate}
  \item if $c=\epair{\kappa_1,\iota_1,c_1}{\kappa_2,\iota_2,c_2}_g$, then $\vdash c_1\ \m{wf}$, $\vdash c_2\ \m{wf}$, and 
    $c_1$ and $c_2$ do not contain pairs
  \item if $c =\; \eif\; e\; \ethen\; c_1\; \eelse\; c_2$, then $\vdash c_1\ \m{wf}$, $\vdash c_2\ \m{wf}$, and $c_1$ and $c_2$ do not contain pairs or braces
  \item if $c =\; \ewhile^X\; e\; \edo\; c$, then $\vdash c\ \m{wf}$ and $c$ does not contain pairs or braces
  \item if $c =\; c_1; c_2$ then $\vdash c_1\ \m{wf}$, $\vdash c_2\ \m{wf}$ and $c_2$ does not contain pairs or braces
  \item if $c =\{c_1\}$, then $\vdash c_1\ \m{wf}$
  \end{enumerate}
\item $\vdash \kappa, \delta \sepidx{i} c\ \m{wf}$ for $i \in \{\cdot, 1, 2\}$) if all of the following hold
  \begin{enumerate}
  \item $\vdash c\ \m{wf}$ and $\vdash \delta\ \m{wf}$
  \item if $i\in\{1,2\}$, then 
    $c$ does not contain pairs 
  \end{enumerate}
\end{enumerate}

\subsection{Theorems and Proofs for Well-formedness}
\begin{lem}[Restrict refines]\label{lem:restrict-refine}
~~
$\newlab(\iota_1, \iota_2) \sqsubseteq \iota_1$
\end{lem}
\begin{proof} By definition of $\newlab$, if $\iota_1 = [\ell_1, \ell_1']$ and $\iota_2 = [\ell_2, \ell_2']$, then
  $\newlab(\iota_1, \iota_2) = [\ell_1 \labjoin \ell_2, \ell_1']$.
  As $\ell_1 \labless (\ell_1 \labjoin \ell_2)$ and $\ell_1' \labless \ell_1'$, the conclusion holds.
\end{proof}

\begin{lem}
  \label{lem:wf-pair-pres}
  If $\forall i \in \{1,2\}, \mathcal{E} :: \kappa,\delta \sepidx{i}  c \stackrel{\alpha}{\stepsto} \kappa',\delta'\sepidx{i} c'$
  and $\vdash \kappa, \delta \sepidx{i} c\ \m{wf}$, 
  then
  \begin{itemize}
  \item[(a).] $c'$ does not contain pairs and
  \item[(b).] $\wtsetof(c') \subseteq \wtsetof(c)$
  \end{itemize}
\end{lem}
\begin{proof}
  By induction on the structure of $\mathcal{E}$.
  \begin{description}
  \item [$c'$]: doesn't contain pairs:
    Follows from well-formedness definition for most rules.  \rulename{P-Seq, P-Pc} use the IH additionally.
  \item[$\wtsetof(c') \subseteq \wtsetof(c)$:]
    \begin{tabbing}
      \\\quad\= \rulename{P-Seq}\qquad\qquad\= By IH, $\wtsetof(c_1') \subseteq \wtsetof(c_1)$.\\
      \>\>Thus by definition of $\wtsetof$, \\$\wtsetof(c_1';c_2) \subseteq \wtsetof(c_1;c_2)$\\
      \>\rulename{P-Pc}\> By IH, $\wtsetof(c') \subseteq \wtsetof(c)$\\
      \>\rulename{P-Pop}\> $\wtsetof(\eskip) = \wtsetof(\eskip)$\\
      \>\rulename{P-Skip}\> $\wtsetof(\eskip) \subseteq$\\$\wtsetof(\eskip) \cup \wtsetof(c)$\\
      \>\rulename{P-Assign}\> $\wtsetof(\eskip) \subseteq \wtsetof(x) $\\
      \>\rulename{P-Out}\> $\wtsetof(\eskip) = \wtsetof(\eoutput)$\\
      \>\rulename{P-If}\> $\wtsetof(c_1) \subseteq \wtsetof(c_1) \cup \wtsetof(c_2)$\\
      \>\rulename{P-If-Refine}\> $\wtsetof(c_1) \cup \wtsetof(c_2) =$\\$ \wtsetof(c_1) \cup \wtsetof(c_2)$\\
      \>\rulename{P-While}\> $\wtsetof(c) \cup \wtsetof(c) \cup$\\$ \wtsetof(\eskip) = \wtsetof(c)$
    \end{tabbing}
  \end{description}
\end{proof}

\begin{lem}
  \label{lem:wf-label-pres}
  If ~$\,\vdash v\ \m{wf}$, $\vdash v'\ \m{wf}$ and~$\,\vdash \delta\ \m{wf}$, then
  \begin{enumerate}
  \item $\forall i \in \{\cdot, 1, 2\}, \iota.$
    $\,\vdash \reflvof{i} (\iota, v)\ \m{wf}$,  $\,\vdash \eupdate_i\ \iota\ v'\ \m{wf}$ and 
  \item $\forall i \in \{\cdot, 1, 2\}, \iota, X. \,\vdash \rflof{i} (\delta, X, \iota)\ \m{wf}$
  \end{enumerate}
\end{lem}
\begin{proof}
  By examining the respective definitions and induction for $\rflof{}$.
\end{proof}

\begin{lem}[Well-formedness preservation]
  \label{lem:wf-pres}
  If $\kappa,\delta \sepidx{i}  c \stackrel{\trace}{\stepsto^*} \kappa',\delta'\sepidx{i} c'$
  where $\vdash \kappa, \delta \sepidx{i} c\ \m{wf}$, \\
  then $\vdash \kappa', \delta' \sepidx{i} c'\ \m{wf}$
\end{lem}
\begin{proof}
  By induction of the number of steps in sequence. Base case follows by assumption.
  \begin{description}
  \item [Ind:] Holds for $n$ steps; To show for $n+1$ steps, i.e.,
    \\if $\kappa,\delta \sepidx{i}  c \stackrel{\trace}{\stepsto^n} \kappa',\delta'\sepidx{i} c'
          \stackrel{\alpha}{\stepsto}  \kappa'',\delta''\sepidx{i} c''$
    where $\vdash \kappa, \delta \sepidx{i} c\ \m{wf}$, 
    then $\vdash \kappa'', \delta'' \sepidx{i} c''\ \m{wf}$
    \begin{tabbing}
      By IH
      \\\,\,\,\,\=(1) \,\,\,\,\=$\vdash \kappa', \delta' \sepidx{i} c'\ \m{wf}$
      \\By (1), T.S. If $\kappa',\delta' \sepidx{i}  c' \stackrel{\alpha}{\stepsto} \kappa'',\delta''\sepidx{i} c''$
      \\where $\vdash \kappa', \delta' \sepidx{i} c'\ \m{wf}$, 
      \\then $\vdash \kappa'', \delta'' \sepidx{i} c''\ \m{wf}$
      \\Induction over the derivation. The proof follows from
      \\Lemma~\ref{lem:wf-pair-pres}(a) and~\ref{lem:wf-label-pres} in most cases.
      \\\rulename{P-Seq, P-Pc} use the IH additionally and
      \\\rulename{P-C-Pair} uses Lemma~\ref{lem:wf-pair-pres}(b) additionally.
    \end{tabbing}
  \end{description}
\end{proof}

\subsection{Soundness of Paired-Execution}
\label{sec:app-sound}
\begin{lem}
  \label{lem:store-proj-pres}
  $\forall x \in \delta, i \in \{1,2\}$, $\proj{\delta(x)}{i} = \proj{\delta}{i}(x)$ and $\proj{\eread\ \delta(x)}{i} = \eread_i\ \delta(x) = \eread\ \proj{\delta}{i}(x)$
\end{lem}
\begin{proof}
  $\delta(x)$ can be a pair or normal value
  \begin{description}
  \item [Case:] $\delta(x) = (\iota\;u)^\glab$
    \\By store-projection $\proj{\delta}{i}(x) = \proj{(\iota\;u)^\glab}{i} = (\iota\;u)^\glab$
    \\By definition of $\eread$, $\eread\ \delta(x) = (\iota\;u)^\glab$ and
    \\$\eread_i\ \delta(x) = (\iota\;u)^\glab$ and $\eread\ \proj{\delta}{i}(x) = (\iota\;u)^\glab$
    \\By projection of values, the conclusion holds.
  \item [Case:] $\delta(x) = \epair{\iota_1\;u_1}{\iota_2\;u_2}^\glab$
    \\By store-projection definition: $\proj{\delta}{i}(x) = (\iota_i\;u_i)^\glab$
    \\By definition of $\eread$, $\eread\ \delta(x) = \epair{\iota_1\;u_1}{\iota_2\;u_2}^\glab$ and
    \\$\eread_i\ \delta(x) = (\iota_i\;u_i)^\glab$ and $\eread\ \proj{\delta}{i}(x) = (\iota_i\;u_i)^\glab$
    \\By projection of values, $\proj{\delta(x)}{i} = (\iota_i\;u_i)^\glab$ and $\proj{\eread\ \delta(x)}{i} = (\iota_i\;u_i)^\glab$
  \end{description}
\end{proof}

    
    

\begin{lem}
  \label{lem:update-proj-pres}
  $\forall i \in \{1,2\}$, $\proj{\eupdate\ v_o\ v_n}{i} = \eupdate\ \proj{v_o}{i}\ \proj{v_n}{i}$
\end{lem}
\begin{proofsketch}
  By examining the definitions.
  \end{proofsketch}

\begin{lem} 
    \label{lem:refvl-updval-pres}
  $\forall i \in \{1,2\}$, $\proj{\updval{}(\iota, v)}{i} =
  \updval{}(\iota, \proj{v}{i})$.

\end{lem}
\begin{proofsketch}
  By examining the definitions.
  \end{proofsketch}

    

\begin{lem}
  \label{lem:reflv-proj-pres-pi}
  $\forall i \in \{1,2\}$, $\proj{\reflvof{}(\Pi, v)}{i} = \reflvof{}(\proj{\Pi}{i}, \proj{v}{i})$
\end{lem}
\begin{proof}
  $\Pi$ can be a pair of intervals or single interval and $v$ can be a pair or normal value. 
  \begin{description}
  \item [Case:] $\Pi = \epair{\iota_1}{\iota_2}$, $v = \epair{\iota_1'\;u_1}{\iota_2';u_2}^\glab$ 
    \begin{tabbing}
      By projection of values, 
      \\\,\,\,\,\=(1) \,\,\,\,\= $\proj{v}{i} = (\iota_i'\;u_i)^\glab$ and $\proj{\Pi}{i} = \iota_i$
      \\By (1), definition of $\reflvof{}$,
      \\\>(2)\> $\reflvof{}(\proj{\Pi}{i}, \proj{v}{i}) = (\iota_i''\;u_i)^\glab$ where
      \\$\refineof(\iota_i, \iota_i') = (\_, \iota_i'')$
      \\By definition of $\reflvof{}$, 
      \\\>(3)\>$\reflvof{}(\Pi, v) = \epair{\iota_1'''\;u_1}{\iota_2'''\;u_2}^\glab$ where
      \\$\refineof(\iota_i, \iota_i') = (\_, \iota_i''')$
      \\By (2), (3) and projection of values, the conclusion holds
    \end{tabbing}
  \item [Case:] $\Pi = \epair{\iota_1}{\iota_2}$, $v = (\iota\;u)^\glab$ 
    \begin{tabbing}
      By projection of values, 
      \\\,\,\,\,\=(1) \,\,\,\,\= $\proj{v}{i} = (\iota\;u)^\glab$ and $\proj{\Pi}{i} = \iota_i$
      \\By (1), definition of $\reflvof{}$,
      \\\>(2)\> $\reflvof{}(\proj{\Pi}{i}, \proj{v}{i}) = (\iota_i'\;u)^\glab$ where
      \\$\refineof(\iota_i, \iota) = (\_, \iota_i')$
      \\By definition of $\reflvof{}$, 
      \\\>(3)\>$\reflvof{}(\Pi, v) = \epair{\iota_1'\;u}{\iota_2'\;u}^\glab$ where
      \\$\refineof(\iota_i, \iota) = (\_, \iota_i')$
      \\By (2), (3) and projection of values, the conclusion holds
    \end{tabbing}
  \item [Case:] $\Pi = \iota$, $v = \epair{\iota_1\;u_1}{\iota_2;u_2}^\glab$ 
    \begin{tabbing}
      By projection of values, 
      \\\,\,\,\,\=(1) \,\,\,\,\= $\proj{v}{i} = (\iota_i\;u_i)^\glab$ and $\proj{\Pi}{i} = \iota$
      \\By (1), definition of $\reflvof{}$,
      \\\>(2)\> $\reflvof{}(\proj{\Pi}{i}, \proj{v}{i}) = (\iota_i'\;u_i)^\glab$ where
      \\$\refineof(\iota, \iota_i) = (\_, \iota_i')$
      \\By definition of $\reflvof{}$, 
      \\\>(3)\>$\reflvof{}(\Pi, v) = \epair{\iota_1'\;u_1}{\iota_2'\;u_2}^\glab$ where
      \\$\refineof(\iota, \iota_i) = (\_, \iota_i')$
      \\By (2), (3) and projection of values, the conclusion holds
    \end{tabbing}
  \item [Case:] $\Pi = \iota_c$, $v = (\iota\;u)^\glab$ 
    \begin{tabbing}
      By projection of values, 
      \\\,\,\,\,\=(1) \,\,\,\,\= $\proj{v}{i} = (\iota\;u)^\glab$ and $\proj{\Pi}{i} = \iota_c$
      \\By (1), definition of $\reflvof{}$,
      \\\>(2)\> $\reflvof{}(\proj{\Pi}{i}, \proj{v}{i}) = (\iota'\;u)^\glab$ where
      \\$\refineof(\iota_c, \iota) = (\_, \iota')$
      \\By definition of $\reflvof{}$, 
      \\\>(3)\>$\reflvof{}(\Pi, v) = (\iota'\;u)^\glab$ where
      \\$\refineof(\iota_c, \iota) = (\_, \iota')$
      \\By (2), (3) and projection of values, the conclusion holds
    \end{tabbing}
  \end{description}
\end{proof}

\begin{lem}
  \label{lem:reflv-proj-pres-pi-i}
  $\forall \{i,j\} \in \{1,2\}$, $\proj{\reflvof{i}(\Pi, v)}{i} = \reflvof{ }(\proj{\Pi}{i}, \proj{v}{i}) \wedge \proj{\reflvof{i}(\Pi, v)}{j} = \proj{v}{j}$
\end{lem}
\begin{proof}
  $\Pi$ can be a pair of intervals or single interval and $v$ can be a pair or normal value. We show for $i=1, j=2$.  
  \begin{description}
  \item [Case:] $\Pi = \epair{\iota_{c1}}{\iota_{c2}}$, $v = \epair{\iota_1\;u_1}{\iota_2;u_2}^\glab$ 
    \begin{tabbing}
      By projection of values, 
      \\\,\,\,\,\=(1) \,\,\,\,\= $\proj{\Pi}{1} = \iota_{c1}$, $\proj{v}{1} = (\iota_1\;u_1)^\glab$ and $\proj{v}{2} = (\iota_2\;u_2)^\glab$
      \\\>(2)\> Let $\refineof(\iota_{c1}, \iota_1) = (\_, \iota_i')$
      \\By (1), definition of $\reflvof{i}$,
      \\\>(2)\> $\reflvof{1}(\proj{\Pi}{1}, \proj{v}{1}) = \epair{\iota_1'\;u_1}{\iota_1\;u_1}^\glab$
      \\By definition of $\reflvof{i}$, 
      \\\>(3)\>$\reflvof{1}(\Pi, v) = \epair{\iota_1'\;u_1}{\iota_2\;u_2}^\glab$ 
      \\By (1), (2), (3) and projection of values,
      \\\>(4)\>$\proj{\reflvof{1}(\Pi, v)}{1} = \reflvof{1}(\proj{\Pi}{1}, \proj{v}{1})$ and
      \\\>\>$\proj{\reflvof{1}(\Pi, v)}{2} = \proj{v}{2}$
    \end{tabbing}
  \item [Case:] $\Pi = \epair{\iota_{c1}}{\iota_{c2}}$, $v = (\iota\;u)^\glab$ 
    \begin{tabbing}
      By projection of values, 
      \\\,\,\,\,\=(1) \,\,\,\,\= $\proj{\Pi}{1} = \iota_{c1}$, $\proj{v}{1} = (\iota\;u)^\glab$ and $\proj{v}{2} = (\iota\;u)^\glab$
      \\\>(2)\> Let $\refineof(\iota_{c1}, \iota) = (\_, \iota')$
      \\By (1), definition of $\reflvof{i}$,
      \\\>(2)\> $\reflvof{1}(\proj{\Pi}{1}, \proj{v}{1}) = \epair{\iota'\;u}{\iota\;u}^\glab$
      \\By definition of $\reflvof{i}$, 
      \\\>(3)\>$\reflvof{1}(\Pi, v) = \epair{\iota'\;u}{\iota\;u}^\glab$ 
      \\By (1), (2), (3) and projection of values,
      \\\>(4)\>$\proj{\reflvof{1}(\Pi, v)}{1} = \reflvof{1}(\proj{\Pi}{1}, \proj{v}{1})$
      \\\>\>$= (\iota'\;u)^\glab$ and $\proj{\reflvof{1}(\Pi, v)}{2} = \proj{v}{2} = (\iota\;u)^\glab$
    \end{tabbing}
  \item [Case:] $\Pi = \iota$, $v = \epair{\iota_1\;u_1}{\iota_2;u_2}^\glab$ 
    \begin{tabbing}
      By projection of values, 
      \\\,\,\,\,\=(1) \,\,\,\,\= $\proj{\Pi}{1} = \iota$, $\proj{v}{1} = (\iota_1\;u_1)^\glab$ and $\proj{v}{2} = (\iota_2\;u_2)^\glab$
      \\\>(2)\> Let $\refineof(\iota, \iota_1) = (\_, \iota_1')$
      \\By (1), definition of $\reflvof{i}$,
      \\\>(2)\> $\reflvof{1}(\proj{\Pi}{1}, \proj{v}{1}) = \epair{\iota_1'\;u_1}{\iota\;u_1}^\glab$
      \\By definition of $\reflvof{i}$, 
      \\\>(3)\>$\reflvof{1}(\Pi, v) = \epair{\iota_1'\;u_1}{\iota_2\;u_2}^\glab$ 
      \\By (1), (2), (3) and projection of values,
      \\\>(4)\>$\proj{\reflvof{1}(\Pi, v)}{1} = \reflvof{1}(\proj{\Pi}{1}, \proj{v}{1})$
      \\\>\>$ = (\iota_1'\;u_1)^\glab$ and $\proj{\reflvof{1}(\Pi, v)}{2} = \proj{v}{2} = (\iota_2\;u_2)^\glab$
    \end{tabbing}
  \item [Case:] $\Pi = \iota_c$, $v = (\iota\;u)^\glab$ 
    \begin{tabbing}
      By projection of values, 
      \\\,\,\,\,\=(1) \,\,\,\,\= $\proj{\Pi}{1} = \iota_{c}$, $\proj{v}{1} = (\iota\;u)^\glab$ and $\proj{v}{2} = (\iota\;u)^\glab$
      \\\>(2)\> Let $\refineof(\iota_{c}, \iota) = (\_, \iota')$
      \\By (1), definition of $\reflvof{i}$,
      \\\>(2)\> $\reflvof{1}(\proj{\Pi}{1}, \proj{v}{1}) = \epair{\iota'\;u}{\iota\;u}^\glab$
      \\By definition of $\reflvof{i}$, 
      \\\>(3)\>$\reflvof{1}(\Pi, v) = \epair{\iota'\;u}{\iota\;u}^\glab$ 
      \\By (1), (2), (3) and projection of values,
      \\\>(4)\>$\proj{\reflvof{1}(\Pi, v)}{1} = \reflvof{1}(\proj{\Pi}{1}, \proj{v}{1})$
      \\\>\>$ = (\iota'\;u)^\glab$ and $\proj{\reflvof{1}(\Pi, v)}{2} = \proj{v}{2} = (\iota\;u)^\glab$
    \end{tabbing}
  \end{description}
\end{proof}

\begin{lem}
  \label{lem:rflof-proj-pres-pi-i}
  $\forall \delta, X, x. x \in X \implies x \in \delta$, we have $\forall \{i,j\} \in \{1,2\}$, $\proj{\rflof{i}(\delta, X, \Pi)}{i} = \rflof{}(\proj{\delta}{i}, X, \proj{\Pi}{i})$ and $\proj{\rflof{i}(\delta, X, \Pi)}{j} = \proj{\delta}{j}$
\end{lem}
\begin{proof} 
  By induction on the size of $X$ and applying Lemma~\ref{lem:reflv-proj-pres-pi-i}
\end{proof}

\begin{lem}
  \label{lem:rflof-proj-pres-pi}
  $\forall \delta, X. X \subseteq \delta$, we have $\forall i \in \{1,2\}$, $\proj{\rflof{}(\delta, X, \Pi)}{i} = \rflof{}(\proj{\delta}{i}, X, \proj{\Pi}{i})$
\end{lem}
\begin{proof} 
  By induction on the size of $X$ and applying Lemma~\ref{lem:reflv-proj-pres-pi}
\end{proof}

\begin{lem}
  \label{lem:no-pair-value-op}
  If $v$ is not a pair, then
  \begin{enumerate}
  \item $\forall \iota, \iota \bowtie v$ is not a pair
  \item $\forall \iota, \reflvof{}(\iota, v)$ is not a pair
  \item $\forall \iota, \updval{}(\iota, v)$ is not a pair
  \end{enumerate}
\end{lem}
\begin{proof}
  By examining the respective definitions.
\end{proof}
  
\begin{lem}
  \label{lem:expr-store-proj}
  If $\forall i\in \{1, 2\}$, $\mathcal{E} :: ~\delta \sepidx{i} e \evalsto v_i$ and $\vdash \delta \sepidx{i} e\ \m{wf}$, then
  $\proj{\delta}{i} \sepidx{} e \evalsto v_i$ and $v_i$ is not a pair
\end{lem}
\begin{proof}
  By induction on the structure of $\mathcal{E}$.
  \begin{description}
    
  \item [Case:] \rulename{P-Const}
    \begin{tabbing}
      By \rulename{P-Const},
      \\\,\,\,\,\=(1) \,\,\,\,\= $\forall i \in \{1, 2\}$, $\proj{\delta}{i} \sepidx{} (\iota\;u)^\glab \evalsto (\iota\;u)^\glab$
    \end{tabbing}
  \item [Case:] \rulename{P-Var}
    \begin{tabbing}
      By \rulename{P-Var}
      \\\,\,\,\,\=(1) \,\,\,\,\= $\forall i \in \{1, 2\}$, $\proj{\delta}{i} \sepidx{} x \evalsto \eread\proj{\delta}{i}(x)$
      \\$\delta(x)$ is either a pair or normal value:
      \\{\bf Subcase I:} $\delta(x) = \epair{v_1}{v_2}^\glab$
      \\\>By definition of $\eread$ and value projection
      \\\>\>(I1)\quad\= $\eread_i \delta(x) = v_i$
      \\\>By (1), definition of store projection and $\eread$
      \\\>\>(I2)\quad\=$\proj{\delta}{i}(x) = v_i$ and $\eread\; v_i = v_i$
      \\\>By (I1) and (I2), $\eread_i \delta(x) = \eread{} \proj{\delta}{i}(x)$
      \\As $\vdash \delta \sepidx{i} e\ \m{wf}$, $\delta(x)$ cannot have nested pairs.
      \\Hence, $v_i$ is not a pair
      \\{\bf Subcase II:} $\delta(x) = v = (\iota\;u)^\glab$
      \\\>By definition of $\eread$ and value projection
      \\\>\>(II1)\quad\= $\eread_i \delta(x) = v$; $v$ is a normal value and not a pair
      \\\>By (1), definition of store projection and $\eread$
      \\\>\>(II2)\quad\=$\proj{\delta}{i}(x) = v$ and $\eread{}\; v = v$
      \\\>By (II1) and (II2), $\eread_i \delta(x) = \eread \proj{\delta}{i}(x)$
    \end{tabbing}
  \item [Case:] \rulename{P-Bop}
    \begin{tabbing}
      \\\,\,\,\,\=(1) \,\,\,\,\= $\inferrule*{
        \delta\sepidx{i} e_1 \evalsto v_1 \\
        \delta\sepidx{i} e_2 \evalsto v_2 \\
        v = v_1\bop v_2 
      }{
        \delta\sepidx{i} e_1\bop e_2 \evalsto v
      }$
      \\By IH
      \\\>(2)\> $\forall i \in \{1, 2\}$, $\proj{\delta}{i} \sepidx{} e_1 \evalsto v_1$ and $\proj{\delta}{i} \sepidx{} e_2 \evalsto v_2$
      \\\>\> and $v_1$ and $v_2$ are not pairs
      \\By (2), \rulename{P-Bop} and binary operation on values, \\\>\> the conclusion holds
    \end{tabbing}
  \item [Case:] \rulename{P-Cast}
    \begin{tabbing}
      \\\,\,\,\,\=(1) \,\,\,\,\= $\inferrule*{
        \delta\sepidx{i} e \evalsto v \\
        v' = (E,\glab) \rhd v \\
      }{
        \delta\sepidx{i} E^\glab e \evalsto v'
      }$
      \\By IH
      \\\>(2)\> $\forall i \in \{1, 2\}$, $\proj{\delta}{i} \sepidx{} e \evalsto v$ and $v$ is not a pair
      \\\>(3)\> Let $v = (\iota\;u)^{\glab'}$, then $v' = (\iota'\;u)^\glab$ such that \\\>\> $\iota' = \iota \bowtie E$
      \\By (2), (3) and \rulename{P-Cast}, the conclusion holds
    \end{tabbing}
  \end{description}
\end{proof}

\begin{lem}
  \label{lem:kappa-inv}
  If $\kappa, \delta \sepidx{} c \stackrel{\alpha}{\stepsto} \kappa', \delta' \sepidx{} c'$,
  where $\kappa = \bar{\kappa} \rhd \iota_\pc\;\gpc$, 
  then $\exists \bar{\kappa}'$ s.t. $\kappa' = \bar{\kappa}' \rhd \iota_\pc\;\gpc$
\end{lem}
\begin{proof}
  By induction on the structure of the derivation.
\end{proof}

\begin{lem}
  \label{lem:store-proj}
  If $\forall \{i,j\} \in \{1, 2\}$,
  $\kappa, \delta \sepidx{i} c \stackrel{\alpha}{\stepsto} \kappa', \delta' \sepidx{i} c'$, $\vdash \kappa,\delta \sepidx{i} c\ \m{wf}$
  then $\kappa, \proj{\delta}{i} \sepidx{} c \stackrel{\proj{\alpha}{i}}{\stepsto} 
  \kappa', \proj{\delta'}{i} \sepidx{} c'$, $\proj{\delta}{j} = \proj{\delta'}{j}$ and $\proj{\alpha}{j} = \cdot$
\end{lem}
\begin{proof}
  By induction on the structure of command-evaluation derivation.
  \begin{description}
  \item [Case:] \rulename{P-Seq}
    \begin{tabbing}
      $\inferrule*{
        \kappa,\delta\sepidx{i} c_1  \stackrel{\alpha}{\stepsto}  \kappa',\delta' \sepidx{i} c'_1
      }{
        \kappa,\delta\sepidx{i} c_1; c_2
        \stackrel{\alpha}{\stepsto}  \kappa',\delta'\sepidx{i} c'_1; c_2
      }$
      \\By IH 
      \\\,\,\,\,\=(1) \,\,\,\,\= $\forall \{i,j\} \in \{1, 2\}$, 
      $\kappa,\proj{\delta}{i} \sepidx{}  c_1$
      \\\>\> $\stackrel{\proj{\alpha}{i}}{\stepsto} \kappa',\proj{\delta'}{i}\sepidx{} c'_1$,
      \\\>\> $\proj{\delta}{j} = \proj{\delta'}{j}$ and $\proj{\alpha}{j} = \cdot$, 
      \\By (1) and definition of \rulename{P-Seq}
      \\\>(2) \>$\forall \{i,j\} \in \{1, 2\}$, 
      $\kappa,\proj{\delta}{i} \sepidx{}  c_1;c_2 \stackrel{\proj{\alpha}{i}}{\stepsto}$\\\>\> $ \kappa',\proj{\delta'}{i}\sepidx{} c'_1;c_2$
      \\\>\> and $\proj{\delta}{j} = \proj{\delta'}{j}$ and $\proj{\alpha}{j} = \cdot$,
    \end{tabbing}
    
  \item [Case:] \rulename{P-Pc}
    \begin{tabbing}
      $\inferrule*{
        \kappa,\delta\sepidx{i} c 
        \stackrel{\alpha}{\stepsto}  \kappa',\delta'\sepidx{i} c'
      }{
        \kappa\rhd \iota_\pc\;\gpc,\delta\sepidx{i} \{c\} 
        \stackrel{\alpha}{\stepsto}  \kappa'\rhd\iota_\pc\;\gpc,\delta'\sepidx{i}\{c'\}
      }$
      \\By IH 
      \\\,\,\,\,\=(1) \,\,\,\,\= $\forall \{i,j\} \in \{1, 2\}$, 
      $\kappa,\proj{\delta}{i} \sepidx{}  c \stackrel{\proj{\alpha}{i}}{\stepsto} \kappa',\proj{\delta'}{i}\sepidx{} c'$ and \\\>\> $\proj{\delta}{j} = \proj{\delta'}{j}$ and $\proj{\alpha}{j} = \cdot$, 
      \\By (1) and definition of \rulename{P-Pc}, the conclusion holds
    \end{tabbing}
    
  \item [Case:] \rulename{P-Pop, P-Skip, P-If, P-While}
    \begin{tabbing}
      By definition of \rulename{P-Pop, P-Skip, P-If, P-While}
    \end{tabbing}

  \item [Case:] \rulename{P-Assign} 
    \begin{tabbing}
      By Lemma~\ref{lem:expr-store-proj}
      \\\,\,\,\,\=(1) \,\,\,\,\= $\forall i \in \{1, 2\}$, 
      $\proj{\delta}{i} \sepidx{} e \evalsto v_0$ such that $v_0 = v$
      \\\>\> and $v$ is not a pair
      \\By definition of \rulename{P-Assign} and Lemma~\ref{lem:no-pair-value-op},
      \\\>(1a)\> $v' = v'_0$ where $v'_0$ is obtained by
      \\\>\> operations on $v_0$ and is not a pair
      \\By assumption
      \\\>(2)\> $\delta' = \delta[x \mapsto \eupdate_i\ \delta(x)\ v']$, 
      \\By (1), and definition of \rulename{P-Assign},
      \\\>(2)\>$\delta'' = \proj{\delta}{i}[x \mapsto \eupdate\ \proj{\delta}{i}(x)\ v']$,
      \\T.S. $\proj{\delta'}{i} = \delta''$ or $\forall x. \proj{\delta'(x)}{i} = \delta''(x)$
      \\We show for $i = 1, j = 2$, similar for $i=2, j=1$
      \\{\bf Subcase I:} Suppose $i = 1, j = 2$, $v' = (\iota\;u)^\glab$
      \\\>\> and $\delta(x) = \epair{\iota_1\;u_1}{\iota_2\;u_2}^\glab$
      \\\>By store projection
      \\\>\>(I1)\quad\= $\proj{\delta}{1}(x) = (\iota_1\;u_1)^\glab$ and $\proj{\delta}{2}(x) = (\iota_2\;u_2)^\glab$
      \\\>By definition of $\eupdate$ and value projection
      \\\>\>(I2)\quad\= $\delta'(x) = \epair{\iota'\;u}{\iota_2\;u_2}^\glab$
      \\\>\>\> where $\newlab(\iota_1, \iota) = \iota'$ and $\proj{\delta'}{2}(x) = (\iota_2\;u_2)^\glab$
      \\\>By (I1), definition of $\eupdate$
      \\\>\>(I3)\quad\= $\delta''(x) = (\iota''\;u)^\glab$ where
       $\newlab(\iota_1, \iota) = (\_, \iota'')$
      \\\>By (I2) and (I3), $\proj{\delta'(x)}{1} = \delta''(x)$, $\proj{\delta}{2}(x) = \proj{\delta'}{2}(x)$ 
      \\\>\>\> and $\forall y. y \neq x \implies \proj{\delta(y)}{i} = \proj{\delta}{i}(y)$
      \\{\bf Subcase II:} Suppose $i = 1, j = 2$, $v' = (\iota\;u)^\glab$ and $\delta(x) = (\iota_1\;u_1)^\glab$
      \\\>By store projection
      \\\>\>(I1)\quad\= $\proj{\delta}{1}(x) = (\iota_1\;u_1)^\glab$ and $\proj{\delta}{2}(x) = (\iota_1\;u_1)^\glab$
      \\\>By definition of $\eupdate$ and value projection
      \\\>\>(I2)\quad\= $\delta'(x) = \epair{\iota'\;u}{\iota_1\;u_1}^\glab$ where $\newlab(\iota_1, \iota) = (\_, \iota')$ and $\proj{\delta'}{2}(x) = (\iota_1\;u_1)^\glab$
      \\\>By (I1), definition of $\eupdate$
      \\\>\>(I3)\quad\= $\delta''(x) = (\iota''\;u)^\glab$ where $\newlab(\iota_1, \iota) = (\_, \iota'')$
      \\\>By (I2) and (I3), $\proj{\delta'(x)}{1} = \delta''(x)$, $\proj{\delta}{2}(x) = \proj{\delta'}{2}(x)$ and $\forall y. y \neq x \implies \proj{\delta(y)}{i} = \proj{\delta}{i}(y)$
    \end{tabbing}

  \item [Case:] \rulename{P-Out}
    \begin{tabbing}
      By Lemma~\ref{lem:expr-store-proj},
      \\\,\,\,\,\=(1) \,\,\,\,\= $\forall i \in \{1, 2\}$, 
      $\proj{\delta}{i} \sepidx{} e \evalsto v_0$ such that $v_0 = v$ and $v$ is not a pair
      \\By definition of \rulename{P-Out} and Lemma~\ref{lem:no-pair-value-op},
      \\\>(1a)\> $v_1 = v'_1$ where $v'_1$ is obtained by operations on $v_0$ and is not a pair
      \\By definition of trace-projection,
      \\\>(2)\>$\proj{i, \ell, v_1}{i} = (\ell, v_1)$ and $\proj{i, \ell, v_1}{j} = \cdot$ where $i \neq j$
    \end{tabbing}
    
  \item [Case:] \rulename{P-If-Refine}
    \begin{tabbing}
      By Lemma~\ref{lem:expr-store-proj}
      \\\,\,\,\,\=(1) \,\,\,\,\= $\forall i \in \{1, 2\}$, 
      $\proj{\delta}{i} \sepidx{} e \evalsto v$ and $v$ is not a pair,
      \\By Lemma~\ref{lem:rflof-proj-pres-pi}
      \\\>(2)\> $\forall \{i,j\} \in \{1, 2\}$, $\proj{\rflof{i}(\delta, X, \iota_\pc\labjoin\labof(v))}{i} = \rflof{}(\proj{\delta}{i}, X, \iota_\pc\labjoin\labof(v))$
      \\\>\>and $\proj{\rflof{i}(\delta, X, \iota_\pc \labjoin\labof(v))}{j} = \proj{\delta}{j}$
      \\By (1), (2) and definition of \rulename{P-If-Refine}, the conclusion holds
    \end{tabbing}

  \end{description}
\end{proof}

\begin{lem}[Expression soundness] 
  \label{lem:expr-soundness}
  If $\mathcal{E} :: \delta \sepidx{} e \evalsto v$ and $\vdash \delta \sepidx{} e\ \m{wf}$ then $\forall i \in \{1, 2\}$,
  $\proj{\delta \sepidx{} e}{i} \evalsto \proj{v}{i}$
\end{lem}
\begin{proof}
  By induction on the structure of $\mathcal{E}$.
  \begin{description}
    
  \item [Case:] \rulename{P-Const}
    \begin{tabbing}
      By \rulename{P-Const} and projection of values
    \end{tabbing}
    
  \item [Case:] \rulename{P-Var}
    \begin{tabbing}
      By assumption
      \\\,\,\,\,\=(1) \,\,\,\,\= $v = \eread\ \delta(x)$
      \\T.S. $\forall i \in \{1, 2\}$, $\eread\ \proj{\delta}{i}(x) = \proj{\eread\ \delta(x)}{i}$
      \\We show for $i=1$, the proof is similar for $i=2$
      \\{\bf Subcase I:} $\delta(x) = \epair{\iota_1\;u_1}{\iota_2\;u_2}^\glab$
      \\\>By definition of $\eread$
      \\\>\>(I1)\quad\= $\proj{\eread\ \delta(x)}{1} = (\iota_1\;u_1)^\glab$
      \\\>By store-projection definition and definition of $\eread$
      \\\>\>(I2)\quad\= $\eread\ \proj{\delta}{1}(x) = (\iota_1\;u_1)^\glab$
      \\\>By (I1) and (I2), the conclusion holds
      \\{\bf Subcase II:} $\delta(x) = (\iota\;u)^\glab$
      \\\>By definition of $\eread$
      \\\>\>(II1)\quad\= $\proj{\eread\ \delta(x)}{1} = (\iota\;u)^\glab$
      \\\>By store-projection definition and definition of $\eread$
      \\\>\>(II2)\quad\= $\eread\ \proj{\delta}{1}(x) = (\iota\;u)^\glab$
      \\\>By (II1) and (II2), the conclusion holds
    \end{tabbing}
    
  \item [Case:] \rulename{P-Bop}
    \begin{tabbing}
      \,\,\,\,\=(1) \,\,\,\,\= $\inferrule*{
        \delta\sepidx{} e_1 \evalsto v_1 \\
        \delta\sepidx{} e_2 \evalsto v_2 \\
        v = v_1\bop v_2 
      }{
        \delta\sepidx{} e_1\bop e_2 \evalsto v
      }$
      \\By (1) and IH
      \\\>(2)\> $\forall i \in \{1, 2\}$, $\proj{\delta}{i} \sepidx{} e_1 \evalsto \proj{v_1}{i}$ and
      $\proj{\delta}{i} \sepidx{} e_2 \evalsto \proj{v_2}{i}$
      \\T.S. $\forall i \in \{1, 2\}$, $\proj{v_1}{i} \bop \proj{v_2}{i} = \proj{v_1 \bop v_2}{i}$
      \\We show for $i=1$, the proof is similar for $i=2$
      \\{\bf Subcase I:} $v_1 = \epair{\iota_1\;u_1}{\iota_1'\;u_1'}^{\glab_1}$,
      $v_2 = \epair{\iota_2\;u_2}{\iota_2'\;u_2'}^{\glab_2}$
      \\\>By definition of $\bop$,
      \\\>\>(I1)\quad\= $v_1 \bop v_2 = \epair{\iota\;u}{\_}^\glab$ and $\proj{v_1 \bop v_2}{1} = (\iota\;u)^\glab$
      \\\>\>\quad\>      where $\iota = \iota_1 \labjoin \iota_2$, $u = u_1 \bop u_2$ and $\glab = \glab_1 \cjoin \glab_2$
      \\\>By value-projection definition 
      \\\>\>(I2)\quad\= $\proj{v_1}{1} = (\iota_1\;u_1)^\glab$ and $\proj{v_2}{1} = (\iota_2\;u_2)^\glab$
      \\\>By (I2) and definition of $\bop$,
      \\\>\>(I3)\quad\= $\proj{v_1}{1} \bop \proj{v_2}{1} = (\iota\;u)^\glab$ 
      where $\iota = \iota_1 \labjoin \iota_2$, $u = u_1 \bop u_2$ and $\glab = \glab_1 \cjoin \glab_2$
      \\\>By (I1) and (I3), the conclusion holds
      \\{\bf Subcase II:} $v_1 = \epair{\iota_1\;u_1}{\iota_1'\;u_1'}^{\glab_1}$,
      $v_2 = (\iota_2\;u_2)^{\glab_2}$, similar for $v_1 = (\iota_1\;u_1)^{\glab_1}$, $v_2 = \epair{\iota_2\;u_2}{\iota_2'\;u_2'}^{\glab_2}$,
      \\\>By definition of $\bop$,
      \\\>\>(II1)\quad\= $v_1 \bop v_2 = \epair{\iota\;u}{\_}^\glab$ and $\proj{v_1 \bop v_2}{1} = (\iota\;u)^\glab$
      \\\>\>\quad\>      where $\iota = \iota_1 \labjoin \iota_2$, $u = u_1 \bop u_2$ and $\glab = \glab_1 \cjoin \glab_2$
      \\\>By value-projection definition 
      \\\>\>(II2)\quad\= $\proj{v_1}{1} = (\iota_1\;u_1)^\glab$ and $\proj{v_2}{1} = (\iota_2\;u_2)^\glab$
      \\\>By (II2) and definition of $\bop$,
      \\\>\>(II3)\quad\= $\proj{v_1}{1} \bop \proj{v_2}{1} = (\iota\;u)^\glab$ 
      where $\iota = \iota_1 \labjoin \iota_2$, $u = u_1 \bop u_2$ and $\glab = \glab_1 \cjoin \glab_2$
      \\\>By (II1) and (II3), the conclusion holds
      \\{\bf Subcase III:} $v_1 = (\iota_1\;u_1)^{\glab_1}$,
      $v_2 = (\iota_2\;u_2)^{\glab_2}$
      \\\>By definition of $\bop$,
      \\\>\>(III1)\quad\= $v_1 \bop v_2 = (\iota\;u)^\glab$ and $\proj{v_1 \bop v_2}{1} = (\iota\;u)^\glab$
      \\\>\>\quad\>where $\iota = \iota_1 \labjoin \iota_2$, $u = u_1 \bop u_2$ and $\glab = \glab_1 \cjoin \glab_2$
      \\\>By value-projection definition 
      \\\>\>(III2)\quad\= $\proj{v_1}{1} = (\iota_1\;u_1)^\glab$ and $\proj{v_2}{1} = (\iota_2\;u_2)^\glab$
      \\\>By (III2) and definition of $\bop$,
      \\\>\>(III3)\quad\= $\proj{v_1}{1} \bop \proj{v_2}{1} = (\iota\;u)^\glab$ 
      where $\iota = \iota_1 \labjoin \iota_2$, $u = u_1 \bop u_2$ and $\glab = \glab_1 \cjoin \glab_2$
      \\\>By (III1) and (III3), the conclusion holds
    \end{tabbing}
    
  \item [Case:] \rulename{P-Cast}
    \begin{tabbing}
      \,\,\,\,\=(1) \,\,\,\,\= $\inferrule*{\delta\sepidx{} e \evalsto v 
        \\ v' = (E, \glab)\rhd v
      }{
        \delta\sepidx{} E^{\glab} e 
        \evalsto  v'
      }$
      \\By IH and (1)
      \\\>(2)\> $\forall i \in \{1, 2\}$, $\proj{\delta \sepidx{} e}{i} \evalsto \proj{v}{i}$
      \\We show for $i=1$, the proof is similar for $i=2$
      \\{\bf Subcase I:} $v = \epair{\iota_1\;u_1}{\iota_2\;u_2}^{\glab'}$,
      \\\>By definition of $\rhd$ cast operation,
      \\\>\>(I1)\quad\= $v' = \epair{\iota_1'\;u_1}{\iota_2'\;u_2}^\glab$ and $\proj{v'}{1} = (\iota_1'\;u_1)^\glab$ where $\iota_1' = \iota_1 \bowtie E$
      \\\>By value-projection definition and cast operation,
      \\\>\>(I2)\quad\= $\proj{v}{1} = (\iota_1\;u_1)^{\glab'}$ and $v_1'' = (\iota_1''\;u_1)^\glab$ where $\iota_1'' = \iota_1 \bowtie E$
      \\\>By (I1) and (I2), the conclusion holds
      \\{\bf Subcase II:} $v = (\iota\;u)^{\glab'}$,
      \\\>By definition of cast
      \\\>\>(II1)\quad\= $v' = (\iota'\;u)^\glab$ and $\proj{v'}{1} = (\iota'\;u)^\glab$ where $\iota' = \iota \bowtie E$
      \\\>By definition of value-projection and cast 
      \\\>\>(II2)\quad\= $\proj{v}{1} = (\iota\;u)^{\glab'}$ and $v'' = (\iota''\;u)^\glab$ where $\iota'' = \iota \bowtie E$
      \\\>By (II1) and (II2), the conclusion holds
    \end{tabbing}
  \end{description}
\end{proof}

\begin{lem}[Soundness] 
  \label{lem:soundness}
  If $\kappa,\delta \sepidx{}  c \stackrel{\alpha}{\stepsto} \kappa',\delta'\sepidx{} c'$
  where $\vdash \kappa, \delta \sepidx{} c\ \m{wf}$, \\
  then $\forall i \in \{1, 2\}$, 
  $\proj{\kappa,\delta \sepidx{}  c}{i} \stackrel{\proj{\alpha}{i}}{\stepsto} 
  \proj{\kappa',\delta'\sepidx{} c'}{i}$, 
  or $\proj{\kappa,\delta \sepidx{}  c}{i} = \proj{\kappa',\delta'\sepidx{} c'}{i}$ and $\proj{\alpha}{i} = \cdot$
\end{lem}
\begin{proof}

  By induction on the structure of the command derivation.
  \begin{description}
    
  \item [Case:] \rulename{P-Seq}
    \begin{tabbing}
      $\inferrule*{
        \kappa,\delta\sepidx{} c_1  \stackrel{\alpha}{\stepsto}  \kappa',\delta' \sepidx{} c'_1
      }{
        \kappa,\delta\sepidx{} c_1; c_2
        \stackrel{\alpha}{\stepsto}  \kappa',\delta'\sepidx{} c'_1; c_2
      }$
      \\By IH 
      \\\,\,\,\,\=(1) \,\,\,\,\= $\forall i \in \{1, 2\}$, 
      $\proj{\kappa,\delta \sepidx{}  c_1}{i} \stackrel{\proj{\alpha}{i}}{\stepsto} 
      \proj{\kappa',\delta'\sepidx{} c'_1}{i}$, 
      or $\proj{\kappa,\delta \sepidx{}  c_1}{i} = \proj{\kappa',\delta'\sepidx{} c_1'}{i}$
      \\By (1), projection of commands and definition of \rulename{P-Seq}
      \\\>(2) \>$\forall i \in \{1, 2\}$, 
      $\proj{\kappa,\delta \sepidx{}  c_1;c_2}{i} \stackrel{\proj{\alpha}{i}}{\stepsto} 
      \proj{\kappa',\delta'\sepidx{} c_1';c_2}{i}$, 
      or $\proj{\kappa,\delta \sepidx{}  c_1;c_2}{i} = \proj{\kappa',\delta'\sepidx{} c_1';c_2}{i}$  
    \end{tabbing}
    
  \item [Case:] \rulename{P-Pc}
    \begin{tabbing}
      $\inferrule{
        \kappa,\delta\sepidx{} c 
        \stackrel{\alpha}{\stepsto}  \kappa',\delta'\sepidx{} c'
      }{
        \kappa\rhd \iota_\pc\;\gpc,\delta\sepidx{} \{c\} 
        \stackrel{\alpha}{\stepsto}  \kappa'\rhd\iota_\pc\;\gpc,\delta'\sepidx{}\{c'\}
      }$
      \\By IH 
      \\\,\,\,\,\=(1) \,\,\,\,\= $\forall i \in \{1, 2\}$, 
      $\proj{\kappa,\delta \sepidx{}  c}{i} \stackrel{\proj{\alpha}{i}}{\stepsto} 
      \proj{\kappa',\delta'\sepidx{} c'}{i}$, 
      or $\proj{\kappa,\delta \sepidx{}  c}{i} = \proj{\kappa',\delta'\sepidx{} c'}{i}$
      \\By (1), projection of commands and definition of \rulename{P-Pc}
      \\\>(2) \>$\forall i \in \{1, 2\}$, 
      $\proj{\kappa\rhd \iota_\pc\;\gpc,\delta \sepidx{}  \{c\}}{i} \stackrel{\proj{\alpha}{i}}{\stepsto} 
      \proj{\kappa' \rhd \iota_\pc\;\gpc,\delta'\sepidx{} \{c'\}}{i}$\\
      or $\proj{\kappa\rhd \iota_\pc\;\gpc,\delta \sepidx{}  \{c\}}{i} = \proj{\kappa' \rhd \iota_\pc\;\gpc,\delta'\sepidx{} \{c'\}}{i}$  
    \end{tabbing}

  \item [Case:] \rulename{P-Pop, P-Skip, P-If, P-While}
    \begin{tabbing}
      By projection of commands, well-formedness definition and \\
      definition of \rulename{P-Pop, P-Skip, P-If, P-While}
    \end{tabbing}
    
  \item [Case:] \rulename{P-Assign}
    \begin{tabbing}
      By Lemma~\ref{lem:expr-soundness}
      \\\,\,\,\,\=(1) \,\,\,\,\= $\forall i \in \{1, 2\}$, 
      $\proj{\delta \sepidx{} e}{i} \evalsto v_i$ such that $v_i = \proj{v}{i}$,
      \\By projection of interval operations, Lemma~\ref{lem:reflv-proj-pres-pi},
      \\\>(2) \>$\forall i \in \{1, 2\}$, $v' = \reflvof{}(\iota_\pc, v)$ and $v_i' = \reflvof{}(\iota_\pc, \proj{v}{i})$ and $\proj{v'}{i} = v_i'$,
      \\By (2),
      \\\>(3)\>$\proj{\delta'(x)}{i} = \proj{\delta'}{i}(x)$
      \\By (1), projection of commands and definition of \rulename{P-Assign},
      \\\>(4) \>$\delta'(x) = \eupdate\ \delta(x)\ v'$ and $\delta'_i(x) = \eupdate\ \proj{\delta}{i}(x)\ v_i'$; 
      \\By (3), (4) and Lemma~\ref{lem:update-proj-pres}, the conclusion follows
    \end{tabbing}

  \item [Case:] \rulename{P-Out} 
    \begin{tabbing}
      By Lemma~\ref{lem:expr-soundness}
      \\\,\,\,\,\=(1) \,\,\,\,\= $\forall i \in \{1, 2\}$, 
      $\proj{\delta \sepidx{} e}{i} \evalsto \proj{v}{i}$,
      \\By projection of interval operations, Lemmas~\ref{lem:refvl-updval-pres} and~\ref{lem:reflv-proj-pres-pi},
      \\\>(2) \>$\forall i \in \{1, 2\}$, $v' = \reflvof{}(\iota_\pc, v)$, $v'' = \updval([\ell, \ell],v')$
      \\\>\> $v''_i = \proj{\updval([\ell, \ell],v')}{i} = \updval([\ell, \ell],\proj{v'}{i})$,
      and  $\proj{v''}{i} = v_i''$
      \\By (1), (2), projection of commands and traces, and definition of \rulename{P-Out},
      \\\>(3) \>$\forall i \in \{1, 2\}$, 
      $\proj{(\ell, v'')}{i} = (\ell, v''_i)$, 
    \end{tabbing}
    
  \item [Case:] \rulename{P-If-Refine}
    \begin{tabbing}
      By Lemma~\ref{lem:expr-soundness}
      \\\,\,\,\,\=(1) \,\,\,\,\= $\forall i \in \{1, 2\}$, 
      $\proj{\delta \sepidx{} e}{i} \evalsto \proj{v}{i}$,
      \\By Lemma~\ref{lem:rflof-proj-pres-pi}
      \\\>(2)\> $\forall i \in \{1, 2\}$, $\proj{\rflof{}(\delta, X, \iota_\pc\labjoin\labof(v))}{i} = \rflof{}(\proj{\delta}{i}, X, \iota_\pc\labjoin\labof(\proj{v}{i}))$ 
      \\By (1), (2), projection of commands and definition of \rulename{P-If-Refine}, the conclusion holds
    \end{tabbing}

  \item [Case:] \rulename{P-C-Pair} where $c_1$ and $c_2$ are not $\eskip$
    \begin{tabbing}
      We show for $i=1, j=2$. The proof is similar for $i=2, j=1$
      \\By assumption and Lemma~\ref{lem:store-proj}
      \\\,\,\,\,\=(1) \,\,\,\,\=
      $\kappa_1 \rhd (\iota_\pc \labjoin \iota_1)\;(\glab \cjoin \gpc), \proj{\delta}{1} \sepidx{}  c_1 \stackrel{\alpha}{\stepsto} 
      \kappa'_1 \rhd (\iota_\pc \labjoin \iota_1)\;(\glab \cjoin \gpc), \proj{\delta'}{1} \sepidx{}  c_1'$,
      \\\>\>and $\kappa_2 = \kappa_2'$, $c_2 = c_2'$, $\proj{\delta}{2} = \proj{\delta'}{2}$
      \\By projection of commands
      \\\>(2)\> $\proj{\iota_\pc\;\gpc, \delta \sepidx{} \epair{\kappa_1, \iota_1, c_1}{\kappa_2, \iota_2, c_2}_\glab}{1} = 
      \kappa_1\rhd  (\iota_\pc\labjoin\iota_1)\;(\gpc\cjoin\glab) 
      \rhd \iota_\pc\;\gpc, \proj{\delta}{1} \sepidx{} \{c_1\}$
      \\\>(3)\> $\proj{\iota_\pc\;\gpc, \delta' \sepidx{} \epair{\kappa_1', \iota_1, c_1'}{\kappa_2', \iota_2, c_2'}_\glab}{1} = 
      \kappa_1'\rhd  (\iota_\pc\labjoin\iota_1)\;(\gpc\cjoin\glab) 
      \rhd \iota_\pc\;\gpc, \proj{\delta'}{1} \sepidx{} \{c_1'\}$
      \\\>(4)\> $\proj{\iota_\pc\;\gpc, \delta \sepidx{} \epair{\kappa_1, \iota_1, c_1}{\kappa_2, \iota_2, c_2}_\glab}{2} = 
      \kappa_2\rhd  (\iota_\pc\labjoin\iota_2)\;(\gpc\cjoin\glab) 
      \rhd \iota_\pc\;\gpc, \proj{\delta}{2} \sepidx{} \{c_2\}$
      \\\>(5)\> $\proj{\iota_\pc\;\gpc, \delta' \sepidx{} \epair{\kappa_1', \iota_1, c_1'}{\kappa_2', \iota_2, c_2'}_\glab}{2} = 
      \kappa_2\rhd  (\iota_\pc\labjoin\iota_2)\;(\gpc\cjoin\glab) 
      \rhd \iota_\pc\;\gpc, \proj{\delta}{2} \sepidx{} \{c_2\}$
      \\By definition of \rulename{P-Pc}
      \\\>(6)\> If $\kappa, \delta \sepidx{} c \stackrel{\alpha}{\stepsto} \kappa', \delta' \sepidx{} c'$ then $\kappa \rhd \iota_\pc\;\gpc, \delta \sepidx{} \{c\} \stackrel{\alpha}{\stepsto} \kappa' \rhd \iota_\pc\;\gpc, \delta' \sepidx{} \{c'\}$
      \\By (1), (2), (3) and (6)
      \\\>(7)\> $\kappa_1\rhd  (\iota_\pc\labjoin\iota_1)\;(\gpc\cjoin\glab) 
      \rhd \iota_\pc\;\gpc, \proj{\delta}{1} \sepidx{} \{c_1\} \stackrel{\alpha}{\stepsto} \kappa_1'\rhd  (\iota_\pc\labjoin\iota_1)\;(\gpc\cjoin\glab) 
      \rhd \iota_\pc\;\gpc, \proj{\delta'}{1} \sepidx{} \{c_1'\}$
      \\By (4), (5)
      \\\>(8)\> $ \kappa_2\rhd  (\iota_\pc\labjoin\iota_2)\;(\gpc\cjoin\glab) 
      \rhd \iota_\pc\;\gpc, \proj{\delta}{2} \sepidx{} \{c_2\} =  \kappa_2' \rhd  (\iota_\pc\labjoin\iota_2)\;(\gpc\cjoin\glab) 
      \rhd \iota_\pc\;\gpc, \proj{\delta'}{2} \sepidx{} \{c'_2\}$
      \\By (7) and (8), the conclusion holds
    \end{tabbing}

  \item [Case:] \rulename{P-C-Pair} where $c_1 = \eskip$ and $c_2$ is not $\eskip$. Similar for the symmetric case.
    \begin{tabbing}
      We show for $i=2, j=1$. 
      \\By assumption and Lemma~\ref{lem:store-proj}
      \\\,\,\,\,\=(1) \,\,\,\,\=
      $\kappa_2 \rhd (\iota_\pc \labjoin \iota_2)\;(\glab \cjoin \gpc), \proj{\delta}{2} \sepidx{}  c_2 \stackrel{\alpha}{\stepsto} 
      \kappa'_2 \rhd (\iota_\pc \labjoin \iota_2)\;(\glab \cjoin \gpc), \proj{\delta'}{2} \sepidx{}  c_2'$,
      \\\>\>and $\kappa_1 = \kappa_1'$, $\proj{\delta}{1} = \proj{\delta'}{1}$
      \\By projection of commands
      \\\>(2)\> $\proj{\iota_\pc\;\gpc, \delta \sepidx{} \epair{\kappa_1, \iota_1, \eskip}{\kappa_2, \iota_2, c_2}_\glab}{1} = 
      \kappa_1\rhd  (\iota_\pc\labjoin\iota_1)\;(\gpc\cjoin\glab) 
      \rhd \iota_\pc\;\gpc, \proj{\delta}{1} \sepidx{} \eskip$
      \\\>(3)\> $\proj{\iota_\pc\;\gpc, \delta' \sepidx{} \epair{\kappa_1, \iota_1, \eskip}{\kappa_2', \iota_2, c_2'}_\glab}{1} = 
      \kappa_1 \rhd  (\iota_\pc\labjoin\iota_1)\;(\gpc\cjoin\glab) 
      \rhd \iota_\pc\;\gpc, \proj{\delta'}{1} \sepidx{} \eskip$
      \\\>(4)\> $\proj{\iota_\pc\;\gpc, \delta \sepidx{} \epair{\kappa_1, \iota_1, \eskip}{\kappa_2, \iota_2, c_2}_\glab}{2} = 
      \kappa_2\rhd  (\iota_\pc\labjoin\iota_2)\;(\gpc\cjoin\glab) 
      \rhd \iota_\pc\;\gpc, \proj{\delta}{2} \sepidx{} \{c_2\}$
      \\\>(5)\> $\proj{\iota_\pc\;\gpc, \delta' \sepidx{} \epair{\kappa_1, \iota_1, \eskip}{\kappa_2', \iota_2, c_2'}_\glab}{2} = 
      \kappa_2'\rhd  (\iota_\pc\labjoin\iota_2)\;(\gpc\cjoin\glab) 
      \rhd \iota_\pc\;\gpc, \proj{\delta'}{2} \sepidx{} \{c_2'\}$
      \\By \rulename{P-Pc}
      \\\>(6)\> $\inferrule*{\kappa, \delta \sepidx{} c \stackrel{\alpha}{\stepsto} \kappa', \delta' \sepidx{} c'}{\kappa \rhd \iota_\pc\;\gpc, \delta \sepidx{} \{c\} \stackrel{\alpha}{\stepsto} \kappa' \rhd \iota_\pc\;\gpc, \delta' \sepidx{} \{c'\}}$
      \\By (1), (4), (5) and (6)
      \\\>(7)\> $\kappa_2\rhd  (\iota_\pc\labjoin\iota_2)\;(\gpc\cjoin\glab) 
      \rhd \iota_\pc\;\gpc, \proj{\delta}{2} \sepidx{} \{c_2\} \stackrel{\alpha}{\stepsto} \kappa_2'\rhd  (\iota_\pc\labjoin\iota_2)\;(\gpc\cjoin\glab) 
      \rhd \iota_\pc\;\gpc, \proj{\delta'}{2} \sepidx{} \{c_2'\}$
      \\By (1), (2), (3)
      \\\>(8)\> $ \kappa_1\rhd  (\iota_\pc\labjoin\iota_1)\;(\gpc\cjoin\glab) 
      \rhd \iota_\pc\;\gpc, \proj{\delta}{1} \sepidx{} \eskip =  \kappa_1 \rhd  (\iota_\pc\labjoin\iota_1)\;(\gpc\cjoin\glab) 
      \rhd \iota_\pc\;\gpc, \proj{\delta'}{1} \sepidx{} \eskip$
      \\By (7) and (8), the conclusion holds
    \end{tabbing}

  \item [Case:] \rulename{P-Skip-Pair}
    \begin{tabbing}
      By projection of commands
      \\\,\,\,\,\=(1) \,\,\,\,\=
      $\proj{\iota_\pc\;\gpc, \delta \sepidx{} \epair{\emptyset, \iota_1, \eskip}{\emptyset, \iota_2, \eskip}_\glab}{i} = 
      (\iota_\pc\labjoin\iota_i)\;(\gpc\cjoin\glab) 
      \rhd \iota_\pc\;\gpc, \proj{\delta}{i} \sepidx{} \{\eskip\}$
      \\By definition of \rulename{P-Pop}
      \\\>(2)\> $\iota_\pc\;\gpc \rhd \kappa, \delta \sepidx{} \{\eskip\} \stackrel{}{\stepsto} \kappa, \delta \sepidx{} \eskip$
      \\By (1), (2)
      \\\>(3)\> $\proj{\iota_\pc\;\gpc, \delta \sepidx{} \epair{\emptyset, \iota_1, \eskip}{\emptyset, \iota_2, \eskip}_\glab}{i} \stackrel{}{\stepsto} \iota_\pc\;\gpc, \proj{\delta}{i} \sepidx{} \eskip$
      \\By (3) and projection of commands, the conclusion holds
    \end{tabbing}

  \item [Case:] \rulename{P-Lift-If}
    \begin{tabbing}
      By \rulename{P-Lift-If}
      \\\,\,\,\,\=(1) \,\,\,\,\= $\iota_\pc\;\gpc, \delta \sepidx{} \eif\; \epair{\iota_1\;u_1}{\iota_2\;u_2}^\glab\ \ethen\ c_1\ \eelse\ c_2 \stackrel{}{\stepsto} \iota_\pc\;\gpc, \delta \sepidx{} \epair{\emptyset, \iota_1, c_j}{\emptyset, \iota_2, c_k}_\glab$
      \\By projection of commands
      \\\>(2) \>
      $\proj{\iota_\pc\;\gpc, \delta \sepidx{} \eif\; \epair{\iota_1\;u_1}{\iota_2\;u_2}^\glab\ \ethen\ c_1\ \eelse\ c_2}{i} =    
      \iota_\pc\;\gpc, \proj{\delta}{i} \sepidx{} \eif\; (\iota_i\;u_i)^\glab\ \ethen\ c_1\ \eelse\ c_2$
      \\\>(3)\>$\proj{\iota_\pc\;\gpc, \delta \sepidx{} \epair{\emptyset, \iota_1, c_j}{\emptyset, \iota_2, c_k}_\glab}{1} = 
      (\iota_\pc\labjoin\iota_1)\;(\gpc\cjoin\glab) 
      \rhd \iota_\pc\;\gpc, \proj{\delta}{1} \sepidx{} \{c_{j}\}$
      \\\>\> where $c_j$ is $c_1$ if $u_1 = \etrue$ and $c_j$ is $c_2$ if $u_1 = \efalse$ 
      \\\>(4)\>$\proj{\iota_\pc\;\gpc, \delta \sepidx{} \epair{\emptyset, \iota_1, c_j}{\emptyset, \iota_2, c_k}_\glab}{2} = 
      (\iota_\pc\labjoin\iota_2)\;(\gpc\cjoin\glab) 
      \rhd \iota_\pc\;\gpc, \proj{\delta}{2} \sepidx{} \{c_{k}\}$
      \\\>\> where $c_k$ is $c_1$ if $u_2 = \etrue$ and $c_j$ is $c_2$ if $u_2 = \efalse$ 
      \\By \rulename{P-If}
      \\\>(5)\>$\iota_\pc\;\gpc, \proj{\delta}{1} \sepidx{} \eif (\iota_1\;u_1)^\glab\ \ethen\ c_1\ \eelse\ c_2 \stepsto 
      (\iota_\pc\labjoin\iota_1)\;(\gpc\cjoin\glab) \rhd \iota_\pc\;\gpc, \proj{\delta}{1} \sepidx{} \{c_j\}$ 
      \\\>\>where $c_j$ is $c_1$ if $u_1 = \etrue$ and $c_j$ is $c_2$ if $u_1 = \efalse$ 
      \\\>(6)\>$\iota_\pc\;\gpc, \proj{\delta}{2} \sepidx{} \eif (\iota_2\;u_2)^\glab\ \ethen\ c_1\ \eelse\ c_2 \stepsto 
      (\iota_\pc\labjoin\iota_2)\;(\gpc\cjoin\glab) \rhd \iota_\pc\;\gpc, \proj{\delta}{2} \sepidx{} \{c_k\}$ 
      \\\>\>where $c_k$ is $c_1$ if $u_2 = \etrue$ and $c_k$ is $c_2$ if $u_2 = \efalse$ 
      \\By (2-6), the conclusion holds
    \end{tabbing}

  \end{description}
\end{proof}



\begin{thm}[Soundness]
  \label{thm:soundness}
  If $\kappa,\delta \sepidx{}  c \stackrel{\trace}{\stepsto^*} \kappa',\delta'\sepidx{} c'$
  where $\vdash \kappa, \delta \sepidx{} c\ \m{wf}$, \\
  then $\forall i \in \{1, 2\}$, 
  $\proj{\kappa,\delta \sepidx{}  c}{i} \stackrel{\proj{\trace}{i}}{\stepsto^*} 
  \proj{\kappa',\delta'\sepidx{} c'}{i}$, 
\end{thm}
\begin{proof}
  By induction on the number of steps in the sequence. Base case follows from assumption.
  
\noindent  Inductive Case: Holds for $n$ steps; To show for $n+1$ steps, i.e.,
    \\if $\kappa, \delta \sepidx{} c \stackrel{\trace}{\stepsto^n} \kappa', \delta' \sepidx{} c' \stackrel{\alpha}{\stepsto} \kappa'', \delta'' \sepidx{} c'' $, \\then $\forall i \in \{1, 2\}$, 
    $\proj{\kappa,\delta \sepidx{}  c}{i} \stackrel{\trace_{i}}{\stepsto^j} 
    \kappa_{0i},\delta_{0i}\sepidx{} c_{0i} \stackrel{\alpha_{i}}{\stepsto} 
    \kappa_{1i},\delta_{1i}\sepidx{} c_{1i}$, such that $\proj{\kappa'', \delta'' \sepidx c''}{i} = \kappa_{1i},\delta_{1i}\sepidx{} c_{1i}$ 
    \begin{tabbing}
      By induction hypothesis,
      \\\,\,\,\,\=(1) \,\,\,\,\=$\forall i \in\{1,2\}, \kappa_{0i},\delta_{0i}\sepidx{} c_{0i} = \proj{\kappa', \delta' \sepidx{} c'}{i}$ and $\trace_i = \proj{\trace}{i}$
      \\By Lemma~\ref{lem:soundness}, Lemma~\ref{lem:wf-pres}, (1) and projection of traces, the conclusion holds
    \end{tabbing}
\end{proof}

\subsection{Completeness of Paired-Execution}
\label{sec:app-comp}
\begin{lem}[Expression Completeness]
  \label{lem:exp-complete}
  If~$\forall i\in \{1, 2\}$, $\proj{\delta}{i} \sepidx{} e \evalsto
  v_i$, then $\exists v'$ s.t. $\delta \sepidx{} e \evalsto v'$
  with $\proj{v'}{i} = v_i$
\end{lem}
\begin{proof}
  By induction on the structure of the expression-evaluation derivation. 
  \begin{description}
  \item [Case:] \rulename{P-Const}
    \begin{tabbing}
      By \rulename{P-Const} and projection of values
    \end{tabbing}
    
  \item [Case:] \rulename{P-Var}
    \begin{tabbing}
      By assumption
      \\\,\,\,\,\=(1) \,\,\,\,\= $v_i = \eread\ \proj{\delta}{i}(x)$ and $v' = \eread\ \delta(x)$
      \\By (1) and Lemma~\ref{lem:store-proj-pres},
      \\\>(2)\>$\proj{\eread\ \delta(x)}{i} = \eread\ \proj{\delta}{i}(x)$ and $\proj{v'}{i} = v_i$
    \end{tabbing}
    
  \item [Case:] \rulename{P-Bop}
    \begin{tabbing}
      \,\,\,\,\=(1) \,\,\,\,\= $\inferrule*{
        \proj{\delta}{i} \sepidx{} e_1 \evalsto v_{1i} \\
        \proj{\delta}{i} \sepidx{} e_2 \evalsto v_{2i} \\
        v_i = v_{1i} \bop v_{2i} 
      }{
        \proj{\delta}{i} \sepidx{} e_1\bop e_2 \evalsto v_i
      }$
      \\By (1) and IH
      \\\>(2)\> $\delta \sepidx{} e_1 \evalsto v_1$, $\delta \sepidx{} e_2 \evalsto v_2$, $v_{1i} = \proj{v_1}{i}$ and $v_{2i} = \proj{v_2}{i}$
      \\T.S. $\forall i \in \{1, 2\}$, $\proj{v_1 \bop v_2}{i} = \proj{v_1}{i} \bop \proj{v_2}{i}$
      \\We show for $i=1$, the proof is similar for $i=2$
      \\{\bf Subcase I:} $v_1 = \epair{\iota_1\;u_1}{\iota_1'\;u_1'}^{\glab_1}$,
      $v_2 = \epair{\iota_2\;u_2}{\iota_2'\;u_2'}^{\glab_2}$
      \\\>By definition of $\bop$,
      \\\>\>(I1)\quad\= $v_1 \bop v_2 = \epair{\iota\;u}{\_}^\glab$ and $\proj{v_1 \bop v_2}{1} = (\iota\;u)^\glab$
      \\\>\>\quad\>      where $\iota = \iota_1 \labjoin \iota_2$, $u = u_1 \bop u_2$ and $\glab = \glab_1 \cjoin \glab_2$
      \\\>By value-projection definition 
      \\\>\>(I2)\quad\= $\proj{v_1}{1} = (\iota_1\;u_1)^\glab$ and $\proj{v_2}{1} = (\iota_2\;u_2)^\glab$
      \\\>By (I2) and definition of $\bop$,
      \\\>\>(I3)\quad\= $\proj{v_1}{1} \bop \proj{v_2}{1} = (\iota\;u)^\glab$ 
      where $\iota = \iota_1 \labjoin \iota_2$, $u = u_1 \bop u_2$ and $\glab = \glab_1 \cjoin \glab_2$
      \\\>By (I1) and (I3), the conclusion holds
      \\{\bf Subcase II:} $v_1 = \epair{\iota_1\;u_1}{\iota_1'\;u_1'}^{\glab_1}$,
      $v_2 = (\iota_2\;u_2)^{\glab_2}$, similarly for $v_1 = (\iota_1\;u_1)^{\glab_1}$, $v_2 = \epair{\iota_2\;u_2}{\iota_2'\;u_2'}_{\glab_2}$,
      \\\>By definition of $\bop$,
      \\\>\>(II1)\quad\= $v_1 \bop v_2 = \epair{\iota\;u}{\_}^\glab$ and $\proj{v_1 \bop v_2}{1} = (\iota\;u)^\glab$
      \\\>\>\quad\>      where $\iota = \iota_1 \labjoin \iota_2$, $u = u_1 \bop u_2$ and $\glab = \glab_1 \cjoin \glab_2$
      \\\>By value-projection definition 
      \\\>\>(II2)\quad\= $\proj{v_1}{1} = (\iota_1\;u_1)^\glab$ and $\proj{v_2}{1} = (\iota_2\;u_2)^\glab$
      \\\>By (II2) and definition of $\bop$,
      \\\>\>(II3)\quad\= $\proj{v_1}{1} \bop \proj{v_2}{1} = (\iota\;u)^\glab$ 
      where $\iota = \iota_1 \labjoin \iota_2$, $u = u_1 \bop u_2$ and $\glab = \glab_1 \cjoin \glab_2$
      \\\>By (II1) and (II3), the conclusion holds
      \\{\bf Subcase III:} $v_1 = (\iota_1\;u_1)^{\glab_1}$,
      $v_2 = (\iota_2\;u_2)^{\glab_2}$
      \\\>By definition of $\bop$,
      \\\>\>(III1)\quad\= $v_1 \bop v_2 = (\iota\;u)^\glab$ and $\proj{v_1 \bop v_2}{1} = (\iota\;u)^\glab$
      \\\>\>\quad\>where $\iota = \iota_1 \labjoin \iota_2$, $u = u_1 \bop u_2$ and $\glab = \glab_1 \cjoin \glab_2$
      \\\>By value-projection definition 
      \\\>\>(III2)\quad\= $\proj{v_1}{1} = (\iota_1\;u_1)^\glab$ and $\proj{v_2}{1} = (\iota_2\;u_2)^\glab$
      \\\>By (III2) and definition of $\bop$,
      \\\>\>(III3)\quad\= $\proj{v_1}{1} \bop \proj{v_2}{1} = (\iota\;u)^\glab$ 
      where $\iota = \iota_1 \labjoin \iota_2$, $u = u_1 \bop u_2$ and $\glab = \glab_1 \cjoin \glab_2$
      \\\>By (III1) and (III3), the conclusion holds
    \end{tabbing}
    
  \item [Case:] \rulename{P-Cast}
    \begin{tabbing}
      \,\,\,\,\=(1) \,\,\,\,\=
      $\inferrule*{\proj{\delta}{i} \sepidx{} e \evalsto v_i 
        \\ v_i' = (E, \glab)\rhd v_i
      }{
        \proj{\delta}{i} \sepidx{} E^{\glab} e 
        \evalsto  v_i'
      }$
      \\By IH and (1)
      \\\>(2)\> $\forall i \in \{1, 2\}$, $\delta \sepidx{} e \evalsto v$ and $\proj{v}{i} = v_i$
      \\We show for $i=1$, the proof is similar for $i=2$
      \\{\bf Subcase I:} $v = \epair{\iota_1\;u_1}{\iota_2\;u_2}^{\glab'}$,
      \\\>By definition of $\rhd$ cast operation,
      \\\>\>(I1)\quad\= $v' = \epair{\iota_1'\;u_1}{\iota_2'\;u_2}^\glab$ and $\proj{v'}{1} = (\iota_1'\;u_1)^\glab$ where $\iota_1' = \iota_1 \bowtie E$
      \\\>By value-projection definition and cast operation,
      \\\>\>(I2)\quad\= $\proj{v}{1} = (\iota_1\;u_1)^{\glab'}$ and $v_1' = (\iota_1''\;u_1)^\glab$ where $\iota_1'' = \iota_1 \bowtie E$
      \\\>By (I1) and (I2), the conclusion holds
      \\{\bf Subcase II:} $v = (\iota\;u)^{\glab'}$,
      \\\>By definition of cast
      \\\>\>(II1)\quad\= $v' = (\iota'\;u)^\glab$ and $\proj{v'}{1} = (\iota'\;u)^\glab$ where $\iota' = \iota \bowtie E$
      \\\>By definition of value-projection and cast 
      \\\>\>(II2)\quad\= $\proj{v}{1} = (\iota\;u)^{\glab'}$ and $v''_1 = (\iota''\;u)^\glab$ where $\iota'' = \iota \bowtie E$
      \\\>By (II1) and (II2), $\proj{v'}{1} = v''_1$
    \end{tabbing}
  \end{description}
\end{proof}

\begin{lem}
  \label{lem:wf-proj}
  If~ $\vdash \kappa, \delta \sepidx{} c \ \m{wf}$ and $c$ does not contain pairs or braces, then
  $\forall i\in \{1, 2\}$, $\proj{\kappa, \delta \sepidx{} c}{i} = \kappa, \proj{\delta}{i} \sepidx{} c$ 
\end{lem}
\begin{proofsketch}
  By induction on the structure of $c$. Most cases use the projection of commands and respective rules. $c_1;c_2$ use the IH additionally. 
\end{proofsketch}

\begin{lem}
  \label{lem:proj-complete}
  If~ $\forall i\in \{1, 2\}$, $\proj{\kappa, \delta \sepidx{} c}{i} \stackrel{\alpha_i}{\stepsto}
  \kappa'_i, \delta'_i \sepidx{} c'_i$ or $\proj{\kappa, \delta \sepidx{} c}{i} = \kappa_0, \delta_0 \sepidx{} \eskip$, and $\vdash \kappa, \delta \sepidx{} c \ \m{wf}$ then 
  $\exists \kappa', \delta', c'$ s.t. 
    $\kappa, \delta \sepidx{} c \stackrel{\trace}{\stepsto^*} \kappa', \delta' \sepidx{} c'$
    with $\proj{\kappa', \delta' \sepidx{} c'}{i} = \kappa'_i, \delta'_i \sepidx{} c'_i$ and $\proj{\trace}{i} = \alpha_i$ 
\end{lem}
\begin{proof}
  By induction on the structure of $c$. If $\proj{\kappa, \delta \sepidx{} c}{1} = \kappa_0, \delta_0 \sepidx{} \eskip$ and $\proj{\kappa, \delta \sepidx{} c}{2} = \kappa_0, \delta_0 \sepidx{} \eskip$, then $c = \eskip$ (from the projection of command-configurations). If $c=\eskip$, then $\kappa, \delta \sepidx{} c$ takes $0$ steps and the conclusion follows from the assumption. If at least one of the projected runs steps:
  \begin{description}
    
  \item [Case:] $c = c_1;c_2$
    \begin{tabbing}
      By projection of commands,
      \\\,\,\,\,\=(1) \,\,\,\,\= $\inferrule*{\proj{\kappa, \delta \sepidx{} c_1}{i} = \kappa_i, \delta_i \sepidx{} c_{1i}}{\proj{\kappa, \delta \sepidx{} c_1;c_2}{i} = \kappa_i, \delta_i \sepidx{} c_{1i};c_2}$
      \\By \rulename{P-Seq}
      \\\>(2)\>
      $\inferrule*{
        \kappa_i,\delta_i\sepidx{} c_{1i}  \stackrel{\alpha_i}{\stepsto}  \kappa_i',\delta_i' \sepidx{} c'_{1i}
      }{
        \kappa_i,\delta_i \sepidx{} c_{1i}; c_2
        \stackrel{\alpha_i}{\stepsto}  \kappa_i',\delta_i'\sepidx{} c'_{1i}; c_2
      }$
      \\By (2),
      \\\>(3)\>$\proj{\kappa,\delta\sepidx{} c_{1}}{i}  \stackrel{\alpha_i}{\stepsto}  \kappa_i',\delta_i' \sepidx{} c'_{1i}$
      \\By (3) and IH,
      \\\>(4)\> $\kappa,\delta \sepidx{}  c_{1} \stackrel{\trace}{\stepsto^*} \kappa',\delta' \sepidx{} c'_1$ and $\forall i \in \{1, 2\}$, $\proj{\kappa', \delta' \sepidx{} c_1'}{i} = \kappa_i', \delta_i' \sepidx{} c_{1i}'$ and $\proj{\trace}{i} = \alpha_i$, (or)
      \\\>\> $\proj{\kappa, \delta \sepidx{} c_1}{i} = \kappa_0, \delta_0 \sepidx{} \eskip$
      \\By (1), (2), (4)
      \\\>(5) \>$\forall i \in \{1, 2\}$, $\proj{\kappa', \delta' \sepidx{} c_1';c_2}{i} = \kappa_i', \delta_i' \sepidx{} {c_{1i}';c_2}$, and  $\proj{\trace}{i} = \alpha_i$
      \\By (5), the conclusion holds
    \end{tabbing}
    
  \item [Case:] $c = \{c\}$
    \begin{tabbing}
      By projection of commands,
      \\\,\,\,\,\=(1) \,\,\,\,\= $\inferrule*{\proj{\kappa, \delta \sepidx{} c}{i} = \kappa_i, \delta_i \sepidx{} c_i}{\proj{\kappa \rhd \iota_\pc\;\gpc, \delta \sepidx{} \{c\}}{i} = \kappa_i \rhd \iota_\pc\;\gpc, \delta_i \sepidx{} \{c_{i}\}}$
      \\By \rulename{P-Pc},
      \\\>(2)\>$\inferrule*{
        \kappa_i,\delta_i\sepidx{} c_i
        \stackrel{\alpha_i}{\stepsto}  \kappa_i',\delta_i'\sepidx{} c_i'
      }{
        \kappa_i \rhd \iota_\pc\;\gpc,\delta_i\sepidx{} \{c_i\} 
        \stackrel{\alpha_i}{\stepsto}  \kappa_i'\rhd\iota_\pc\;\gpc,\delta_i'\sepidx{}\{c_i'\}
      }$
      \\By (2), IH,
      \\\>(3)\> $\kappa,\delta \sepidx{}  c \stackrel{\trace}{\stepsto^*} \kappa',\delta' \sepidx{} c'$ and
      $\forall i \in \{1, 2\}$,  $\proj{\kappa', \delta' \sepidx{} c'}{i} = \kappa_i', \delta_i' \sepidx{} c_i'$ and $\proj{\trace}{i} = \alpha_i$, 
      \\By (1), (2) and (3), the conclusion holds
    \end{tabbing}

  \item [Case:] $c = \eskip;c$
    \begin{tabbing}
      By projection of commands,
      \\\,\,\,\,\=(1) \,\,\,\,\= $\inferrule*{\proj{\iota_\pc\;\gpc, \delta \sepidx{} \eskip}{i} = \iota_\pc\;\gpc, \proj{\delta}{i} \sepidx{} \eskip}{\proj{\iota_\pc\;\gpc, \delta \sepidx{} \eskip;c}{i} = \iota_\pc\;\gpc, \proj{\delta}{i} \sepidx{} \eskip;c}$
      \\By \rulename{P-Skip},
      \\\>(2)\>$\inferrule*{ }{
        \iota_\pc\;\gpc, \proj{\delta}{i} \sepidx{} \eskip;c \stepsto \iota_\pc\;\gpc, \proj{\delta}{i} \sepidx{} c}$
      \\\>(3)\>$\inferrule*{ }{
        \iota_\pc\;\gpc, \delta \sepidx{} \eskip;c \stepsto \iota_\pc\;\gpc, \delta \sepidx{} c}$
      \\T.S. $\proj{\iota_\pc\;\gpc, \delta \sepidx{} c}{i} = \iota_\pc\;\gpc, \proj{\delta}{i} \sepidx{} c$
      \\\>(4)\> From well-formedness definition, $c$ does not contain pairs or braces
      \\From Lemma~\ref{lem:wf-proj} and (4), the conclusion holds
    \end{tabbing}
  
  \item [Case:] $c = \{\eskip\}$
    \begin{tabbing}
      By projection of commands,
      \\\,\,\,\,\=(1) \,\,\,\,\= $\inferrule*{ 
        \proj{\iota\;\glab, \delta \sepidx{} \eskip}{i} = \iota\;\glab, \proj{\delta}{i} \sepidx{} \eskip
      }{
        \proj{\iota\;\glab\rhd \iota_\pc\; \gpc, \delta \sepidx{} \{\eskip\}}{i} = \iota\;\glab\rhd
        \iota_\pc\;\gpc, \proj{\delta}{i} \sepidx{} \{\eskip\}
      }$
      \\The rule forces $\kappa = \iota\;\glab$
      \\By \rulename{P-Pop},
      \\\>(2)\>$\inferrule*{ }{
        \iota\;\glab \rhd \iota_\pc\;\gpc, \delta \sepidx{} \{\eskip\} \stepsto \iota_\pc\;\gpc, \delta \sepidx{} \eskip}$
      \\\>(3)\>$\inferrule*{ }{
        \iota\;\glab \rhd \iota_\pc\;\gpc, \proj{\delta}{i} \sepidx{} \{\eskip\} \stepsto \iota_\pc\;\gpc, \proj{\delta}{i} \sepidx{} \eskip}$
      \\By (1), (2) and (3), the conclusion holds
    \end{tabbing}

  \item [Case:] $c = \eif (\iota\;b)^\glab\ \ethen\ c_1\ \eelse\ c_2$
    \begin{tabbing}
      By projection of commands,
      \\\,\,\,\,\=(1) \,\,\,\,\= $\inferrule*{ }{
        \proj{\iota_\pc\;\gpc, \delta \sepidx{} \eif\; (\iota\;b)^\glab\; \ethen\; c_1\; \eelse\; c_2 }{i} = 
        \iota_\pc\;\gpc, \proj{\delta}{i} \sepidx{} \eif\; (\iota\;b)^\glab\; \ethen\; c_1\; \eelse\; c_2 
      }
      $
      \\By \rulename{P-If},
      \\\>(2)\>$\inferrule*{
        \\ \iota'_\pc = \iota_\pc\labjoin \iota
        \\ \gpc' =\gpc\cjoin\glab
        \\ c_j = c_1~\mbox{if}~b = \etrue
        \\ c_j = c_2~\mbox{if}~b = \efalse
      }{
        \iota_\pc\;\gpc,\proj{\delta}{i} \sepidx{} \eif\; (\iota\;b)^\glab\ \ethen\ c_1\ \eelse\ c_2
        \stepsto  \iota'_\pc\;\gpc'\rhd \iota_\pc\;\gpc, \proj{\delta}{i} \sepidx{}  \{c_j\}
      }$
      \\\>(3)\>$\inferrule*{
        \\ \iota'_\pc = \iota_\pc\labjoin \iota
        \\ \gpc' =\gpc\cjoin\glab
        \\ c_j = c_1~\mbox{if}~b = \etrue
        \\ c_j = c_2~\mbox{if}~b = \efalse
      }{
        \iota_\pc\;\gpc,\delta \sepidx{} \eif\; (\iota\;b)^\glab\ \ethen\ c_1\ \eelse\ c_2
        \stepsto  \iota'_\pc\;\gpc'\rhd \iota_\pc\;\gpc, \delta \sepidx{}  \{c_j\}
      }$
      \\Suppose $b = \etrue$. Similar for $b=\efalse$
      \\T.S. $\proj{\iota'_\pc\;\gpc'\rhd \iota_\pc\;\gpc, \delta \sepidx{}  \{c_1\}}{i} = \iota'_\pc\;\gpc'\rhd \iota_\pc\;\gpc, \proj{\delta}{i} \sepidx{}  \{c_1\}$
      \\By projection of commands,
      \\\>(4)\> $\inferrule*{ 
        \proj{\iota'_\pc\;\gpc', \delta \sepidx{} c_1}{i} = \iota'_\pc\;\gpc', \delta' \sepidx{} c_1'
      }{
        \proj{\iota'_\pc\;\gpc'\rhd \iota_\pc\;\gpc, \delta \sepidx{} \{c_1\}}{i} = \iota'_\pc\;\gpc'\rhd 
        \iota_\pc\;\gpc, \delta' \sepidx{} \{c_1'\}
      }
      $
      \\\>(5)\>By well-formedness definition, $c_1$ does not contain pairs or braces.
      \\ By Lemma~\ref{lem:wf-proj} and (5), the conclusion holds
    \end{tabbing}
    
  \item [Case:] $c = \ewhile^X\ e\ \edo\ c$
    \begin{tabbing}
      By projection of commands and \rulename{P-While}
    \end{tabbing}

  \item [Case:] $c = x := e$
    \begin{tabbing}
      By projection of commands, 
      \\T.S. $\proj{\delta'}{i} = \delta_i'$ where $\delta_i' = \proj{\delta}{i}[x \mapsto \eupdate\ \proj{\delta}{i}(x)\ v_i']$ and $\delta' = \delta[x \mapsto \eupdate\ \delta(x)\ v']$
      \\ By Lemma~\ref{lem:exp-complete}
      \\\,\,\,\,\=(1) \,\,\,\,\= $\forall i \in \{1, 2\}$, if $\proj{\delta}{i} \sepidx{} e \evalsto v_i$, then 
      $\delta \sepidx{} e \evalsto v$ such that $\proj{v}{i} = v_i$,
      \\By (1), Lemma~\ref{lem:reflv-proj-pres-pi},
      \\\>(2)\> $\proj{v'}{i} = \proj{\reflvof{}(\iota_\pc, v)}{i} = \reflvof{}(\iota_\pc, \proj{v}{i}) = v_i'$
      \\By (2), Lemma~\ref{lem:update-proj-pres} and~\ref{lem:store-proj-pres}
      \\\>(3)\> $\proj{v''}{i} = \proj{\eupdate\ \delta(x)\ v'}{i} = \eupdate\ \proj{\delta}{i}(x)\ \proj{v'}{i} = v_i''$
      \\\>(4)\> $\delta'(x) = v''$ and $\delta_i'(x) = v_i''$
      \\By (3) and (4) the conclusion holds
    \end{tabbing}

  \item [Case:] $c = \eif^X\ e\ \ethen\ c_1\ \eelse\ c_2$
    \begin{tabbing}
      By projection of commands,
      \\T.S. $\proj{\delta'}{i} = \delta_i'$ where $\delta_i' = \rflof{}(\proj{\delta}{i}, X, \iota_\pc\labjoin\labof(v_i))$
      and $\delta' = \rflof{}(\delta, X, \iota_\pc\labjoin\labof(v))$
      \\ By Lemma~\ref{lem:exp-complete}
      \\\,\,\,\,\=(1) \,\,\,\,\= $\proj{v}{i} = v_i$,
      \\By (1), Lemma~\ref{lem:rflof-proj-pres-pi}
      \\\>(2)\> $\proj{\rflof{ }(\delta, X, \iota_\pc\labjoin\labof(v))}{i} = \rflof{}(\proj{\delta}{i}, X, \iota_\pc\labjoin\labof(\proj{v}{i}))$
      \\By (1), (2), the conclusion holds
    \end{tabbing}

  \item [Case:] $c = \eoutput(\ell, e)$
    \begin{tabbing}
      By projection of commands, $\proj{\iota_\pc\;\gpc, \delta \sepidx{} \eskip}{i} = \iota_\pc\;\gpc, \proj{\delta}{i} \sepidx{} \eskip$.
      \\T.S. $\proj{(\ell, v'')}{i} = (\ell, v_i'')$ 
      \\ By Lemma~\ref{lem:exp-complete}
      \\\,\,\,\,\=(1) \,\,\,\,\= $\forall i \in \{1, 2\}$, if $\proj{\delta}{i} \sepidx{} e \evalsto v_i$, then 
      $\delta \sepidx{} e \evalsto v$ such that $\proj{v}{i} = v_i$,
      \\By (1), Lemma~\ref{lem:reflv-proj-pres-pi}
      \\\>(2)\> $\proj{v'}{i} = \proj{\reflvof{}(\iota_\pc, v)}{i} = \reflvof{}(\iota_\pc, \proj{v}{i}) = v_i'$
      \\By (2),Lemma~\ref{lem:refvl-updval-pres}
      \\\>(3)\> $\proj{v''}{i} = \proj{\updval([\ell, \ell], v')}{i} = \updval{}([\ell, \ell], \proj{v'}{i}, [\ell, \ell]) = v_i''$
      \\By (3) and projection of traces, the conclusion holds
    \end{tabbing}
    
  \item [Case:] $c = \epair{\kappa_1,\iota_1,c_1}{\kappa_2,\iota_2,c_2}$ where $c_1$ and $c_2$ are not $\eskip$
    \begin{tabbing}
      By projection of commands
      \\\,\,\,\,\=(1) \,\,\,\,\=$\proj{\iota_\pc\;\gpc, \delta \sepidx{} \epair{\kappa_1, \iota_1, c_1}{\kappa_2, \iota_2, c_2}_\glab}{i} = 
      \kappa_i\rhd  (\iota_\pc\labjoin\iota_i)\;(\gpc\cjoin\glab) \rhd \iota_\pc\;\gpc, \proj{\delta}{i} \sepidx{} \{c_i\}$
      \\By assumption, \rulename{P-Pc} and (1)
      \\\>(2)\>$\kappa_i\rhd  (\iota_\pc\labjoin\iota_i)\;(\gpc\cjoin\glab) \rhd \iota_\pc\;\gpc, \proj{\delta}{i} \sepidx{} \{c_i\} \stackrel{\alpha_i}{\stepsto} \kappa_i'' \rhd \iota_\pc\;\gpc, \delta_i'' \sepidx{} \{c_i''\}$
      \\\>(3)\>$\kappa_i\rhd  (\iota_\pc\labjoin\iota_i)\;(\gpc\cjoin\glab), \proj{\delta}{i} \sepidx{} c_i \stackrel{\alpha_i}{\stepsto} \kappa_i'', \delta_i'' \sepidx{} c_i''$
      \\Suppose $i =1, j=2$ followed by $i=2,j=1$. Similar for the symmetric case
      \\\>(4)\>$\iota_\pc\;\gpc, \delta \sepidx{} \epair{\kappa_1, \iota_1, c_1}{\kappa_2, \iota_2, c_2}_\glab \stackrel{\alpha'}{\stepsto} \iota_\pc\;\gpc, \delta' \sepidx{} \epair{\kappa_1', \iota_1, c_1'}{\kappa_2, \iota_2, c_2}_\glab$
      \\\>\>$\stackrel{\alpha''}{\stepsto} \iota_\pc\;\gpc, \delta'' \sepidx{} \epair{\kappa_1', \iota_1, c_1'}{\kappa_2', \iota_2, c_2'}_\glab$
      \\By (4), Lemma~\ref{lem:store-proj},
      \\\>(5)\>$\kappa_1 \rhd (\iota_\pc\labjoin\iota_1\;\gpc \cjoin g), \proj{\delta}{1} \sepidx{} c_1 \stackrel{\proj{\alpha'}{1}}{\stepsto} \kappa_1' \rhd (\iota_\pc\labjoin\iota_1\;\gpc \cjoin g), \proj{\delta'}{1} \sepidx{} c_1'$,
      \\\>\>$\proj{\delta}{2} = \proj{\delta'}{2}$ and $\proj{\alpha'}{2} = \cdot$
      \\\>(6)\>$\kappa_2 \rhd (\iota_\pc\labjoin\iota_2\;\gpc \cjoin g), \proj{\delta'}{2} \sepidx{} c_2 \stackrel{\proj{\alpha''}{2}}{\stepsto} \kappa_2' \rhd (\iota_\pc\labjoin\iota_2\;\gpc \cjoin g), \proj{\delta''}{2} \sepidx{} c_2'$,
      \\\>\>$\proj{\delta'}{1} = \proj{\delta''}{1}$ and $\proj{\alpha''}{1} = \cdot$
      \\\>(7)\>$\iota_\pc\;\gpc, \delta \sepidx{} \epair{\kappa_1, \iota_1, c_1}{\kappa_2, \iota_2, c_2}_\glab \stackrel{\alpha', \alpha''}{\stepsto^*} \iota_\pc\;\gpc, \delta'' \sepidx{} \epair{\kappa_1', \iota_1, c_1'}{\kappa_2', \iota_2, c_2'}_\glab$
      \\T.S. $\proj{\iota_\pc\;\gpc, \delta'' \sepidx{} \epair{\kappa_1', \iota_1, c_1'}{\kappa_2', \iota_2, c_2'}_\glab}{i} = \kappa_i'' \rhd \iota_\pc\;\gpc, \delta_i'' \sepidx{} \{c_i''\}$ and $\proj{\alpha', \alpha''}{i} = \alpha_i$
      \\By projection of commands,
      \\\>(8)\>$\proj{\iota_\pc\;\gpc, \delta'' \sepidx{} \epair{\kappa_1', \iota_1, c_1'}{\kappa_2', \iota_2, c_2'}_\glab}{i} = 
      \kappa_i' \rhd  (\iota_\pc\labjoin\iota_i)\;(\gpc\cjoin\glab) \rhd \iota_\pc\;\gpc, \proj{\delta''}{i} \sepidx{} \{c_i'\}$
      \\By (8), T.S. $\kappa_i' \rhd (\iota_\pc\labjoin\iota_i)\;(\gpc\cjoin\glab) = \kappa_i''$, $\proj{\delta''}{i} = \delta''$, $c_i' = c_i''$ and $\proj{\alpha', \alpha''}{i} = \alpha_i$
      \\By (3), (5), (6)
      \\\>(9)\>$\kappa''_1 = \kappa_1' \rhd (\iota_\pc\labjoin\iota_1) \;(\gpc\cjoin\glab)$, $\delta_1'' = \proj{\delta'}{1}$, $c_1' = c_1''$, $\alpha_1 = \proj{\alpha'}{1}$
      \\\>(10)\>$\kappa''_2 = \kappa_2' \rhd (\iota_\pc\labjoin\iota_2) \;(\gpc\cjoin\glab)$, $\delta_2'' = \proj{\delta''}{2}$, $c_2' = c_2''$, $\alpha_2 = \proj{\alpha''}{2}$
      \\By (5), (6), (9), (10) and projection of traces, the conclusion holds
    \end{tabbing}

  \item [Case:] $c = \epair{\kappa_1,\iota_1,c_1}{\kappa_2,\iota_2,c_2}$ where $c_1 = \eskip$ and $c_2$ is not. Similar for $c_2 = \eskip$ and $c_1 \neq \eskip$
    \begin{tabbing}
      By projection of commands
      \\\,\,\,\,\=(1) \,\,\,\,\=$\proj{\iota_\pc\;\gpc, \delta \sepidx{} \epair{\kappa_1, \iota_1, \eskip}{\kappa_2, \iota_2, c_2}_\glab}{1} = 
      \kappa_1\rhd  (\iota_\pc\labjoin\iota_1)\;(\gpc\cjoin\glab) \rhd \iota_\pc\;\gpc, \proj{\delta}{1} \sepidx{} \eskip$
      \\\>(2)\>$\proj{\iota_\pc\;\gpc, \delta \sepidx{} \epair{\kappa_1, \iota_1, \eskip}{\kappa_2, \iota_2, c_2}_\glab}{2} = 
      \kappa_2\rhd  (\iota_\pc\labjoin\iota_2)\;(\gpc\cjoin\glab) \rhd \iota_\pc\;\gpc, \proj{\delta}{2} \sepidx{} \{c_2\}$
      \\(1) does not proceed as the projection of $c$ gives $\eskip$
      \\By assumption, \rulename{P-Pc} and (2)
      \\\>(3)\>$\kappa_2\rhd  (\iota_\pc\labjoin\iota_2)\;(\gpc\cjoin\glab) \rhd \iota_\pc\;\gpc, \proj{\delta}{2} \sepidx{} \{c_2\} \stackrel{\alpha_2}{\stepsto} \kappa_2'' \rhd \iota_\pc\;\gpc, \delta_2'' \sepidx{} \{c_2''\}$
      \\\>(4)\>$\kappa_2\rhd  (\iota_\pc\labjoin\iota_2)\;(\gpc\cjoin\glab), \proj{\delta}{2} \sepidx{} c_2 \stackrel{\alpha_2}{\stepsto} \kappa_2'', \delta_2'' \sepidx{} c_2''$
      \\Suppose $i =2, j=1$. 
      \\\>(5)\>$\iota_\pc\;\gpc, \delta \sepidx{} \epair{\kappa_1, \iota_1, \eskip}{\kappa_2, \iota_2, c_2}_\glab \stackrel{\alpha'}{\stepsto} \iota_\pc\;\gpc, \delta' \sepidx{} \epair{\kappa_1, \iota_1, \eskip}{\kappa_2', \iota_2', c_2'}_\glab$
      \\By Lemma~\ref{lem:store-proj},
      \\\>(6)\>$\kappa_2 \rhd (\iota_\pc\labjoin\iota_2\;\gpc \cjoin g), \proj{\delta}{2} \sepidx{} c_2 \stackrel{\proj{\alpha'}{2}}{\stepsto} \kappa_2' \rhd (\iota_\pc\labjoin\iota_2\;\gpc \cjoin g), \proj{\delta'}{2} \sepidx{} c_2'$,
      \\\>\>$\proj{\delta'}{1} = \proj{\delta}{1}$ and $\proj{\alpha'}{1} = \cdot$
      \\T.S. $\proj{\iota_\pc\;\gpc, \delta' \sepidx{} \epair{\kappa_1, \iota_1, \eskip}{\kappa_2', \iota_2, c_2'}_\glab}{1} = \kappa_1\rhd  (\iota_\pc\labjoin\iota_1)\;(\gpc\cjoin\glab) \rhd \iota_\pc\;\gpc, \proj{\delta}{1} \sepidx{} \eskip$ and
      \\$\proj{\iota_\pc\;\gpc, \delta' \sepidx{} \epair{\kappa_1, \iota_1, \eskip}{\kappa_2', \iota_2, c_2'}_\glab}{2} = \kappa_2'' \rhd \iota_\pc\;\gpc, \delta_2'' \sepidx{} \{c_2''\}$ and $\proj{\alpha'}{1} =\cdot $ and $\proj{\alpha'}{2} =  \alpha_2$ 
      \\By projection of commands,
      \\\>(7)\>$\proj{\iota_\pc\;\gpc, \delta' \sepidx{} \epair{\kappa_1, \iota_1, \eskip}{\kappa_2', \iota_2, c_2'}_\glab}{1} = 
      \kappa_1 \rhd  (\iota_\pc\labjoin\iota_1)\;(\gpc\cjoin\glab) \rhd \iota_\pc\;\gpc, \proj{\delta'}{1} \sepidx{} \eskip$
      \\\>(8)\>$\proj{\iota_\pc\;\gpc, \delta' \sepidx{} \epair{\kappa_1, \iota_1, \eskip}{\kappa_2', \iota_2, c_2'}_\glab}{2} = 
      \kappa_2' \rhd  (\iota_\pc\labjoin\iota_2)\;(\gpc\cjoin\glab) \rhd \iota_\pc\;\gpc, \proj{\delta'}{2} \sepidx{} \{c_2'\}$
      \\By (8), T.S. $\kappa_2' \rhd (\iota_\pc\labjoin\iota_2)\;(\gpc\cjoin\glab) = \kappa_2''$, $\proj{\delta'}{2} = \delta''$, $c_2' = c_2''$ and $\proj{\alpha'}{2} = \alpha_2$
      \\By (4) and (6)
      \\\>(9)\>$\kappa''_2 = \kappa_2' \rhd (\iota_\pc\labjoin\iota_2) \;(\gpc\cjoin\glab)$, $\delta'' = \proj{\delta'}{2}$, $c_2' = c_2''$, $\alpha_2 = \proj{\alpha'}{2}$
      \\By (6), (7), (9) and projection of traces, the conclusion holds 
    \end{tabbing}
    
  \item [Case:] $c = \epair{\emptyset, \iota_1, \eskip}{\emptyset, \iota_2, \eskip}$
    \begin{tabbing}
      By projection of commands
      \\\,\,\,\,\=(1) \,\,\,\,\=
      $\proj{\iota_\pc\;\gpc, \delta \sepidx{} \epair{\emptyset, \iota_1, \eskip}{\emptyset, \iota_2, \eskip}_\glab}{i} = 
      (\iota_\pc\labjoin\iota_i)\;(\gpc\cjoin\glab) 
      \rhd \iota_\pc\;\gpc, \proj{\delta}{i} \sepidx{} \{\eskip\}$
      \\\>(2)\> $\proj{\iota_\pc\;\gpc, \delta \sepidx{} \eskip}{i} = \iota_\pc\;\gpc, \proj{\delta}{i} \sepidx{} \eskip$
      \\By (1) and definition of \rulename{P-Pop}
      \\\>(3)\> $(\iota_\pc\labjoin\iota_i)\;(\gpc\cjoin\glab) \rhd \iota_\pc\;\gpc, \proj{\delta}{i} \sepidx{} \{\eskip\} \stackrel{}{\stepsto} \iota_\pc\;\gpc, \proj{\delta}{i} \sepidx{} \eskip$
      \\By (2) and (3), the conclusion holds
    \end{tabbing}

  \item [Case:] $c = \eif\ \epair{\iota_1\;u_1}{\iota_2\;u_2}^\glab\ \ethen\ c_1\ \eelse\ c_2$
    \begin{tabbing}
      By \rulename{P-Lift-If}
      \\\,\,\,\,\=(1) \,\,\,\,\= $\proj{\iota_\pc\;\gpc, \delta \sepidx{} \eif\; \epair{\iota_1\;u_1}{\iota_2\;u_2}^\glab\ \ethen\ c_1\ \eelse\ c_2}{i} \stackrel{}{\stepsto} \proj{\iota_\pc\;\gpc, \delta \sepidx{} \epair{\emptyset, \iota_1, c_j}{\emptyset, \iota_2, c_k}_\glab}{i}$
      \\By projection of commands
      \\\>(2) \>
      $\proj{\iota_\pc\;\gpc, \delta \sepidx{} \eif\; \epair{\iota_1\;u_1}{\iota_2\;u_2}^\glab\ \ethen\ c_1\ \eelse\ c_2}{i} =    
      \iota_\pc\;\gpc, \proj{\delta}{i} \sepidx{} \eif\; (\iota_i\;u_i)^\glab\ \ethen\ c_1\ \eelse\ c_2$
      \\\>(3)\>$\proj{\iota_\pc\;\gpc, \delta \sepidx{} \epair{\emptyset, \iota_1, c_j}{\emptyset, \iota_2, c_k}_\glab}{1} = 
      (\iota_\pc\labjoin\iota_1)\;(\gpc\cjoin\glab) 
      \rhd \iota_\pc\;\gpc, \proj{\delta}{1} \sepidx{} \{c_{j}\}$
      \\\>\> where $c_j$ is $c_1$ if $u_1 = \etrue$ and $c_j$ is $c_2$ if $u_1 = \efalse$ 
      \\\>(4)\>$\proj{\iota_\pc\;\gpc, \delta \sepidx{} \epair{\emptyset, \iota_1, c_j}{\emptyset, \iota_2, c_k}_\glab}{2} = 
      (\iota_\pc\labjoin\iota_2)\;(\gpc\cjoin\glab) 
      \rhd \iota_\pc\;\gpc, \proj{\delta}{2} \sepidx{} \{c_{k}\}$
      \\\>\> where $c_k$ is $c_1$ if $u_2 = \etrue$ and $c_j$ is $c_2$ if $u_2 = \efalse$ 
      \\By \rulename{P-If}
      \\\>(5)\>$\iota_\pc\;\gpc, \proj{\delta}{1} \sepidx{} \eif (\iota_1\;u_1)^\glab\ \ethen\ c_1\ \eelse\ c_2 \stepsto 
      (\iota_\pc\labjoin\iota_1)\;(\gpc\cjoin\glab) \rhd \iota_\pc\;\gpc, \proj{\delta}{1} \sepidx{} \{c_j\}$ 
      \\\>\>where $c_j$ is $c_1$ if $u_1 = \etrue$ and $c_j$ is $c_2$ if $u_1 = \efalse$ 
      \\\>(6)\>$\iota_\pc\;\gpc, \proj{\delta}{2} \sepidx{} \eif (\iota_2\;u_2)^\glab\ \ethen\ c_1\ \eelse\ c_2 \stepsto 
      (\iota_\pc\labjoin\iota_2)\;(\gpc\cjoin\glab) \rhd \iota_\pc\;\gpc, \proj{\delta}{2} \sepidx{} \{c_k\}$ 
      \\\>\>where $c_k$ is $c_1$ if $u_2 = \etrue$ and $c_k$ is $c_2$ if $u_2 = \efalse$ 
      \\By (2), (3), (4), (5), (6), the conclusion holds
    \end{tabbing}
  \end{description}
\end{proof}

\begin{thm}[Completeness]
  \label{thm:completeness}
  If $\forall i\in \{1, 2\}$, $\proj{\kappa, \delta \sepidx{} c}{i} \stackrel{\trace_i}{\stepsto^*} 
  \kappa_i, \delta_i \sepidx{} \m{skip}$ and $\vdash \kappa, \delta \sepidx{} c \ \m{wf}$, then $\exists \kappa', \delta'$ s.t. 
  $\kappa, \delta \sepidx{} c \stackrel{\trace}{\stepsto^*} \kappa', \delta' \sepidx{} \m{skip}$
  with $\proj{\kappa', \delta' \sepidx{} \m{skip}}{i} = \kappa_i, \delta_i \sepidx{} \m{skip}$ and $\proj{\trace}{i} = \trace_i$
\end{thm}
\begin{proof}
  Suppose $\forall i\in \{1, 2\}$, $\proj{\kappa, \delta \sepidx{} c}{i} \stackrel{\trace_i}{\stepsto^{n_i}} \kappa_i, \delta_i \sepidx{} \m{skip}$. 
  By induction over the number of steps $n_1, n_2$ in the two runs. 

  \begin{description}  
  \item [Base Case:] 
    \begin{tabbing}
      \\\,\,\,\,\=(1) \,\,\,\,\=$\forall i\in \{1, 2\}$, $\proj{\kappa, \delta \sepidx{} c}{i} = \kappa_i, \delta_i \sepidx{} \m{skip}$
      \\By (1)
      \\\>(2)\>$c = \eskip$ and $\proj{\delta}{i} = \delta_i$
      \\By assumption
      \\\>(3)\>$\kappa, \delta \sepidx{} c = \kappa', \delta' \sepidx{} \eskip$
      \\By (2), (3)
      \\\>(4)\>$\proj{\kappa', \delta' \sepidx{} \eskip}{i} = \kappa_i, \delta_i \sepidx{} \eskip$
    \end{tabbing}

  \item [Inductive Case:] 
    \begin{tabbing}
      \\By assumption
      \\\,\,\,\,\=(1)\,\,\,\,\=  $\proj{\kappa, \delta \sepidx{} c}{i} \stackrel{\alpha_i}{\stepsto} \kappa_i', \delta_i' \sepidx{} c_i'$ or $\proj{\kappa, \delta \sepidx{} c}{i}$ ends in $\eskip$ and doesn't step
      \\By Lemma~\ref{lem:proj-complete},
      \\\>(2)\> $\exists \kappa_0, \delta_0, c_0.  \kappa, \delta \sepidx{} c \stackrel{\trace'}{\stepsto^*} \kappa_0, \delta_0 \sepidx{} c_0$
      \\By (1), (2), and Theorem~\ref{thm:soundness},
      \\\>(3)\>$\forall i\in \{1, 2\}$, $\proj{\kappa_0, \delta_0 \sepidx{} c_0}{i} = \kappa_i', \delta_i' \sepidx{} c_i'$ and $\proj{\trace'}{i} = \alpha_i$
      \\Conclusion follows by IH
    \end{tabbing}
    
  \end{description}

\end{proof}

\subsection{Preservation}
\label{sec:app-pres}



We define the following constraints on configurations to facilitate
proofs related to paired values and commands. We start by defining
$\iota\in H(\lab_A)$, $\Pi\in H(\lab_A)$ and $\kappa\in H(\lab_A)$
for any observer at level $\lab_A$. 
\begin{mathpar}
  \inferrule*[right=$\iota$-H]{
    \iota= [\lab_l, \lab_r]  \\ 
    \lab_l\not\labless\lab_A
  }{\iota\in H(\lab_A)
  }
  \and
  \inferrule*[right=$\Pi$-H]{
    \Pi = \epair{\iota_1}{\iota_2}
    \\ \iota_i\in H(\lab_A)~\mbox{where}~i\in\{1,2\}
  }{\Pi\in H(\lab_A)
  }
  \and
  \inferrule*[right=$\kappa$-H]{
    \kappa = \iota\;g \rhd \kappa' 
    \\ \iota \in H(\lab_A)
    \\ (\kappa' \in H(\lab_A) \vee \kappa' = \emptyset)
  }{\kappa\in H(\lab_A)
  }
\end{mathpar}
We say a configuration is safe (written
$\vdash \kappa, \delta \sepidx{i} c\ \m{sf}$ for
$i \in \{\cdot, 1, 2\}$) if all of the following hold
\begin{enumerate}
\item if $i\in\{1,2\}$, then $\kappa\in H(\lab_A)$, 
   $\forall x \in \wtsetof(c)$, $\labof(\delta(x))\in H(\lab_A)$ 
\item if $c=\eif\;\epair{\_}{\_}\;\ethen\; c_1\;\eelse\;c_2$,
  then $\forall x\in\wtsetof(c)$, $\labof(\delta(x))\in H(\lab_A)$
\item if $c=\epair{\kappa_1,\iota_1,c_1}{\kappa_2,\iota_2,c_2}_g$,
  then $\forall i\in\{1,2\}$, 
$\iota_i\vdash g\in H(\lab_A)$,
  and $\forall x\in\wtsetof(c)$, $\labof(\delta(x))\in H(\lab_A)$ 
\end{enumerate}

\begin{lem}[Refinement maintains high label]\label{lem:refine-high}
$\iota_1\sqsubseteq\iota_2$ and $\iota_2\in H(\lab_A)$ imply
$\iota_1\in H(\lab_A)$
\end{lem}
\begin{proofsketch} By examining the definitions of the operations.
\end{proofsketch}

\begin{lem}[Consistent subtyping maintains high]\label{lem:csubtp-high}
$\iota\vdash g\in H(\lab_A)$, $E \vdash g\clabless g'$, and
imply $\iota\bowtie E \vdash g'\in H(\lab_A)$
\end{lem}
\begin{proof}
\begin{tabbing}
\\\quad\= By inversion, $ \iota = [\lab_l, \lab_r]$,
    $\lab_l \not\labless \lab_A$
    $\iota \sqsubseteq \gamma(g)$
\\\> Assume $\iota\bowtie E = (\lab_x, \lab_y)$, 
\\\> By the definition of $\iota\bowtie E$,  $\lab_l\labless \lab_x$
\\\> Therefore, $\lab_x\not\labless \lab_A$
\\\> By Lemma~\ref{lem:evd-subtyping}, $\iota \bowtie E \sqsubseteq \gamma(\glab')$.
\\\> Therefore, $\iota\bowtie E \vdash g'\in H(\lab_A)$
\end{tabbing}
\end{proof}

\begin{lem}[Join maintains high]\label{lem:join-high}
$\iota\vdash g\in H(\lab_A)$, $\iota'\sqsubseteq\gamma(g')$
 imply $\iota\labjoin\iota' \vdash (g\cjoin g')\in H(\lab_A)$
\end{lem}
\begin{proof}
\begin{tabbing}
\\\quad\= By inversion, $ \iota = [\lab_l, \lab_r]$,
    $\lab_l \not\labless \lab_A$,     $\iota \sqsubseteq \gamma(g)$
\\\> Assume $\iota\labjoin\iota' = (\lab_x, \lab_y)$, 
\\\> By the definition of $\iota\labjoin\iota' $,  $\lab_l\labless \lab_x$
\\\> Therefore, $\lab_x\not\labless \lab_A$
\\\> By Lemma~\ref{lem:join-refine}, 
 $\iota \labjoin \iota' \sqsubseteq \gamma(\glab\cjoin\glab')$.
\\\> Therefore, $\iota\labjoin\iota'  \vdash \glab\cjoin\glab'\in H(\lab_A)$
\end{tabbing}
\end{proof}

\begin{lem}[Preservation (cast)]
\label{lem:preservation-cast}
$\Gamma\vdash v: U_1$,  $E \vdash U_1\csubtp U_2$ 
and $U_2=\tau^{\glab'}$ imply $\Gamma\vdash (E, \glab')\rhd v: U_2$.
\end{lem}
\begin{proof}
By examining the definition of $(E, \glab')\rhd v$.
\begin{description}
\item[Case:] $v= (\iota\; u)^\glab$
\begin{tabbing}
By assumptions, 
\\
\quad\= (1)\quad\=  
$\iota' = \iota \bowtie E$ and $(E, \glab')\rhd v = (\iota'\;
u)^{\glab'}$
\\By inversion of typing of $v$
\\\>(2)\> $U_1 = \tau^\glab$, $\iota\sqsubseteq \gamma(\glab)$
\\By inversion of subtyping
\\\>(3)\> $E \vdash \glab\csubtp \glab'$
\\By Lemma~\ref{lem:evd-subtyping} on (1), (3)
\\\>(4)\> $\iota' \sqsubseteq \gamma(\glab')$
\\By applying the same typing rule using (4), the conclusion holds
\end{tabbing}
\item[Case:] $v= \epair{\iota_1\; u_1}{\iota_1\; u_1}^\glab$
\begin{tabbing}
By assumptions, 
\\
\quad\= (1)\quad\=  
$\iota'_i = \iota_i \bowtie E$ ($i\in\{1,2\}$) and $(E, \glab')\rhd v =  \epair{\iota'_1\;u_1}{\iota'_2\;u_2}^{\glab'}$
\\By inversion of typing of $v$, 
\\\>(2)\> $U_1 = \tau^\glab$, 
 and for all $i\in\{1,2\}$ $\iota_i\sqsubseteq \gamma(\glab)$,
 $\iota_i\vdash g\in H(\lab_A)$
\\By inversion of subtyping
\\\>(3)\> $E \vdash \glab\csubtp \glab'$
\\By Lemma~\ref{lem:evd-subtyping} on (1), (3)
\\\>(4)\> $\iota' _i\sqsubseteq \gamma(\glab')$
\\By Lemma~\ref{lem:csubtp-high} 
\\\>(5)\> $\iota'_i\vdash g'\in H(\lab_A)$
\\By applying the same typing rule using (4) and (5), the conclusion holds
\end{tabbing}
\end{description}
\end{proof}

\begin{lem}[Preservation (bop)]
\label{lem:preservation-bop}
If $\forall i\in\{1,2\}$, $\Gamma\vdash v_i: \tau^{\glab_i}$, 
then $\Gamma\vdash v_1\bop v_2 : \tau^{\glab_1\cjoin\glab_2}$.
\end{lem}
\begin{proofsketch}
By examining the definitions of $v_1\bop v_2$. Apply
Lemma~\ref{lem:join-refine} and Lemma~\ref{lem:join-high}.
\end{proofsketch}

\begin{lem}[Preservation (expression)] \label{lem:preservation-exp}
If $\ee::\delta \sepidx{i} e \evalsto v$, $\vdash \delta : \Gamma$, and 
$\Gamma \vdash e$, then $\Gamma \vdash v$
\end{lem}
\begin{proof}
By induction on the structure of $\ee$. The proof is straightforward when $\ee$ ends in
\rulename{P-Const} or \rulename{P-Var}.
\begin{description}
\item[Case:] $\ee$ ends in \rulename{P-Cast}
\begin{tabbing}
By assumptions, 
\\
\quad\= (1)\quad\=  
 $\inferrule*{\delta\sepidx{} e \evalsto v  
  \\ v' = (E, \glab')\rhd v
    }{
      \delta\sepidx{} E^{\glab'} e \evalsto  v'}$
\\\>(2)\> $\vdash \delta : \Gamma$, and  $\Gamma \vdash E^{\glab'} e: U_2$
\\By inversion of typing
\\\>(3)\> $\Gamma\vdash e: U_1 $,  $E \vdash U_1\csubtp U_2$ and $U_2 =\tau^{\glab'}$
\\By I.H. on $e$
\\\>(4)\>$\Gamma\vdash v: U_1$
\\By Lemma~\ref{lem:preservation-cast}, $\Gamma\vdash v': U_2$
\end{tabbing}
\item[Case:] $\ee$ ends in \rulename{P-Bop}
\begin{tabbing}
By assumptions, 
\\
\quad\= (1)\quad\=  
 $\inferrule*{\delta\sepidx{i} e_1 \evalsto v_1 \\
      \delta\sepidx{i} e_2 \evalsto v_2 \\
      v = v_1\bop v_2 
    }{
      \delta\sepidx{i} e_1\bop e_2 \evalsto v
    }$
\\\>(2)\> $\vdash \delta : \Gamma$, and  $\Gamma \vdash e_1\bop e_2 : U$
\\By inversion of typing, $\forall i\in\{1,2\}$
\\\>(3)\> $\Gamma\vdash e_i: \tau^{\glab_i} $,  and $U =\tau^{\glab_1\cjoin\glab_2}$
\\By I.H. on $e_i$
\\\>(4)\>$\Gamma\vdash v_i: \tau^{\glab_i}$
\\By Lemma~\ref{lem:preservation-bop}, $\Gamma\vdash v: U$
\end{tabbing}
\end{description}
\end{proof}

\begin{lem}[PC refinement]\label{lem:pc-refinement}
If $\Gamma ;  \iota_\pc\;\gpc \vdash c$, and
$\iota\sqsubseteq\iota_\pc$, then $\Gamma ;  \iota\;\gpc \vdash c$.
\end{lem}
\begin{proofsketch}
By induction over the typing derivation of $c$.
\end{proofsketch}

\begin{lem}\label{lem:preservation-refinelv}
If $\vdash v: \tau^\glab$, $\reflvof{i}(\Pi, v) = v'$,
and when $i\in\{1,2\}$ or $\Pi = \epair{\iota}{\iota'}$, $\Pi\in H(\lab_A)$, 
then $\vdash v': \tau^\glab$ and $\labof(v') \in H(\lab_A)$ when $i\in{1,2}$.
\end{lem}
\begin{proof}
By examining all the rules. Use Lemma~\ref{lem:interval-refine} to
show that the resulting intervals are still H.
\end{proof}

\begin{lem}\label{lem:preservation-refinevl}
\label{lem:preservation-updval}
If $\vdash v: \tau^\glab$, $\updval{}(\iota, v) = v'$,
then $\vdash v': \tau^\glab$ and $\labof(v')\sqsubseteq\iota$. 
\end{lem}
\begin{proofsketch} 
By examining all the rules. Use Lemma~\ref{lem:bowtie-refine}. 
\end{proofsketch}

\begin{lem}\label{lem:preservation-rfl}
If $\vdash \delta: \Gamma$ and $\rflof{i}(\delta, X, \Pi) = \delta'$
and when $i\in\{1,2\}$ or $\Pi = \epair{\iota}{\iota'}$, $\Pi\in H(\lab_A)$, 
then $\vdash \delta': \Gamma$ and 
when $\Pi\in H(\lab_A)$, $\forall x\in X$, $\labof(\delta'(x))\in  H(\lab_A)$
\end{lem}
\begin{proof}
By induction over the size of $X$ and apply Lemma~\ref{lem:preservation-refinelv}.
\end{proof}


\begin{lem}\label{lem:preservation-upd}
If $\vdash v_n: \tau^\glab$, $\vdash v_o: \tau^\glab$,
and when $i\in\{1,2\}$, $\labof(v_o) \in H(\lab_A)$
 and $\labof(v_n) \in H(\lab_A)$, 
then $\vdash \eupdate_i\ v_o\  v_n : \tau^\glab$
\end{lem}
\begin{proof} 
By examining the definition of $\eupdate_i\ v_o\  v_n$
\begin{description}
  \item [Case:] $\eupdate \ v_o\  v_n$ where either $v_o$ or $v_n$ is
    a pair
    \begin{tabbing}
      By inversion of typing rules,  $\labof(v_o)\in H(\lab_A)$ 
       or  $\labof(v_n)\in H(\lab_A)$
       \\By Lemma~\ref{lem:interval-refine},  Lemma~\ref{lem:bowtie-refine}, and 
        Lemma~\ref{lem:refine-high}, 
          $\labof(\eupdate \ v_o\  v_n)\in H(\lab_A)$
      \\ Therefore \rulename{V-Pair} applies and $\vdash \eupdate\ v_o\  v_n : \tau^\glab$
    \end{tabbing}

\item [Case:] $\eupdate \ v_o\  v_n$ where neither $v_o$ nor $v_n$ is
    a pair 
    \begin{tabbing}
      We use Lemma~\ref{lem:interval-refine},  Lemma~\ref{lem:bowtie-refine}, and transitivity of
      $\sqsubseteq$ and apply value typing rule directly.
    \end{tabbing}

\item[Case:] $\eupdate_i \ v_o\  v_n$, where $i\in\{1, 2\}$
\begin{tabbing}
 By assumption, $\labof(v_o) \in H(\lab_A)$ and $\labof(v_n) \in
 H(\lab_A)$
\\ By Lemma~\ref{lem:interval-refine},  Lemma~\ref{lem:bowtie-refine}, and 
        Lemma~\ref{lem:refine-high},  the updated value's interval is
        in $H(\lab_A)$
\\  Therefore \rulename{V-Pair} applies and $\vdash \eupdate_i\ v_o\  v_n : \tau^\glab$
\end{tabbing}
\end{description}

\end{proof}

\begin{lem}[Preservation one-step]
  \label{lem:pres-onestep}
If $\ee::\kappa, \delta \sepidx{i} c \stackrel{\alpha}{\stepsto}
\kappa', \delta' \sepidx{i} c'$, 
$\vdash \delta : \Gamma$ and
$\de::\Gamma; \kappa \vdash_r c$, 
and $\vdash \kappa, \delta\sepidx{i} c\ \m{sf}$ 
then  $\vdash \delta' : \Gamma$, 
$\Gamma; \kappa' \vdash_r c'$, $\vdash \alpha$, and 
 $\vdash \kappa', \delta'\sepidx{i} c'\ \m{sf}$ 
\end{lem}
\begin{proof}
  By induction on the structure of $c$ 
  \begin{description}
  \item [Case:] $c = \m{skip}$ 
      \\There is no rule to step $\m{skip}$, so the conclusion holds trivially.
  \item [Case:] $c = c_1;c_2$ 
   \begin{tabbing}
     By inversion of $\de$
      \\ \quad\= (1)~~\= $\kappa = \kappa'\rhd \iota_\pc\;\glab_\pc$,
    $\Gamma; \iota_\pc\;\glab_\pc\vdash c_2$, and $\Gamma; \kappa
    \vdash_r c_1$ or ($\Gamma; \iota_\pc\;\glab_\pc \vdash c_1$ 
    and $\kappa= \iota_\pc\;\glab_\pc$)
      \\ By examining $\ee$, there are two cases: \rulename{P-Seq} or
     \rulename{P-Skip} applies
     \\{\bf Subcase I:} $\ee$ ends in \rulename{P-Seq}
     \\\> By assumption
     \\\>\>(I2)\quad\= $\inferrule*{
      \kappa,\delta\sepidx{i} c_1  \stackrel{\alpha}{\stepsto}  \kappa',\delta' \sepidx{i} c'_1
    }{
      \kappa,\delta\sepidx{i} c_1; c_2
      \stackrel{\alpha}{\stepsto}  \kappa',\delta'\sepidx{i} c'_1; c_2
    }$
    \\\> By I.H. on $c_1$
    \\\>\>(I3)\>$\vdash \delta' : \Gamma$, 
    $\Gamma; \kappa' \vdash_r c'_1$  and $\vdash \alpha$
    \\\> By Lemma~\ref{lem:kappa-inv}, (I3), (1), either
    \rulename{R-End} or \rulename{R-C-Seq} applies
    \\\>\>(I4)\> $\Gamma; \kappa' \vdash_r c'_1; c_2$  
     \\{\bf Subcase II:} $\ee$ ends in \rulename{P-Skip}
     \\\>\>(II2)\>$\kappa, \delta \sepidx{i} \m{skip}; c_2 
       \stepsto \kappa, \delta \sepidx{i} c_2$
       \\\>By (1) and $c_1=\m{skip}$
       \\\>\>(II3)\> $\kappa = \iota_\pc\;\glab_\pc$
       \\\>By (1) and (II3) and \rulename{R-End}
       \\\>\>(II4)\>  $\Gamma; \kappa\vdash_r c_2$
    \end{tabbing}

 
  \item [Case:] $c = x\, :=\,e$ 
\begin{tabbing}
     By inversion of $\de$
      \\ \quad\= (1)~~\= $\kappa = \iota_\pc\;\gpc$,
      $\Gamma \vdash x : \tau^\glab$,
      $\Gamma \vdash e : \tau^\glab$, 
      and $\gpc\clabless \glab$
\\ By examining $\ee$, only \rulename{P-Assign} applies
  \\\>(2)\> $\inferrule*{
      \delta\sepidx{i} e \evalsto v' 
      \\ v'' = \reflvof{}(\iota_\pc, v')
      \\ \delta' = \delta[x\mapsto \eupdate_i\ \delta(x)\  v'']
    }{
      \iota_\pc\;\gpc,\delta \sepidx{i} x:= e
      \stepsto  \iota_\pc\;\gpc, 
      \delta' \sepidx{i} \m{skip}
    }$
\\By Lemma~\ref{lem:preservation-exp},
$\Gamma \vdash v' : \tau^\glab$, 
\\By Lemma~\ref{lem:preservation-refinelv},
$\Gamma \vdash v'' : \tau^\glab$, and $\labof(v'') \in H(\lab_A)$ when
$i\in{1,2}$
\\By assumption, $\labof(\delta(x)) \in H(\lab_A)$
\\By Lemma~\ref{lem:preservation-upd}, 
$\Gamma \vdash \eupdate_i\ \delta(x)\  v'' : \tau^\glab$
\\By store typing $\vdash \delta':\Gamma$
\end{tabbing}


  \item [Case:] $c = \eoutput(\lab, e)$
\begin{tabbing}
     By inversion of $\de$
      \\ \quad\= (1)~~\= $\kappa = \iota_\pc\;\gpc$,
   $\Gamma \vdash e: \tau^\glab$, 
   $ \glab\clabless\lab$, 
   $\gpc \clabless\lab$
 \\ By examining $\ee$, only \rulename{P-Out} applies
 \\\>(2)\> $\inferrule*{
      \delta\sepidx{i} e \evalsto v' 
      \\ v'' = \reflvof{}(\iota_\pc, v')
      \\ v_1 = \updval{}([\lab,\lab], v'')
    }{
      \iota_\pc\;\gpc,\delta \sepidx{i} \eoutput(\lab, e)
      \stackrel{(i, \lab, v_1)}{\stepsto}  \iota_\pc\;\gpc, 
      \delta\sepidx{i} \eskip
    }$
\\By Lemma~\ref{lem:preservation-exp}
\\\>(3)\>$\Gamma \vdash v': \tau^\glab$
\\By Lemma~\ref{lem:preservation-refinelv}
\\\>(4)\>$\Gamma \vdash v'': \tau^\glab$  and
$\labof(v'')\in H(\lab_A)$ when $i\in\{1,2\}$
\\By Lemma~\ref{lem:preservation-updval}
\\\>(5)\>$\Gamma \vdash v_1: \tau^\glab$ and
$\labof(v_1)\sqsubseteq[\lab,\lab]$
\\ By (5)
\\\>(6)\> $\labof(v_1) \labless [\lab, \lab]$
\\By (4), $\labof(v_1)\sqsubseteq\labof(v'')$, and Lemma~\ref{lem:refine-high}
\\\>(7)\> $\labof(v_1)\in H(\lab_A)$ when $i\in\{1,2\}$
\\By \rulename{T-A-*}, $\vdash (i, \lab, v_1)$
 \end{tabbing}


  \item [Case:]  $c = \eif^X e\, \ethen\, c_1\, \eelse \, c_2$
\begin{tabbing}
  By inversion of $\de$
  \\ \quad\= (1)~~\= $\kappa = \iota_\pc\;\gpc$,
  $\Gamma  \vdash e :  \tbool^{\glab_c}$, 
  $\iota_c = \gamma(\glab_c)$,
  \\\>(2)\>  $\Gamma ;  \iota_\pc\labjoin\iota_c\;\gpc\cjoin \glab_c \vdash c_i$,
  where $i\in\{1,2\}$ and 
  \\\>(3)\> $X = \wtsetof(c_1)\cup \wtsetof(c_2)$
  \\ By examining $\ee$, only \rulename{P-If-Refine} applies
  \\\>(4)\>$\inferrule*{
    \delta\sepidx{i} e \evalsto v \\
    \delta' = \rflof{i}(\delta, X, \iota_\pc\labjoin\labof(v))
  }{
    \iota_\pc\;\gpc,\delta\sepidx{i} \eif^X\;e\ \ethen\ c_1\ \eelse\ c_2
    \stepsto  \iota_\pc\;\gpc, \delta' \sepidx{i} \eif\;v\ \ethen\ c_1\ \eelse\ c_2
  }$
  \\By Lemma~\ref{lem:preservation-exp}
  \\\>(5)\> $\Gamma  \vdash v :  \tbool^{\glab_c}$, 
  \\By (1) and \rulename{R-C-If} 
  \\\>(6)\> $\Gamma; \iota_\pc\;\gpc\vdash_r \eif\;v\ \ethen\
  c_1\ \eelse\ c_2$
  \\By Lemma~\ref{lem:preservation-rfl}, $\vdash \delta':
  \Gamma$
  \\ By  $\vdash \iota_\pc\;\gpc, \delta \sepidx{i}  c: \m{sf}$
  \\\>(7)\> $\iota_\pc\in H(\lab_A)$ when $i\in\{1,2\}$
  \\ By Lemma~\ref{lem:preservation-rfl}, 
  \\\>(8)\> when $i\in\{1,2\}$ or $\labof(v)\in H(\lab_A)$,  $\forall x\in X$, $\labof(\delta'(x))\in  H(\lab_A)$. 
  \\By $\m{sf}$ definition,  
 $\vdash \iota_\pc\;\gpc, \delta'\sepidx{i} \eif\;v\ \ethen\ c_1\ \eelse\ c_2\ \m{sf}$ 
\end{tabbing}


  \item [Case:]  $c = \eif\, v\, \ethen\, c_1\, \eelse \, c_2$
 \begin{tabbing}
     By inversion of $\de$
      \\ \quad\= (1)~~\= $\kappa = \iota_\pc\;\gpc$,
    $\Gamma  \vdash e :  \tbool^{\glab_c}$, 
    $\iota_c = \gamma(\glab_c)$,
  $\Gamma ;  \iota_\pc\labjoin\iota_c\;\gpc\cjoin \glab_c \vdash c_i$,
  where $i\in\{1,2\}$
      \\ By examining $\ee$, there are three cases:
      \rulename{P-Lift-If}, \rulename{P-If} applies
      \\{\bf Subcase I:} $v = (\iota\;\etrue)^\glab$ 
     \\\>By assumption, $v = (\iota\;\etrue)^\glab$
      \\\>\>(I2)\quad\= $\inferrule*{
      \\ \iota'_\pc = \iota_\pc\labjoin \iota
      \\ \gpc' =\gpc\cjoin\glab
    }{
      \iota_\pc\;\gpc,\delta\sepidx{i} \eif\; (\iota\;\etrue)^\glab\ \ethen\ c_1\ \eelse\ c_2
      \stepsto  \iota'_\pc\;\gpc'\rhd \iota_\pc\;\gpc, \delta \sepidx{i}  \{c_1\}
    }$
      \\\>By inversion of typing for $e$
      \\\>\>(I3)\> $g_c = g$ 
      \\\>By Lemma~\ref{lem:pc-refinement} and (I3), and (1)
      \\\>\>(I4)\> $\Gamma ;  \iota_\pc\labjoin\iota'_c\;\gpc\cjoin
      \glab_c \vdash c_1$
      \\\>By (I4) and \rulename{R-Pop}, $\Gamma; \iota'_\pc\;\gpc'\rhd
      \iota_\pc\;\gpc \vdash \{c_1\}$.
     
     \\{\bf Subcase II:} $v = (\iota\;\efalse)^\glab$, the proof
     is similar to the previous case.

      \\{\bf Subcase III:} $\ee$ ends in \rulename{P-Lift-If},
      \\\>By assumption, $v = \epair{\iota_1\;u_1}{\iota_2\;u_2}^\glab$
      \\\> \>(III2)\> $\inferrule*{
        i = \{1, 2\}
        \\ c_j = c_1~\mbox{if}~u_1 = \etrue
        \\ c_j = c_2~\mbox{if}~u_1 = \efalse
        \\\\ c_k = c_1~\mbox{if}~u_2 = \etrue
        \\ c_k = c_2~\mbox{if}~u_2 = \efalse
      }{
        \iota_\pc\;\gpc,\delta\sepidx{}
        \eif\; \epair{\iota_1\;u_1}{\iota_2\;u_2}^\glab
        \ \ethen\ c_1\ \eelse\ c_2
        \stepsto  \iota_\pc\;\gpc,\delta\sepidx{} 
        \epair{\emptyset, \iota_{1}, c_j}{
          \emptyset, \iota_{2}, c_k}_g
      }$
      \\\>By inversion of typing for $v$
      \\\>\>(III3)\> $g=g_c$, $\iota_i = \gamma(\glab)$, 
       and $\iota_i\vdash g\in H(\lab_A)$
       \\\>By similar arguments used in the previous subcase,
       \\\>\>(III4)\> $\Gamma ;  \iota_\pc\labjoin\iota_i\;\gpc\cjoin
       \glab_c \vdash c_i$
       \\\>By (III3), (III4) and \rulename{R-C-Pair}, 
  $\Gamma; \iota_\pc\;\gpc \vdash 
        \epair{\emptyset, \iota_{1}, c_j}{
          \emptyset, \iota_{2}, c_k}_g$
        \\\> By $\vdash \iota_\pc\;\gpc, \delta \sepidx{i}  c: \m{sf}$
        \\\>\>(III5)\> $\forall x\in \wtsetof(c_i)$,
        $\labof(\delta(x))\in  H(\lab_A)$
        \\\> By $\iota_i \in H(\lab_A)$,  
        \\\>\>(III6)\>$\iota_i \in H(\lab_A)$, 
        \\\> By (III5) and (III6), $\vdash \iota_\pc\;\gpc, \delta
        \sepidx{} \epair{\emptyset, \iota_{1}, c_j}{
          \emptyset, \iota_{2}, c_k}_g: \m{sf}$
   \end{tabbing}


  \item [Case:]  $c = \ewhile^X e\, \edo\, c'$
   \begin{tabbing}
     By inversion of $\de$
      \\ \quad\= (1)~~\= $\kappa = \iota_\pc\;\gpc$,
    $\Gamma  \vdash e :  \tbool^\glab$, 
    $\iota_c = \gamma(\glab)$,
  $\Gamma ;  \iota_\pc\labjoin\iota_c\;\gpc\cjoin \glab \vdash c'$,
   $X = \wtsetof(c')$
      \\ By examining $\ee$, only \rulename{P-While} applies
      \\\>(2)\>$\iota_\pc\;\gpc,\delta\sepidx{i} \ewhile^X\;e\ \edo\ c'
      \stepsto   \iota_\pc\;\gpc, \delta \sepidx{i}
      \eif^X\;e\ \ethen\ (c'; \ewhile^X\;e\ \edo\ c')\ \eelse\ \m{skip}$
      \\To type the resulting if statement, we need to show the following:
      \\\>\> $\Gamma ;  \iota_\pc\labjoin\iota_c\;\gpc\cjoin \glab 
                          \vdash \ewhile^X\;e\ \edo\ c'$
       \\Because expression typing does not use pc context and
       \\ $\labjoin$ ($\cjoin$) the same interval (label) twice does not
       change the result, 
       \\the conclusion holds
   \end{tabbing}

  \item [Case:]  $c = \{c'\}$ 
    \begin{tabbing}
     By inversion of $\de$
      \\ \quad\= (1)~~\= $\kappa = \kappa_1\rhd \iota_\pc\;\gpc$,
    $\Gamma; \kappa_1\vdash c'$ 
      \\ By examining $\ee$, there are two cases: \rulename{P-Pc} or
     \rulename{P-Pop} applies
     \\{\bf Subcase I:} $\ee$ ends in \rulename{P-Pop}
     \\\>By $c'= \m{skip}$ and (1)
      \\\>\>(I2)\quad\=$\kappa_1 = \iota\; \glab$
     \\\> By assumption
     \\\>\>(I3)\> $\inferrule*{   }{
      \iota\;\glab\rhd \iota_\pc\;\gpc,\delta\sepidx{i} \{\eskip\} 
      \stepsto \iota_\pc\;\gpc,\delta\sepidx{i} \eskip
      }$
      \\\>By \rulename{G-Skip} and \rulename{R-End}
      \\\>\>(I4)\> $\Gamma; \iota_\pc\;\gpc\vdash_r \m{skip}$
   
     \\{\bf Subcase II:} $\ee$ ends in \rulename{P-Pc}
     \\\> By assumption
     \\\>\>(II2)\> $\inferrule*{\kappa,\delta\sepidx{i} c 
       \stackrel{\alpha}{\stepsto}  \kappa',\delta'\sepidx{}c' }{
       \kappa\rhd \iota_\pc\;\gpc,\delta\sepidx{i} \{c\} 
       \stackrel{\alpha}{\stepsto}  \kappa'\rhd\iota_\pc\;\gpc,\delta'\sepidx{i}\{c'\}
           }$
      \\\>By I.H. on $c$
      \\\>\>(II3)\> $\vdash \delta' : \Gamma$, 
      $\Gamma; \kappa' \vdash_r c'$  and $\vdash \alpha$
    \\\> By \rulename{R-Pop}
    \\\>\>(II4)\> $\Gamma; \kappa' \rhd \iota_\pc\;\gpc\vdash_r \{c'\}$  
    \end{tabbing}
    

  \item [Case:]   $c=\epair{\kappa_1, \iota_1, c_1}{\kappa_2, \iota_2, c_2}_\glab$ 
\begin{tabbing}
     By inversion of $\de$
      \\ \quad\= (1)~~\= $\kappa = \iota_\pc\;\gpc$,
       $\Gamma;\kappa_i\rhd(\iota_\pc\labjoin\iota_i)\; (\gpc\cjoin g)
       \vdash_r c_i$,  and $\iota_i\vdash \glab\in\ H$, for  $i\in\{1,2\}$
      \\ By examining $\ee$, there are two cases: \rulename{P-Skip-Pair} or
     \rulename{P-C-Pair} applies
     \\{\bf Subcase I:} $\ee$ ends in \rulename{P-Skip-Pair}
      \\\>\>(I2)\quad\=$\kappa_i = \emptyset$
     \\\> By assumption
     \\\>\>(I3)\> $\inferrule*{ }{
      \iota_\pc\; \gpc, \delta\sepidx{} \epair{\emptyset, \iota_1,
        \eskip }{\emptyset, \iota_2, \eskip}_g    
      \stepsto \iota_\pc\; \gpc,\delta\sepidx{} \eskip }$
      \\\>By \rulename{G-Skip} and \rulename{R-End}
      \\\>\>(I4)\> $\Gamma; \iota_\pc\;\gpc\vdash_r \m{skip}$
   
     \\{\bf Subcase II:} $\ee$ ends in \rulename{P-C-Pair}
     \\\> By assumption
     \\\>\>(II2)\> $\inferrule* {
      \kappa_i\rhd \iota_\pc\labjoin\iota_i\;\gpc\cjoin\glab,\delta
       \sepidx{i} c_i  \stackrel{\alpha}{\stepsto}  
        \kappa'_i\rhd \iota_\pc\labjoin\iota_i\;\gpc\cjoin\glab,\delta'\sepidx{i}c'_i
      \\\\ c_j = c'_j
      \\ \kappa_j = \kappa'_j
      \\ \{i,j\} = \{1,2\}
    }{
      \iota_\pc\;\gpc,\delta\sepidx{} 
      \epair{\kappa_1, \iota_1, c_1}{\kappa_2, \iota_2, c_2}_\glab
      \stackrel{\alpha}{\stepsto}  \iota_\pc\;\gpc,\delta'\sepidx{} 
      \epair{\kappa'_1, \iota_1, c'_1}{\kappa'_2, \iota_2, c'_2}_{\glab}
    }$
      \\\>Assume $c_1$ takes a step. The other case when $c_2$ takes a
      step can be proven similarly.
      \\ By I.H. on $c_1$
      \\\>\>(II3)\> $\vdash \delta' : \Gamma$, 
      $\Gamma; \kappa'  \rhd \iota_\pc\labjoin\iota_i\;\gpc\cjoin\glab\vdash_r c'_1$  
       and $\vdash \alpha$
    \\\> By \rulename{R-C-Pair}, (II3), and (1)
    \\\>\>(II4)\> $\Gamma; \iota_\pc\;\gpc
     \vdash_r \epair{\kappa'_1, \iota_1, c'_1}{\kappa'_2, \iota_2, c'_2}_{\glab}$
    \end{tabbing}
  \end{description}
\end{proof}

\begin{lem}
  \label{lem:preservation}
If $\kappa, \delta \sepidx{} c \stackrel{\alpha}{\stepsto} \kappa', \delta' \sepidx{} c'$
with $\vdash \kappa, \delta, c$ and $\vdash \kappa, \delta \sepidx{} c \ \m{sf}$, then
$\vdash \kappa', \delta', c'$ and $\vdash \kappa', \delta' \sepidx{} c' \ \m{sf}$ 
and $\vdash \alpha$
\end{lem}
\begin{proof}
  Follows from Lemma~\ref{lem:pres-onestep}
\end{proof}

\begin{thm}[Preservation]
  \label{thm:preservation}
  If $\kappa, \delta \sepidx{} c \stackrel{\trace}{\stepsto^*} \kappa', \delta' \sepidx{} c'$
with $\vdash \kappa, \delta, c$ and $\vdash \kappa, \delta \sepidx{} c \ \m{sf}$, then $\vdash \kappa', \delta', c'$
and $\vdash \trace$
\end{thm}
\begin{proof}
  By induction on the number of steps in the sequence. Base case is trivial and follows from assumption.

\noindent Inductive Case: Holds for $n$ steps; To show for $n+1$ steps, i.e.,
  if $\kappa, \delta \sepidx{} c \stackrel{\trace}{\stepsto^n} \kappa', \delta' \sepidx{} c' \stackrel{\alpha}{\stepsto} \kappa'', \delta'' \sepidx{} c'' $, then $\vdash \kappa'', \delta'', c''$ and $\vdash \trace, \alpha$
    \begin{tabbing}
      \,\,\,\,\=(IH) \,\,\,\,\=If $\kappa, \delta \sepidx{} c \stackrel{\trace}{\stepsto^n} \kappa', \delta' \sepidx{} c'$
      with $\vdash \kappa, \delta, c$, then $\vdash \kappa', \delta', c'$ and $\vdash \trace$
      \\By Lemma~\ref{lem:preservation} and IH
      \\\>(1)\>$\vdash \kappa'', \delta'', c''$ and $\vdash \trace$ and $\vdash \alpha$
      \\By \rulename{T-T-Ind}, conclusion holds
    \end{tabbing}
\end{proof}

\subsection{Noninterference}
\label{sec:app-ni}
We first define when a gradual label of an initial store location is not
observable by the attacker. Formally: $g\in H(\lab_A)$ iff $g=\lab$ and $\lab\not\labless\lab_A$. 

We define merging of two stores ($\delta_1 \bowtie \delta_2$) as below:
\begin{mathpar}
  \inferrule*[right=MgS-Emp]{ }{
    \Gamma \vdash \cdot \bowtie \cdot = \cdot
  }
  \and
  \inferrule*[right=MgS-H]{
    \Gamma \vdash \delta_1 \bowtie \delta_2 = \delta 
    \\  \mathit{lab}(\Gamma(x)) \in H(\lab_A)
    \\ v_i = (\iota_i\; u_i)^\glab (i\in\{1,2\}
  }{
    \Gamma \vdash \delta_1, x \mapsto v_1 \bowtie \delta_2, x \mapsto v_2
    = \delta, x \mapsto \epair{\iota_1\;u_1}{\iota_2\; u_2}^\glab
  }
  \and
  \inferrule*[right=MgS-L]{ \Gamma \vdash \delta_1 \bowtie \delta_2 = \delta \\
    \mathit{lab}(\Gamma(x)) \not\in H(\lab_A) \\ v_1 = v_2 = v}{
    \Gamma \vdash \delta_1, x \mapsto v_1 \bowtie \delta_2, x \mapsto v_2
    = \delta, x \mapsto v
  }
\end{mathpar}

\begin{lem}\label{lem:trace-proj-eq}
  If  
  $\vdash \trace $
  then $\vdash \proj{\trace}{1}\approx_{\lab_A}
  \proj{\trace}{2}$ 
\end{lem}
\begin{proof}
  By induction on the structure of $\trace$. The base case is trivial. 
  \begin{description}
  \item[Case:] $\trace=\alpha,\trace'$
    \begin{tabbing}
      By inversion of $\vdash\trace$
      \\
      \quad\= (1)\quad\=   $\vdash \alpha$ and $\vdash \trace'$
      \\By I.H. on $\trace'$
      \\\>(2)\> $\vdash \proj{\trace'}{1}\approx_{\lab_A}  \proj{\trace'}{2}$ 
      \\By inversion of $\vdash \alpha$, we have two cases
      \\{\bf Subcase I:} $\alpha = (\lab, v)$
      \\\>By projection rules
      \\\>\>(I3)\quad\= $\proj{\alpha, \trace'}{1}= \alpha,
      \proj{\trace'}{1}$, and $\proj{\alpha, \trace'}{2}= \alpha,
      \proj{\trace'}{2}$
      \\\> By applying either \rulename{EqT-L} or \rulename{EqT-H-l} then
      \rulename{EqT-H-r} on (2), 
      $\vdash \proj{\trace}{1}\approx_{\lab_A} \proj{\trace}{2}$ 
      \\{\bf Subcase II:} $\alpha = (i, \lab, v)$, we show the case when
      $i=1$, the other case when $i=2$ is similar
      \\\> By projection rules
      \\\>\>(II3)\> $\proj{\alpha, \trace'}{1}= (\lab,v),
      \proj{\trace'}{1}$, and $\proj{\alpha, \trace'}{2}= \proj{\trace'}{2}$
      \\\> By typing of $\alpha$
      \\\>\>(II4)\>$\vdash \lab \not\labless \lab_A$
      \\\>By (II4) and \rulename{EqT-H-l} and (2), $\vdash \proj{\trace}{1}\approx_{\lab_A} \proj{\trace}{2}$ 
    \end{tabbing}
  \end{description}
\end{proof}

\begin{lem}
  \label{lem:store-projection}
  If $\vdash \delta_1 \approx_{\lab_A} \delta_2 : \Gamma$ then $\exists
  \delta$ s.t. $\Gamma \vdash\delta_1 \bowtie \delta_2 = \delta$, $\vdash \delta : \Gamma$ and
  $\proj{\delta}{i} = \delta_i$
\end{lem}
\begin{proof}
  By induction on the structure of the equivalence definition. Base case is trivial

\noindent Inductive Case: \rulename{EqS-Ind}. 
    \begin{tabbing}
      By assumption
      \\
      \quad\= (1)\quad\=  $\inferrule*{ 
        \ee_1:: \vdash \delta_1 \approx_{\lab_A} \delta_2: \Gamma
        \\ \ee_2::\vdash v_1\approx_{\lab_A} v_2 : U
      }{
        \vdash \delta_1, x\mapsto v_1 \approx_{\lab_A} \delta_2, x\mapsto
        v_2 :\Gamma, x:U
      }$
      \\ By I.H. on $\ee_1$
      \\\>(2)\> $\exists
      \delta$ s.t. $\Gamma \vdash\delta_1 \bowtie \delta_2 = \delta$, $\vdash \delta : \Gamma$ and
      $\proj{\delta}{i} = \delta_i$
      \\By inversion of $\ee_2$, there are two subcases:
      \\{\bf Subcase I:} $\ee_2$ ends in \rulename{EqV-L}
      \\\> By assumption
      \\\>\>(I3)\quad\= $v_i = (\iota\; u)^\glab$, $\glab\clabless
      \lab_A$,  $\iota\sqsubseteq \gamma(\glab)$ and $\vdash (\iota\;
      u)^\glab : U$ 
      \\\> By casing on $g$ and (I3), 
      \\\>\>(I4)\> $g\not\in H(\lab_A)$
      \\\>By \rulename{MgS-L} on (I3), (I4), (2)
      \\\>\>(I5)\> $\Gamma, x:U \vdash\delta_1, x\mapsto v_1 \bowtie
      \delta_2,  x\mapsto v_1 = \delta, x\mapsto (\iota\; u)^\glab$, 
      \\\>By \rulename{T-S-Ind}, (2), and (I3)
      \\\>\>(I6)\> $\vdash \delta, x\mapsto (\iota\; u)^\glab: \Gamma, x: U$ 
      \\{\bf Subcase II:} $\ee_2$ ends in \rulename{EqV-H}
      \\\> By assumption
      \\\>\>(II3)\quad\= $v_i = (\iota_i\; u_i)^\glab$, $\iota_i \vdash
      \glab\in H(\lab_A)$,  and $\vdash (\iota_i\;u_i)^\glab : U$  where $i\in\{1,2\}$
      \\\>By \rulename{MgS-H} on (II3),  (2)
      \\\>\>(II4)\> $\Gamma, x:U \vdash\delta_1, x\mapsto v_1 \bowtie
      \delta_2,  x\mapsto v_1 = \delta, x\mapsto \epair{\iota_2\;u_2}{\iota_2\;u_2}^\glab$, 
      \\\>By \rulename{T-S-Ind}, (2), and (II3)
      \\\>\>(II5)\> $\vdash \delta, x\mapsto \epair{\iota_1\;u_1}{\iota_2\;u_2}^\glab: \Gamma, x: U$ 
    \end{tabbing}
\end{proof}



\begin{thm}[Noninterference]
  Given an adversary label $\lab_A$, a program $c$, and two stores $\delta_1$,
  $\delta_2$, s.t., $\vdash \delta_1\approx_{\lab_A}\delta_2: \Gamma$,
  and $\Gamma; [\bot, \bot]\; \bot\vdash c$, and $\forall i\in\{1,2\}$, $[\bot,\bot]\; \bot,
  \delta_i \sepidx{} c \stackrel{\trace_i}{\stepsto^*}
  \kappa_i,\delta'_i\sepidx{}\m{skip}$, then $\vdash \trace_1\approx_{\lab_A}\trace_2$.
\end{thm}
\begin{proof}
  ~
  \begin{tabbing}
    \quad\=By assumptions and Lemma~\ref{lem:store-projection}
    \\\>
    \quad\= (1)\quad\= $\exists\delta$ s.t. $\Gamma \vdash \delta_1
    \bowtie \delta_2 = \delta$, $\vdash \delta : \Gamma$ and
    $\proj{\delta}{i} = \delta_i$
    \\\>By Completeness Theorem (Theorem~\ref{thm:completeness}) 
    \\\>\>(2)\>  $\exists\kappa'$, $\delta'$, $\trace$, s.t. 
    $[\bot,\bot]\; \bot,
    \delta \sepidx{} c \stackrel{\trace}{\stepsto^*}
    \kappa',\delta' \sepidx{}\m{skip}$ 
    \\\> By Soundness Theorem (Theorem~\ref{thm:soundness})
    \\\>\>(3)\> $\proj{[\bot,\bot]\; \bot, \delta \sepidx{} c}{i}
    \stackrel{\proj{\trace}{i}}{\stepsto^*}
    \proj{\kappa,\delta'\sepidx{}\m{skip}}{i}$
    \\\>By the operational semantic rules are deterministic and (3)
    \\\>\>(4)\> $\trace_i = \proj{\trace}{i}$
    \\\>By assumption and (1) and \rulename{T-Conf}
    \\\>\>(5)\> $\vdash [\bot, \bot]\; \bot, \delta, c$
    \\\>By Preservation Theorem (Theorem~\ref{thm:preservation}), (5), and
    (2)
    \\\>\>(6)\> $\vdash \trace$
    \\\>By Lemma~\ref{lem:trace-proj-eq} and (6)
    \\\>\>(7)\> $\vdash \proj{\trace}{1} \approx_{\lab_A} \proj{\trace}{2}$
    \\\>By (4) and (7), $\vdash \trace_1 \approx_{\lab_A} \trace_2$
  \end{tabbing}
\end{proof}

\subsection{Gradual Guarantees}
\label{sec:app-gg}
\begin{lem}[Operations are closed under refinement]~\label{lem:closed-under-refinement}
  \begin{enumerate}[1]
  \item if  $\iota_1\sqsubseteq \iota'_1$ and $\iota_2\sqsubseteq \iota'_2$,
    then $\iota_1\bowtie\iota_2 \sqsubseteq \iota'_1\bowtie\iota'_2$.
  \item $\iota_1\sqsubseteq\iota'_1$ and $\iota_2\sqsubseteq\iota'_2$,
    then $\iota_1\labjoin\iota_2\sqsubseteq\iota'_1\labjoin\iota'_2$.
  \item $\iota_1\sqsubseteq\iota'_1$ and $\iota_2\sqsubseteq\iota'_2$,
    then $\refineof(\iota_1,\iota_2)\sqsubseteq
    \refineof(\iota'_1,\iota'_2)$.
  \item if $\iota_1\sqsubseteq \iota'_1$ and $ E\sqsubseteq E'$,
    then $\iota_1\bowtie E \sqsubseteq \iota'_1\bowtie E'$.
  \item $\glab_1\sqsubseteq\glab'_1$ and $\glab_2\sqsubseteq\glab'_2$,
    then $\glab_1\cjoin\glab_2\sqsubseteq\glab'_1\cjoin\glab'_2$.
  \end{enumerate}
\end{lem}
\begin{proof}

  \begin{tabbing}
    \\By assumptions 
    \\
    \quad\= (1)\quad\=   
    $\iota_1 = [\lab_{1l}, \lab_{1r}]$, $\iota'_1 = [\lab'_{1l}, \lab'_{1r}]$,
    $\lab'_{1l}\labless \lab_{1l}$, and $\lab_{1r}\labless \lab'_{1r}$
    \\\>(2)\> $\iota_2 = [\lab_{2l}, \lab_{2r}]$, $\iota'_2 = [\lab'_{2l}, \lab'_{2r}]$,
    $\lab'_{2l}\labless \lab_{2l}$, and $\lab_{2r}\labless \lab'_{2r}$
    \\By (1) and (2)
    \\\>(3)\> $\lab'_{1l}\labmeet\lab'_{2l}\labless
    \lab_{1l}\labmeet\lab_{2l}$ and  $\lab_{1r}\labmeet\lab_{2r}\labless
    \lab'_{1r}\labmeet\lab'_{2r}$
    \\\>(4)\> $\lab'_{1l}\labjoin\lab'_{2l}\labless
    \lab_{1l}\labjoin\lab_{2l}$ and $\lab_{1r}\labjoin\lab_{2r}\labless
    \lab'_{1r}\labjoin\lab'_{2r}$
    \\ By definition of $\bowtie$
    \\\>(5)\>$\iota_1\bowtie \iota_2 = [\lab_{1l}\labjoin\lab_{2l},
    \lab_{1r}\labmeet\lab_{2r}]$, and 
    $\iota'_1\bowtie \iota'_2 = [\lab'_{1l}\labjoin\lab'_{2l},
    \lab'_{1r}\labmeet\lab'_{2r}]$
    \\By (3) and (4),  $\iota_1\bowtie\iota_2 \sqsubseteq \iota'_1\bowtie\iota'_2$
    \\ By definition of $\labjoin$
    \\\>(6)\> $\iota_1\labjoin\iota_2 = [\lab_{1l}\labjoin\lab_{2l},
    \lab_{1r}\labjoin\lab_{2r}]$ and $\iota'_1\labjoin\iota'_2 = [\lab'_{1l}\labjoin\lab'_{2l},
    \lab'_{1r}\labjoin\lab'_{2r}]$ 
    \\ By (3), and (4)
    $\iota_1\labjoin\iota_2\sqsubseteq\iota'_1\labjoin\iota'_2$.
    \\ By definition of $\refineof$
    \\\>(7)\>  $\refineof(\iota_1, \iota_2)
    =([\lab_{1l},\lab_{1r}\labmeet\lab_{2r}], [\lab_{2l}\labjoin\lab_{1l},\lab_{2r}])$, and
    \\\>(8)\> $\refineof(\iota'_1, \iota'_2)
    =([\lab'_{1l},\lab'_{1r}\labmeet\lab'_{2r}],
    [\lab'_{2l}\labjoin\lab'_{1l},\lab'_{2r}])$
    \\ T.S. $[\lab_{1l},\lab_{1r}\labmeet\lab_{2r}]\sqsubseteq
    [\lab'_{1l},\lab'_{1r}\labmeet\lab'_{2r}]$ and 
    $[\lab_{2l}\labjoin\lab_{1l},\lab_{2r}]\sqsubseteq
    [\lab'_{2l}\labjoin\lab'_{1l},\lab'_{2r}]$
    \\By (1) and (3),  $[\lab_{1l},\lab_{1r}\labmeet\lab_{2r}]\sqsubseteq
    [\lab'_{1l},\lab'_{1r}\labmeet\lab'_{2r}]$
    \\By (2) and (4), $[\lab_{2l}\labjoin\lab_{1l},\lab_{2r}]\sqsubseteq
    [\lab'_{2l}\labjoin\lab'_{1l},\lab'_{2r}]$
    \\By assumptions
    \\\>(9)\> $E = (\iota_a, \iota_b)$, $E' = (\iota'_a, \iota'_b)$, and
    $\iota_a\sqsubseteq\iota'_a$, $\iota_b\sqsubseteq\iota'_b$
    \\By $\iota_1\sqsubseteq \iota'_1$ and (9) and we have proven 1-3 of
    this lemma
    \\\>(10)\> $\iota_a\bowtie\iota_1\sqsubseteq \iota'_a\bowtie\iota'_1$
    \\By (10) and we have proven 1-3 of this lemma
    \\\>(11)\> $\refineof(\iota_a\bowtie\iota_1,\iota_b)
    \sqsubseteq\refineof(\iota'_a\bowtie\iota'_1,\iota'_b)$
    \\By definition of $\iota\bowtie E$ and (11), $\iota_1\bowtie
    E\sqsubseteq \iota'_1\bowtie E'$
  \end{tabbing}

  The proof of 5 cases on $\glab'_1$ and $\glab'_2$. When neither one is
  $?$, the conclusion holds because $\sqsubseteq$ is reflexive. When
  one of them is $?$, $\glab'_1\cjoin\glab'_2=?$, it is defined that
  $\glab_1\cjoin\glab_2 \sqsubseteq ?$.
\end{proof}

\subsection{Static Gradual Guarantee}
\label{sec:app-sgg}

We omit the definitions of $e\gsubtp e'$ for $\wg$, which are
inductively defined over the structure of $e$. We
only show  the definition for cast expression below, as it requires
the types to be the same. 
\[
  \inferrule*{
    e \gsubtp e' 
  }{
    e::U \gsubtp e'::U 
  }
\]
The precision of the cast operator is defined as shown above. As the
precision operation is defined over labels of values in stores, the
expressions have to be cast to the same type $U$. If we cast $e'$ to a
different type $U'$ and try to show that $U \gsubtp U'$, the proof
results in cases that require a proof for $\ell_1 \gsubtp \ell_2$
where $\ell_1 \neq \ell_2$, which does not hold. 

\begin{lem}
  \label{lem:g-refine}
  ~ If~ $\glab_1 \sqsubseteq \glab_1'$, $\glab_2 \sqsubseteq \glab_2'$ and 
  $\glab_1 \clabless \glab_2$, 
  then $\glab_1' \clabless \glab_2'$. 
\end{lem}
\begin{proof}
  Assume $\glab_1 = \ell_1$ and $\glab_2 = \ell_2$ (L.H.S of $\sqsubseteq$ is a precise label).
  \\By precision definition, $\glab_1' = \ell_1$ or $\glab_1' = ?$, and $\glab_2' = \ell_2$ or $\glab_2' = ?$.
  \\ If $\glab_1' = \ell_1$ and $\glab_2' = \ell_2$, then the conclusion holds by assumption
  \\ If $\glab_1' = ?$ and $\glab_2' = \ell_2$, then $? \clabless \ell_2$ holds by definition of $\clabless$
  \\ If $\glab_1' = \ell_1$ and $\glab_2' = ?$, then $\ell_1 \clabless ?$ holds by definition of $\clabless$
  \\ If $\glab_1' = ?$ and $\glab_2' = ?$, then $? \clabless ?$ holds by definition of $\clabless$
\end{proof}

\begin{lem}[Static Gradual Guarantee - Expressions]
  \label{lem:static-exp}
  ~ If~ $\Gamma \vdash e : U$,
    $\Gamma\sqsubseteq\Gamma'$, and $e\sqsubseteq e'$, 
    then $\Gamma' \vdash e' : U'$ and $U \sqsubseteq U'$. 
\end{lem}
\begin{proof}
  By induction on the structure of the typing derivation. Follows from assumption for \rulename{Bool, Int, Var}.
  \begin{description}
  \item [Case:] \rulename{Bop}
    \begin{tabbing}
      By IH, $\tau^{\glab_1} \sqsubseteq \tau^{\glab_1'}$ and $\tau^{\glab_2} \sqsubseteq \tau^{\glab_2'}$
      \\By Lemma~\ref{lem:closed-under-refinement}, $\glab_1 \cjoin \glab_2 \sqsubseteq \glab_1' \cjoin \glab_2'$
      \\By \rulename{Bop}, $\tau^{\glab} \sqsubseteq \tau^{\glab'}$
    \end{tabbing}
  \item [Case:] \rulename{Cast}
    \begin{tabbing}
      $
        \inferrule*{
          \Gamma  \vdash e : \tau^{\glab_1} \\
          U = \tau^{\glab} \\
          \glab_1 \clabless \glab
        }{ 
          \Gamma  \vdash e::U : \tau^{\glab}
        }
      $
      \\ T.S. $\Gamma' \vdash e' :: U : \tau^{\glab}$ and $\glab_1' \clabless \glab'$ where $\Gamma' \vdash e' : \tau^{\glab_1'}$
      \\  $e'$ is also cast to $U$ instead of another type $U'$ for reasons mentioned above.
      \\By IH, $\Gamma' \vdash e' : \tau^{\glab_1'}$, $\tau^{\glab_1} \sqsubseteq \tau^{\glab_1'}$. Thus, $\glab_1 \sqsubseteq \glab_1'$.
      \\By assumption, $\glab_1 \clabless \glab$. $\glab_1$ is a precise label (left-value of $\sqsubseteq$ is a precise label).
      \\$\glab_1 = \ell_1$, then $\glab_1' = \ell_1$ or $\glab_1' = ?$. In both cases, $\glab_1' \clabless \glab$.
      \\By \rulename{Cast}, the conclusion holds.
    \end{tabbing}    
  \end{description}
\end{proof}


\begin{thm}[Static Gradual Guarantee]
  ~
  If~ $\Gamma; \glab \vdash c$, $\Gamma\sqsubseteq\Gamma'$, 
  $\glab\sqsubseteq\glab'$, and $c\sqsubseteq c'$, 
  then $\Gamma'; \glab'\vdash c'$. 
\end{thm}
\begin{proof}
  By induction on command typing derivation. Most cases can be proven by using the induction hypothesis and the typing rule. 
  When the derivation ends in \rulename{Assign, Out, If, While}, apply Lemma~\ref{lem:static-exp} on the premises 
  and when the derivation ends in  \rulename{If, While}, we additionally apply Lemma~\ref{lem:closed-under-refinement}.
  We use the same typing rule to reach the conclusion.
\end{proof}

\subsection{Dynamic Gradual Guarantee}
\label{sec:app-dgg}
\begin{lem}[Dynamic Guarantee (Expressions)]
  ~\label{lem:dynamic-g-exp}
  If $\delta_1 \sepidx{} e_1 \evalsto v_1$,
  $\delta_1\sqsubseteq\delta_2$, and $e_1\sqsubseteq e_2$, 
  then $\delta_2 \sepidx{} e_2 \evalsto  v_2$ and
  $v_1\sqsubseteq v_2$. 
\end{lem}
\begin{proof}
  By induction on the structure of the expression evaluation. We
  apply the induction hypothesis directly for.
  The basecases \rulename{M-Const} and
  \rulename{M-Var} can be shown using assumptions directly.
  \begin{description}
  \item[Case:] \rulename{M-Bop}
    \begin{tabbing}
      By assumption the evaluation ends in \rulename{M-Bop} rule:
      \\
      \quad\= (1)\quad\=   
      $\inferrule*[right=M-Bop]{
        \delta\sepidx{} e_1 \evalsto (\iota_1\, u_1)^{\glab_1} \\
        \delta\sepidx{} e_2 \evalsto (\iota_2\, u_2)^{\glab_2} \\\\
        \iota = (\iota_1\labjoin\iota_2) \\ \glab = \glab_1\cjoin\glab_2 \\ u = (u_1\bop u_2) 
      }{
        \delta\sepidx{} e_1\bop e_2 \evalsto   (\iota\;u)^{\glab}
      }$
      \\\>(2)\> $e= e_1\bop e_2$,  $\delta \sqsubseteq\delta'$, and
      $e\sqsubseteq e'$
      \\By inversion of (2)
      \\\>(3)\> $e' = e'_1 \bop e'_2$, $e_1\sqsubseteq e'_1$, and
      $e_2\sqsubseteq e'_2$
      \\By I.H.
      \\\>(4)\> $\delta'\sepidx{} e'_1 \evalsto
      (\iota'_1\,u'_1)^{\glab'_1}$, $\iota_1\sqsubseteq\iota'_1$,
      and $\glab_1\sqsubseteq\glab'_1$,
      \\\>(5)\>$\delta'\sepidx{} e'_2 \evalsto (\iota'_2\,
      u'_2)^{\glab'_2}$, $\iota_2\sqsubseteq\iota'_2$,
      and $\glab_2\sqsubseteq\glab'_2$,
      \\By \rulename{M-Bop}
      \\\>(6)\> $\delta'\sepidx{} e'_1\bop e'_2 \evalsto
      (\iota'\;u)^{\glab'}$, where $\iota'=\iota'_1\labjoin\iota'_2$,
      $\glab' = \glab'_1\cjoin\glab'_2$'
      \\By Lemma~\ref{lem:closed-under-refinement}
      \\\>(7)\> $(\iota\;u)^{\glab}\sqsubseteq (\iota'\;u)^{\glab'}$
    \end{tabbing}

  \item[Case:] \rulename{M-Cast} and \rulename{P-Cast} can be proven
    similarly by applying I.H. and Lemma~\ref{lem:closed-under-refinement}.

  \end{description}
\end{proof}

\begin{lem}~
  \label{lem:refine-value}
  \begin{enumerate}
  \item if  $\iota_1\sqsubseteq \iota_2$ and $v_1\sqsubseteq v_2$, and
    $v_1' = \reflvof{}(\iota_1, v_1)$ and   $v_2' =
    \reflvof{}(\iota_2, v_2)$, then $v_1'\sqsubseteq v_2'$
  \item if  $\iota_1\sqsubseteq \iota_2$ and $v_1\sqsubseteq v_2$, and
    $v_1' = \updval(\iota_1,v_1)$ and   $v_2' =
    \updval(\iota_2, v_2)$, then $v_1'\sqsubseteq v_2'$
  \item if  $\iota_1\sqsubseteq \iota_2$ and $\delta_1\sqsubseteq \delta_2$, and
    $\delta'_1 = \rflof{}(\delta_1, X, \iota_1)$, 
    and $\delta'_2 = \rflof{}(\delta_2, X, \iota_2)$, 
    then $\delta_1'\sqsubseteq \delta_2'$
  \end{enumerate}
\end{lem}
\begin{proof}
  Proofs of (1) and (2) examine the definitions of the operations and
  apply Lemma~\ref{lem:closed-under-refinement} directly. Proof of (3)
  is by induction over the size of $X$ and (1).
\end{proof}

\begin{thm}[Dynamic Gradual Guarantee]
  If~ $\kappa_1, \delta_1 \sepidx{} c_1 \stackrel{{\alpha_1}}{\stepsto}
  \kappa'_1,\delta'_1\sepidx{} c'_1$ and 
  $\kappa_1, \delta_1 \sepidx{} c_1 \sqsubseteq \kappa_2, \delta_2 \sepidx{} c_2$, 
  then $\kappa_2, \delta_2 \sepidx{} c_2 \stackrel{{\alpha_2}}{\stepsto}
  \kappa'_2,\delta'_2\sepidx{} c'_2$ such that
  $\kappa'_1, \delta_1' \sepidx{} c_1' \sqsubseteq\kappa'_2, \delta_2' \sepidx{} c_2'$ 
  and $\alpha_1 = \alpha_2$.
\end{thm}
\begin{proof}
  By induction on the command semantics. Most cases can be proven by
  using the induction hypothesis and
  Lemma~\ref{lem:closed-under-refinement} directly. 

  When the derivation ends in \rulename{M-Assign}, \rulename{M-Out},
  or \rulename{M-If-Refine}, apply Lemma~\ref{lem:dynamic-g-exp} and
  Lemma~\ref{lem:refine-value} on the premises and use the same
  semantic rule to reach the conclusion.
\end{proof}

\end{document}